\documentclass[journal, draftcls, onecolumn, twospace]{IEEEtran}

\usepackage{latexsym}
\usepackage{amsmath}
\usepackage{amsfonts}
\usepackage{rawfonts,amsthm}
\usepackage{amssymb}
\usepackage{float}
\usepackage{paralist}
\usepackage{verbatim}
\usepackage{mathrsfs}
\usepackage{multirow,makecell}
\usepackage{amsfonts}
\usepackage{longtable}
\usepackage{mathrsfs}
\usepackage{ifthen}
\usepackage{listings}
\usepackage{url}
\usepackage{extarrows}
\usepackage{mathrsfs}
\usepackage{amsmath}
\usepackage{graphicx,color}
\usepackage{subfigure}
\usepackage{xcolor}
\usepackage[framemethod=TikZ]{mdframed}
\usepackage{framed}
\usepackage{mathrsfs}
\usepackage{multicol}

\usepackage{hyperref}

\usepackage{calligra}
\usetikzlibrary{arrows}
\usetikzlibrary{decorations.pathreplacing}
\usepackage{threeparttable}
\usepackage{graphicx}
\usepackage{epstopdf}
\usepackage{epsfig}
\usepackage{algpseudocode}
\usepackage{soul} 
\newtheorem{theorem}{Theorem}[section]

\newtheorem{proposition}{Proposition}[section]
\newtheorem{lemma}{Lemma}[section]
\newtheorem{corollary}{Corollary}[section]
\newtheorem{definition}{Definition}[section]
\newenvironment{example}[1][Example]{\begin{trivlist}\item[\hskip\labelsep{\bfseries#1}]}{\end{trivlist}}
\newenvironment{remark}[1][Remark]{\begin{trivlist}\item[\hskip\labelsep{\bfseries#1}]}{\end{trivlist}}

\usepackage{tikz}
\usepackage{amsmath}
\usetikzlibrary{positioning}
\usetikzlibrary{arrows,arrows.meta}
\DeclareMathAlphabet{\mathpzc}{OT1}{pzc}{m}{it}

\usepackage{graphicx}
\graphicspath{{../eps/}}
\DeclareGraphicsExtensions{.eps}
\usepackage{ifthen}
\newboolean{ONECOLUMN}
\setboolean{ONECOLUMN}{true}
\usepackage{amsmath}
\usepackage{amsfonts}
\usepackage{color}
\usepackage{amssymb}

\usepackage{tikz}
\usetikzlibrary{positioning}
\usetikzlibrary{arrows,arrows.meta}
\usetikzlibrary{decorations.pathreplacing,calligraphy}

\usepackage{verbatim} 
\usepackage{array}
\usepackage{fixltx2e}
\usepackage{url}

\usepackage{flushend}

\begin{document}
	
\title{Bridging Hamming Distance Spectrum with Coset Cardinality Spectrum for Overlapped Arithmetic Codes}

\author{Yong~Fang
	\thanks{The author is with the School of Information Engineering, Chang'an University, Xi'an, Shaanxi 710064, China (email: fy@chd.edu.cn).}
}

\markboth{IEEE Transactions on Information Theory (Submission)}{Fang}

\maketitle



\maketitle

\begin{abstract}
	Overlapped arithmetic codes\footnote{In previous works, the terminology for \textit{the scheme making use of arithmetic codes to implement Slepian-Wolf coding by interval overlapping} is in a mess: Sometimes it is called \textit{distributed arithmetic coding}; while in rare cases it is named as \textit{distributed arithmetic codes}. We now recognize that \textit{coding} refers to a problem, \textit{e.g.}, source coding, channel coding, \textit{etc.}, while \textit{code} refers to a practical realization of a coding problem, \textit{e.g.}, turbo codes and polar codes can be used to implement channel coding, while arithmetic codes and Huffman codes can be used to implement source coding. Therefore, this scheme should be named as \textit{code} rather than \textit{coding}. Moreover, there is a similar scheme that makes use of arithmetic codes to implement Slepian-Wolf coding by bit puncturing. To strictly distinguish these two schemes, they will be formally named as \textit{overlapped arithmetic codes} and \textit{punctured arithmetic codes}, respectively, from now on.}, featured by overlapped intervals, are a variant of arithmetic codes that can be used to implement Slepian-Wolf coding. To analyze overlapped arithmetic codes, we have proposed two theoretical tools: \textit{Coset Cardinality Spectrum} (CCS) and \textit{Hamming Distance Spectrum} (HDS). The former describes how source space is partitioned into cosets (equally or unequally), and the latter describes how codewords are structured within each coset (densely or sparsely). However, until now, these two tools are almost parallel to each other, and it seems that there is no intersection between them. The main contribution of this paper is bridging HDS with CCS through a rigorous mathematical proof. Specifically, HDS can be quickly and accurately calculated with CCS in some cases. All theoretical analyses are perfectly verified by simulation results.
\end{abstract}


\IEEEpeerreviewmaketitle

\newpage
\section{Introduction}\label{sec:intro}
\IEEEPARstart{L}{et} $(X,Y)\sim p_{X,Y}(x,y)$ be a pair of correlated discrete random variables, where $x\in{\cal X}$ and $y\in{\cal Y}$. Let $(X,Y)^n\triangleq((X_1,Y_1),\dots,(X_n,Y_n))$ be $n$ independent copies of $(X,Y)$. Let $R_X$ ($R_Y$, resp.) be the achievable per-symbol rate of compressing $X^n$ ($Y^n$, resp.) without loss. Slepian and Wolf proved in \cite{SlepianIT73} that, given $|{\cal X}|<\infty$ and $|{\cal Y}|<\infty$, where $|\cdot|$ denotes the cardinality of a set, if only the bitstreams of $X$ and $Y$ are jointly parsed at the decoder, then $R_X\geq H(X|Y)$, $R_Y\geq H(Y|X)$, and $R_X+R_Y\geq H(X,Y)$ as $n\to\infty$, even though there is no communication between $X$ and $Y$. Shortly afterwards, in \cite{Cover}, Cover simplified Slepian and Wolf's proof and generalized the results to arbitrary ergodic processes with countably infinite alphabets. For brevity, lossless distributed source coding is often called \textit{Slepian-Wolf coding}. 

\subsection{Implementations of Slepian-Wolf Coding}
In the invited overview paper \cite{Wyner}, by making use of correlated binary sources as an example, Wyner revealed that Slepian-Wolf coding can be implemented by linear parity-check codes (see sub-Sect.~VI.C of \cite{Wyner}). However, it was not until \cite{ChenTCOM09,ChenTIT09b} that the duality between Slepian-Wolf coding and channel coding was strictly proved. In \cite{Pradhan}, Pradhan and Ramchandran proposed the famous \textit{DIstributed Source Coding Using Syndromes} (DISCUS) scheme which makes use of the syndrome of linear coset codes to implement Slepian-Wolf coding. Inspired by the DISCUS scheme, many practical implementations of Slepian-Wolf coding based on linear channel codes appeared, \textit{e.g.}, Turbo codes \cite{swcturbo}, \textit{Low-Density Parity-Check} (LDPC) codes \cite{swcldpc}, and polar codes \cite{swcpolar}, \textit{etc.} 

Recently, there are some important and interesting findings. It was proved that, to realize Slepian-Wolf coding, \textit{nonlinear} codes are strictly better than \textit{linear} codes \cite{ChenTIT09a}, and \textit{variable}-rate codes are strictly better than \textit{fixed}-rate codes \cite{ChenEntropy}. Unfortunately, most contemporary channel codes are \textit{linear} and \textit{fixed}-rate. Hence, there are also some attempts of implementing Slepian-Wolf coding with source codes, which are usually \textit{nonlinear} and \textit{variable}-rate, \textit{e.g.}, overlapped arithmetic codes \cite{GrangettoTSP09}, overlapped quasi-arithmetic codes \cite{Artigas}, and punctured quasi-arithmetic codes \cite{Malinowski}, \textit{etc.} However, needless to say, for memoryless sources or sources with memoryless correlation, the source-code-based implementations of Slepian-Wolf coding in \cite{GrangettoTSP09,Artigas,Malinowski} have not yet exhibited better performance as predicted by the theoretical analyses in \cite{ChenTIT09a,ChenEntropy} (see a comprehensive comparison between overlapped arithmetic codes with LDPC codes and polar codes in \cite{FangCL23}). As for non-memoryless sources or sources with non-memoryless correlation, the source-code-based implementations of Slepian-Wolf coding may perform better than those channel-code-based approaches. For example, for memoryless sources with hidden-Markov correlation, overlapped arithmetic codes \cite{FangHMM} are superior to LDPC codes \cite{LDPCHMM}.

Actually, all above realizations of Slepian-Wolf coding were originally designed for binary sources, while most of real-world sources, \textit{e.g.}, images, videos, \textit{etc.}, are nonbinary. To implement nonbinary Slepian-Wolf coding, a brute-force way is to decompose every nonbinary source into multiple bitplanes and then compress every bitplane with a binary Slepian-Wolf code. Putting compression performance aside, this solution is infeasible because it is very hard to jointly optimize bit allocation between multiple bitplanes. Thus, a better way is compressing nonbinary sources with nonbinary Slepian-Wolf codes directly. In \cite{AlgebraicSWC}, Reed-Solomon codes are used to implement nonbinary Slepian-Wolf coding; while in \cite{NBDAC}, overlapped arithmetic codes are extended to achieve the same goal. However, both \cite{AlgebraicSWC} and \cite{NBDAC} model the correlation between correlated nonbinary sources as a symmetric channel, which deviates from reality in many cases. For example, the difference between two adjacent video frames can usually be modeled as a Laplacian process. For this reason, nonbinary overlapped arithmetic codes are formally proposed in \cite{FangTIT23}, which significantly outperform nonbinary LDPC codes.

\subsection{Motivations of This Paper}
This paper focuses on overlapped arithmetic codes for uniform binary sources. Under this setup, one interesting thing is that overlapped arithmetic codes can be taken as a kind of \textit{nonlinear} coset codes. As we know, \textit{minimum distance} is a vital intrinsic attribute of coset codes, because it can be used to derive block error rate and bit error rate of coset codes. However, compared with those \textit{linear} coset codes, \textit{e.g.}, LDPC codes or polar codes, there is one more intrinsic attribute for overlapped arithmetic codes due to its \textit{nonlinearity}. That is, \textit{linear} coset codes usually partition source space into cosets \textit{equally}, while \textit{nonlinear} coset codes partition source space into cosets \textit{unequally}. Hence, the additional problem for overlapped arithmetic codes is: How coset cardinality is distributed? By intuition, both the classic problem of \textit{minimum distance} and the additional problem of \textit{coset cardinality distribution} are very difficult, and according to our experiences, the former is even much knottier than the latter.

The additional problem cased by the nonlinearity has been well solved. For overlapped arithmetic codes, the concept of \textit{Coset Cardinality Spectrum} (CCS) was defined and a recursive formula was derived, which can be numerically implemented to obtain CCS \cite{FangSPL09,FangTC13,FangTCOM16b}. With the help of CCS, the performance of overlapped arithmetic codes can be improved if decoder complexity is limited \cite{FangTC14}. The work on CCS was extended to nonuniform binary sources in \cite{FangTIT20,FangTIT21} and to uniform nonbinary sources in \cite{FangTIT23}.

As for the classic problem of \textit{minimum distance}, we consider its generalized form: How Hamming distance is distributed? In parallel with CCS, we also defined \textit{Hamming Distance Spectrum} (HDS) for overlapped arithmetic codes, which is a function \textit{with respect to} (w.r.t.) Hamming distance $d$, parameterized by code length $n$. To analyze HDS, we developed a tool named \textit{coexisting interval} in \cite{FangTCOM16a}, which was also exploited later in \cite{FangCL21} to calculate the block error rate of overlapped arithmetic codes. With the help of \textit{coexisting interval}, we successfully obtained an approximate formula of HDS in \cite{FangTCOM16a}. In this paper, we refer to the approximate formula obtained in \cite{FangTCOM16a} as \textit{Soft Approximation}, to distinguish it from its variant obtained in \cite{FangTCOM16b}, which is referred to as {\em Hard Approximation}.

Our work on CCS has been very comprehensive and solid. By contrast, our work on HDS is quite thin and weak. We summarize the weakness of related work as below:
\begin{itemize}
	\item The approximate formulas of HDS obtained in \cite{FangTCOM16a,FangTCOM16b} lack strict proofs;
	\item Both approximate formulas have the same complexity ${\cal O}(2^d\binom{n}{d})$, ascending hyper-exponentially as $d$ increases, unacceptable for $d\gg 1$;
	\item Up to now, it seems that there is hardly any intersection between CCS and HDS. Though a binomial approximate formula of HDS was given based on CCS in \cite{FangTCOM16a}, such connection is arbitrary and farfetched. This is a very trivial attempt to bridge HDS with CCS.
\end{itemize}
The above shortcomings of previous work on HDS motivate this paper. 

\subsection{Contributions of This Paper}
The novelties of this paper include the following aspects:
\begin{itemize}
	\item For the approximate formulas of HDS in \cite{FangTCOM16a,FangTCOM16b}, we give strict proofs (cf. sub-Sect.~\ref{subsec:hdscal} and Sect.~\ref{sec:hdsn});
	\item Most importantly, we derive a new approximate formula of HDS with the help of CCS through a rigorous proof (cf. Sect.~\ref{sec:hdsfast}), which is referred to as {\em Fast Approximation}. For $d\approx n$, it dramatically reduces the complexity from ${\cal O}(2^d\binom{n}{d})$ to ${\cal O}(1)$. This formula beautifully and concisely bridges HDS with CCS.
	\item We prove the necessary and sufficient condition for the convergence of HDS when $d=3$ (cf. sub-Sect.~\ref{subsec:converge}), and derive the closed form of HDS in two divergent cases (cf. sub-Sect.~\ref{subsec:psi3}).
\end{itemize}
\hfill\\
\indent
The rest of this paper is arranged as bellow. Some necessary background knowledge is reviewed in Sect.~\ref{sec:mod1} and Sect.~\ref{sec:ccs}. In Sect.~\ref{sec:mod1}, we define the concept of {\em sequences uniformly distributed modulo 1} and introduce some related properties; while in Sect.~\ref{sec:ccs}, we summarize our previous work on CCS. In Sect.~\ref{sec:coexist}, we define the concept of {\em coexisting interval} and prove some important properties, which will be used in the following sections. In Sect.~\ref{sec:hds1}, we define the concept of HDS, give a strict proof of its {\em Soft Approximation}, and discuss its convergence in detail. In Sect.~\ref{sec:hdsn}, we derive the {\em Hard Approximation}, which is a variant of the {\em Soft Approximation}. In Sect.~\ref{sec:hdsfast}, we make use of CCS to derive the {\em Fast Approximation} of HDS, which is of low complexity for $d\approx n$. Then Sect.~\ref{sec:example} gives some experimental results and finally Sect.~\ref{sec:conclusion} concludes this paper.

\section{Overview on Sequences Uniformly Distributed Modulo 1}\label{sec:mod1}
Our developed system on CCS \cite{FangTIT20,FangTIT21,FangTIT23} is laid on an important theorem. This section will give a stricter proof of this theorem and extend it to a more general case. Before doing so, let us define two concepts.

\begin{definition}[Counting Function]
	Let $\omega=(a_1,a_2,\dots)$ be a sequence of real numbers. Let $\{a_i\}\triangleq a_i-\lfloor{a_i}\rfloor \in [0,1)$ denote the fractional part of $a_i\in\mathbb{R}$. For ${\cal I}\subseteq[0,1)$, the counting function is defined as
	\begin{align}
		C({\cal I};n;\omega) \triangleq \left|\left\{a_i: 1\leq i\leq n~{\rm and}~\{a_i\}\in{\cal I}\right\}\right|.
	\end{align}
\end{definition}

\begin{definition}[mod-1 u.d. Sequence]\label{def:udm1}
	We say that the sequence $\omega=(a_1,a_2,\dots)$ is {\em uniformly distributed modulo 1} (u.d. mod 1) if for any $[l,h)\subseteq[0,1)$,
	\begin{align}
		\lim_{n\to\infty}\frac{C([l,h);n;\omega)}{n} = h-l.
	\end{align}
\end{definition}

Based on the above definitions, we can now introduce the following important theorem, which lays the theoretical foundation for our work on CCS \cite{FangTIT20,FangTIT21,FangTIT23}.

\begin{theorem}[mod-1 Weighted Sum of mod-1 u.d. Sequence]
	\label{thm:uniform}
	Let $(a_1,a_2,\dots)$ be a sequence of real numbers u.d. mod 1. Let $(X_1,X_2,\dots)$ be a sequence of {\em independent and identically-distributed} (i.i.d.) discrete random variables. Let $E\triangleq\{a_1X_1+a_2X_2+\cdots\}$, where $\{\cdot\}$ denotes the fractional part of a real number. Then 
	\begin{itemize}
		\item $E$ is uniformly distributed over $[0,1)$; and
		\item For any $\mathbb{I} = \{i_1,i_2,\dots\} \subset \mathbb{N}$ whose complement is an infinite set, {\em i.e.}, $|\mathbb{N}\setminus\mathbb{I}| = \infty$, where $\mathbb{N}$ denotes the set of natural numbers and $|\cdot|$ denotes the cardinality of a set, $E$ is independent of $X_\mathbb{I}=(X_{i_1},X_{i_2},\dots)$ or $I(E;X_{\mathbb{I}})=0$ in other words. 
	\end{itemize}
\end{theorem}

The above theorem is just {Lemma~III.1} of \cite{FangTIT23}, which is actually a generalized form of {Lemma~II.3} of \cite{FangTIT21}. The difference between {Lemma~III.1} of \cite{FangTIT23} and {Lemma~II.3} of \cite{FangTIT21} lies in two aspects:
\begin{itemize}
	\item In {Lemma~II.3} of \cite{FangTIT21}, $X_i$ is a binary random variable, while in {Lemma~III.1} of \cite{FangTIT23}, $X_i$ can be an arbitrary discrete (may not be binary) random variable; and
	\item In {Lemma~II.3} of \cite{FangTIT21}, the set $\mathbb{I}$ is a finite set, \textit{i.e.}, $|\mathbb{I}|<\infty$, while in {Lemma~III.1} of \cite{FangTIT23}, the set $\mathbb{I}$ may be an infinite set, \textit{i.e.}, $|\mathbb{I}|=\infty$, if only $|\mathbb{N}\setminus\mathbb{I}| = \infty$.  
\end{itemize}
Let us review the proof of Theorem~\ref{thm:uniform} in \cite{FangTIT21,FangTIT23} again, which is divided into two folds:
\begin{itemize}
	\item The first fold is to prove that $E$ is uniformly distributed over $[0,1)$, which holds obviously due to the premise of this lemma, \textit{i.e.}, $(a_1,a_2,\dots)$ is a mod-1 u.d. sequence and $(X_1,X_2,\dots)$ is an i.i.d. discrete random sequence.	
	\item The second fold is to prove that $E$ is independent of $X_{\mathbb{I}}$. We divide the sequence $X_\mathbb{N}\triangleq(X_1,X_2,\dots)$ into two sub-sequences $X_{\mathbb{I}}$ and $X_{\mathbb{N}\setminus\mathbb{I}}$. Let $E_{\mathbb{N}\setminus\mathbb{I}} \triangleq \{\sum_{i\in\mathbb{N}\setminus\mathbb{I}}{a_iX_i}\}$. Then $E=\{\sum_{i\in\mathbb{I}}{a_iX_i}+E_{\mathbb{N}\setminus\mathbb{I}}\}$. According to the first fold, if $|\mathbb{N}\setminus\mathbb{I}|=\infty$, then $E_{\mathbb{N}\setminus\mathbb{I}}$ is u.d. over $[0,1)$. In turn, since $X_\mathbb{N}$ is an i.i.d. sequence, \cite{FangTIT21,FangTIT23} {\em directly} draw the conclusion that $E$ is independent of $X_\mathbb{I}$, even if $|\mathbb{I}|=\infty$.
\end{itemize}

However, there is a big jump in the final step of the second fold: Why will $E=\{\sum_{i\in\mathbb{I}}{a_iX_i}+E_{\mathbb{N}\setminus\mathbb{I}}\}$ lead to $I(E;X_\mathbb{I})=0$, given that $E_{\mathbb{N}\setminus\mathbb{I}}$ is u.d. over $[0,1)$? Below we give a lemma to bridge this gap. 	

\begin{lemma}[Virtual Continuous mod-1 Channel]\label{lem:contmod1}
	Let $X$ and $Z$ be two independent continuous random variables. If $Z$ is u.d. over $[0,1)$, then $Y \triangleq \{X+Z\} = (X+Z)-\lfloor{X+Z}\rfloor$ is also u.d. over $[0,1)$ and independent of $X$, no matter how $X$ is distributed.
\end{lemma}
\begin{proof}
	Since $X=\lfloor X\rfloor+\{X\}$, we have $\lfloor \lfloor X\rfloor+\{X\}+Z \rfloor = \lfloor X\rfloor + \lfloor \{X\}+Z \rfloor$, which is followed by
	\begin{align}
		Y = (\lfloor X\rfloor+\{X\}+Z) - \lfloor X\rfloor - \lfloor \{X\}+Z \rfloor = (\{X\}+Z) - \lfloor \{X\}+Z \rfloor.\nonumber
	\end{align}
	Hence we can assume $\lfloor X\rfloor=0$ and $X=\{X\}\in[0,1)$ for simplicity. According to the definition of $Y$,
	\begin{align}
		Y = \begin{cases}
			X+Z, 	&0\leq X<1-Z\\
			X+Z-1, 	&1-Z\leq X<1
		\end{cases}.\nonumber
	\end{align}
	Let $f_{X,Z}(x,z)$ be the joint {\em probability density function} (pdf) of $X$ and $Z$. Since $X$ and $Z$ are mutually independent and $f_Z(z)\equiv 1$ for all $z\in[0,1)$, we have $f_{X,Z}(x,z)=f_X(x)f_Z(z)=f_X(x)$. Thus, the pdf of $Y$ is
	\begin{align}
		f_Y(y) 
		&= \int_{0}^{1-z}{f_{X,Z}(x,y-x)dx} + \int_{1-z}^{1}f_{X,Z}(x,y+1-x)dx\nonumber\\
		&= \int_{0}^{1-z}{f_X(x)dx} + \int_{1-z}^{1}f_X(x)dx = \int_{0}^{1}{f_X(x)dx} = 1,\nonumber
	\end{align}
	implying that $Y$ is u.d. over $[0,1)$. According to the definition of $Y$, we have $f_{Y|X}(y|x)= f_Z(z)=1$. Thus 
	\begin{align}
		f_{X,Y}(x,y) = f_X(x)f_{Y|X}(y|x) = f_X(x) = f_X(x)f_Y(y),\nonumber
	\end{align}
	showing that $X$ and $Y$ are mutually independent.
\end{proof}

To better understand {Lemma~\ref{lem:contmod1}}, one can imagine that $Y=\{X+Z\}$ is a virtual mod-1 channel with $X$ as the input, $Y$ as the output, and $Z$ as the additive noise. If $Z$ is u.d. over $[0,1)$, the input $X$ will be fully buried by the additive noise $Z$, so the capacity of this virtual channel is $C=I(X;Y)=0$, and the output $Y$ is always u.d. over $[0,1)$ and independent of the input $X$. Now let us return to the proof of {Theorem~\ref{thm:uniform}}. Given $E=\{\sum_{i\in\mathbb{I}}{a_iX_i}+E_{\mathbb{N}\setminus\mathbb{I}}\}$, if $E_{\mathbb{N}\setminus\mathbb{I}}$ is u.d. over $[0,1)$ and independent of $X_{\mathbb{I}}$, then according to {Lemma~\ref{lem:contmod1}}, $E$ will be u.d. over $[0,1)$ and independent of $X_{\mathbb{I}}$. Therefore, the loophole in the proof of {Theorem~\ref{thm:uniform}} is filled.

Following is the discrete version of {Lemma~\ref{lem:contmod1}}.

\begin{lemma}[Virtual Discrete mod-1 Channel]\label{lem:discmod1}
	Let $X$ and $Z$ be two independent discrete random variables. Assume that $X$ is defined over ${\cal X} = \{\varepsilon+i/S: i\in\mathbb{Z}\}$, where $S\in\mathbb{Z}$ and $0\leq\varepsilon<1/S$, while $Z$ is u.d. over ${\cal Z}=\{i/S: i\in[0:S)\}$, where $[0:S)\triangleq\{0,1,\dots,S-1\}$. Then $Y \triangleq \{X+Z\} =  (X+Z)-\lfloor{X+Z}\rfloor$ is u.d. over ${\cal Y} = \varepsilon+{\cal Z} \triangleq \{\varepsilon+i/S: i\in[0:S)\}$ and independent of $X$, no matter how $X$ is distributed over ${\cal X}$.
\end{lemma}

\begin{proof}
	Obviously, $X=\lfloor X\rfloor +\{X\}$, where $\{X\}$ is defined over $\varepsilon+{\cal Z}$, so we assume ${\cal X}=\varepsilon+{\cal Z}$ for simplicity. If $X=\varepsilon+i/S$ and $Z=j/S$, where $i,j\in[0:S)$, then $Y = \varepsilon+ \frac{\bmod(i+j,S)}{S} \in \varepsilon+{\cal Z}$. Hence, $p_{Y|X}(y|x) = p_Z(z) = 1/S$. Further, 
	\begin{align}
		p_Y(y) = \sum_{x\in{\cal X}}{p_{X,Y}(x,y)} = \sum_{x\in{\cal X}}{p_X(x)\cdot p_{Y|X}(y|x)} = (1/S)\cdot\sum_{x\in{\cal X}}{p_X(x)} = 1/S,\nonumber
	\end{align}
	\textit{i.e.}, $Y$ is u.d. over ${\cal Y}={\cal X}$. Meanwhile, we have
	\begin{align}
		p_{X,Y}(x,y) = p_X(x)\cdot p_{Y|X}(y|x) = p_X(x)\cdot (1/S) = p_X(x)\cdot p_Y(y),\nonumber
	\end{align}	
	showing that $X$ and $Y$ are mutually independent.
\end{proof}

In fact, the so-called \textit{virtual discrete mod-1 channel} is a special instance of the \textit{mod-$c$ channel}, which is defined by {Eq.~(7.18)} in Cover's canonical textbook \cite{covertextbook}. In turn, the mod-$c$ channel falls into the class of \textit{symmetric channels}. Let ${\bf P} = [p(y|x)]_{|{\cal X}|\times|{\cal Y}|}$ be the transition probability matrix of a symmetric channel, whose rows and columns are indexed by the input $x$ and the output $y$, respectively. Thus all rows of ${\bf P}$ are permutations of each other and all columns of ${\bf P}$ are permutations of each other. According to {Theorem~7.2.1} of \cite{covertextbook}, the capacity of such a symmetric channel is $C=\log{|{\cal Y}|}-H({\bf r})$, where ${\bf r}$ is a row of ${\bf P}$ and $H(\cdot)$ is the entropy function. If the channel output is u.d., then $H({\bf r})=\log{|{\cal Y}|}$, which is followed by $C=I(X;Y)=0$.

As mentioned in \cite{FangTIT21}, {Theorem~\ref{thm:uniform}} requires that $(a_1,a_2,\dots)$ must be a mod-1 u.d. sequence. This is a too strong premise, and in experiments, similar properties are also found for $E$ even though $(a_1,a_2,\dots)$ is not a mod-1 u.d. sequence. Hence for {Theorem~\ref{thm:uniform}}, the mod-1 u.d. requirement imposed on $(a_1,a_2,\dots)$ can actually be relaxed. To support this relaxation, we first define the following condition, which is actually the premise of {Theorem~2.1} in page 20 of Wilms' canonical textbook \cite{Wilms}.
\begin{definition}[Wilms' condition]
	Let $(A_1,A_2,\dots)$ be a sequence of i.i.d. non-degenerate random variables. If there are no constants $\tau\in\mathbb{N}$ and $\xi\in[0,1/r)$, where $r\in\mathbb{Z}$, such that $\{A_i\}=A_i-\lfloor A_i\rfloor$ has its distribution concentrated on the set $\{\xi+j/r: j\in[0:r)\}$, we say that $(A_1,A_2,\dots)$ satisfies Wilms' condition. 
\end{definition}
We rewrite {Theorem~2.1} in page 20 of Wilms' canonical textbook \cite{Wilms} as below.
\begin{theorem}[Generalized mod-1 Uniform Distribution]\label{thm:genud}
Let $(A_1,A_2,\dots)$ be a sequence of random variables satisfying Wilms' condition. Then $S_n\triangleq\sum_{i=1}^{n}A_i-\lfloor\sum_{i=1}^{n}A_i\rfloor$ will be u.d. over $[0,1)$ as $n\to\infty$.
\end{theorem}
According to Theorem~\ref{thm:genud}, we can easily get the generalized form of Theorem~\ref{thm:uniform} as below.
\begin{theorem}[mod-1 Weighted Sum of Generalized Sequence]\label{thm:genuni}
	Let $(A_1,A_2,\dots)$ be a sequence of random variables satisfying Wilms' condition. Let $(a_1,a_2,\dots)$ be a realization of $(A_1,A_2,\dots)$. Let $(X_1,X_2,\dots)$ be a sequence of i.i.d. discrete random variables. Define $E=\{a_1X_1+a_2X_2+\cdots\}$, where $\{\cdot\}$ denotes the fractional part of a real number. Then $E$ will be u.d. over $[0,1)$ as $n\to\infty$. In addition, for any $\mathbb{I} = \{i_1,i_2,\dots\} \subset \mathbb{N}$ whose complement is an infinite set, {\em i.e.}, $|\mathbb{N}\setminus\mathbb{I}| = \infty$, where $|\cdot|$ denotes the cardinality of a set, $E$ is independent of $X_\mathbb{I}=(X_{i_1},X_{i_2},\dots)$ or $I(E;X_{\mathbb{I}})=0$ in other words.
\end{theorem}

\section{Overview on Coset Cardinality Spectrum}\label{sec:ccs}
As pointed out in the introduction, we are investigating a very complex problem, so it is necessary to begin with the simplest case. Throughout this paper, only uniform binary sources are studied, and all source symbols are mapped onto overlapped intervals in the same manner. This scheme is named as \textit{tailless overlapped arithmetic codes} in \cite{FangTCOM16b}. Let $[l,h)\subset\mathbb{R}$ be a half-open interval. For $a,b\in\mathbb{R}$, we define 
\begin{align}
	a\cdot[l,h) + b \triangleq [al+b,ah+b) \subset \mathbb{R}.
\end{align}
Let $X$ be a uniform binary random variable with bias probability $p\triangleq\Pr(X=1)=1/2$. Let $X^n\triangleq(X_1,\dots,X_n)$ be $n$ independent copies of $X$. The encoder of overlapped arithmetic codes recursively shrinks the initial interval $[0,1)$ according to $X_i$. Let ${\cal I}(X^i)\triangleq[l(X^i),h(X^i))$ be the mapping interval of $X^i$, and initially, ${\cal I}(X^0)={\cal I}(\emptyset)=[0,1)$. Let $|{\cal I}(X^i)|=h(X^i)-l(X^i)$ denote the length of ${\cal I}(X^i)$. For a rate-$r$, where $r\in[0,1]$, overlapped arithmetic code, the update rule for ${\cal I}(X^i)$ is:
\begin{itemize}
	\item If $X_i=0$, then ${\cal I}(X^i) = |{\cal I}(X^{i-1})|\cdot [0,2^{-r}) + l(X^{i-1})$; and
	\item If $X_i=1$, then ${\cal I}(X^i) = |{\cal I}(X^{i-1})|\cdot [(1-2^{-r}),1) + l(X^{i-1})$. 
\end{itemize}
It is easy to know $\frac{|{\cal I}(X^i)|}{|{\cal I}(X^{i-1})|}=2^{-r}$ and thus $|{\cal I}(X^i)|=2^{-ir}$. Since $h(X^i)-l(X^i)\equiv 2^{-ir}$, it is enough to trace one of $h(X^i)$ and $l(X^i)$. It is more convenient to trace $l(X^i)$, which is updated by
\begin{align}
	l(X^i) 
	&= l(X^{i-1}) + X_i\cdot(1-2^{-r})\cdot|{\cal I}(X^{i-1})|\nonumber\\
	&= l(X^{i-1}) + X_i\cdot(1-2^{-r})\cdot2^{-(i-1)r}\label{eq:lXirecusion}\\
	&= (1-2^{-r})\cdot\sum_{i'=1}^{i}{X_{i'}\cdot2^{-(i'-1)r}} = (2^r-1)\cdot\sum_{i'=1}^{i}{X_{i'}2^{-i'r}}.\label{eq:lXi}
\end{align}
The length of the final interval is $|{\cal I}(X^n)|\equiv2^{-nr}$. For simplicity, we assume $nr\in\mathbb{Z}$ below. To determine the output bitstream, the final interval ${\cal I}(X^n)$ is enlarged by $2^{nr}$ times to obtain a normalized interval with unit length. That is $2^{nr}{\cal I}(X^n)=[0,1)+\ell(X^n)$, where
\begin{align}\label{eq:ell}
	\ell(X^n) \triangleq 2^{nr}l(X^n) = (2^r-1)\cdot\sum_{i=1}^{n}{X_i\cdot2^{(n-i)r}}.
\end{align}
Since $|2^{nr}{\cal I}(X^n)|\equiv1$, there is one and only one integer in $2^{nr}{\cal I}(X^n)$, which is
\begin{align}
	m(X^n) \triangleq \lceil \ell(X^n) \rceil.
\end{align}
Obviously, $m(X^n)\in[0:2^{nr})$, so it can be represented by $nr$ bits to form the bitstream of $X^n$.
\begin{definition}[Bitstream Projection]
	The projection of $m(X^n)$ onto ${\cal I}(X^i)$ is defined as
	\begin{align}\label{eq:Uin}
		U_{i,n} \triangleq \frac{2^{-nr}m(X^n)-l(X^i)}{h(X^i)-l(X^i)} = 2^{ir}\left(\underbrace{2^{-nr}m(X^n)}_{U_{0,n}}-l(X^i)\right).
	\end{align}
\end{definition}

For conciseness, $U_{i,n}$ can be abbreviated to $U_i$ without causing ambiguity. It is easy to know that $U_i$ is defined over $[0,1)$. Especially, $U_0=2^{-nr}m(X^n)$ is called the initial projection, and $U_n=m(X^n)-2^{nr}l(X^n)=m(X^n)-\ell(X^n)$ is called the final projection. From \eqref{eq:Uin}, we have $U_i=2^{ir}\left(U_0-l(X^i)\right)$, and conversely, 
\begin{align}\label{eq:U=}
	U_0 = 2^{-ir}U_i + l(X^i) 
	&= 2^{-(i+1)r}U_{i+1} + l(X^{i+1})\nonumber\\
	&\stackrel{(a)}{=} 2^{-(i+1)r}U_{i+1} + l(X^i) + X_{i+1}\cdot(1-2^{-r})\cdot2^{-ir},
\end{align}
where $(a)$ comes from \eqref{eq:lXirecusion}. After removing $2^{-ir}$ at both sides, \eqref{eq:U=} can be rewritten as
\begin{align}
	U_i = 2^{-r}U_{i+1} + X_{i+1}\cdot(1-2^{-r}).
\end{align}
Therefore, if the receiver knows $X^n$ through an oracle, then $m(X^n)$ can be decoded along the path $X^n$, and we will obtain the sequence $U_0^n\triangleq(U_0,\dots,U_n)$ via a forward recursion
\begin{align}\label{eq:Uiforward}
	U_{i+1} = 2^r\left(U_i - X_{i+1}(1-2^{-r})\right).
\end{align}

\begin{definition}[Coset Cardinality Spectrum]
	The pdf of $U_{i,n}$, denoted as $f_{i,n}(u)$, for $u\in[0,1)$, is called the level-$i$ {\em Coset Cardinality Spectrum} (CCS). The conditional
	pdf of $U_{i,n}$ given $X_j=x$ is called the conditional level-$i$ CCS and denoted by $f_{i|j,n}(u|x)$.
\end{definition}

For conciseness, $f_{i,n}(u)$ can be abbreviated to $f_i(u)$, and $f_{i|j,n}(u|x)$ can be abbreviated to $f_{i|j}(u|x)$,  without causing ambiguity. Especially, $f_0(u)$ is called the initial CCS, and $f_n(u)$ is called the final CCS. To deduce $f_i(u)$, we should begin with $f_n(u)$ and then go back to $f_0(u)$ recursively \cite{FangTCOM16b}.

\begin{theorem}[Properties of CCS]
	Let $a_i \triangleq \{2^{ir}\} = 2^{ir}-\lfloor 2^{ir}\rfloor$. If the sequence $(a_1,a_2,\dots)$ satisfies Wilms' condition, then as $n\to\infty$,
	\begin{itemize}
		\item $f_{n,n}(u)$ will converge to a uniform function over $[0,1)$;
		\item the sequence $(U_{0,n},\dots,U_{n,n})$ will form a Markov chain; and
		\item $f_{i,n}(u)$ can be deduced via a backward recursion
		\begin{align}\label{eq:ccs}
			f_{i,n}(u) = 2^{r-1}\left(f_{i+1,n}(u2^r) + f_{i+1,n}\left(\left(u-(1-2^{-r})\right)2^{r}\right)\right).
		\end{align}			
	\end{itemize}
\end{theorem}
\begin{proof}
	It was proved in {sub-Sect.~VI.B} of \cite{FangTCOM16b} that if $(a_1,a_2,\dots)$ is u.d. over $[0,1)$, then $f_{n,n}(u)$ will be uniform over $[0,1)$ as $n\to\infty$. This theorem relaxes the mod-1 u.d. constraint on $(a_1,a_2,\dots)$. According to {Theorem~\ref{thm:genuni}}, $f_{n,n}(u)$ will be uniform over $[0,1)$ as $n\to\infty$, if only $(a_1,a_2,\dots)$ satisfies Wilms' condition, even though it is not u.d. over $[0,1)$. Similarly, the two other bullet points of this theorem also hold.
\end{proof}

\begin{definition}[Asymptotic Projection and Asymptotic CCS]
	As $n\to\infty$, the final interval ${\cal I}(X^n)$ will converge to a point, and thus according to \eqref{eq:lXi},
	\begin{align}\label{eq:U0infty}
		U_{0,\infty} = l(X^\infty) = h(X^\infty) = (2^r-1)\cdot\sum_{i=1}^{\infty}{X_i2^{-ir}}.
	\end{align}
	We call $U_{0,\infty}$ the Asymptotic Projection. According to \eqref{eq:ccs}, as $(n-i)\to\infty$, both $f_{i,n}(u)$ and $f_{i+1,n}(u)$ will converge to the same function $f(u)$ that satisfies   
	\begin{align}\label{eq:asympccs}
		f(u) = 2^{r-1}\left(f(u2^r) + f\left(\left(u-(1-2^{-r})\right)2^{r}\right)\right).
	\end{align}
	We call $f(u)$ the Asymptotic CCS, which is also the pdf of $U_{0,\infty}$ defined by \eqref{eq:U0infty}.		
\end{definition}

\begin{corollary}[A Trivial Solution to Asymptotic CCS]
	Let us define $l(x^n)$ as \eqref{eq:lXi}. Let $\delta(u)$ denote the Dirac delta function. For uniform binary sources, the asymptotic CCS can be simply calculated by
	\begin{align}\label{eq:futrivial}
		f(u) = \lim_{n\to\infty}2^{-n}\sum_{x^n\in\mathbb{B}^n}{\delta(u-l(x^n))}.
	\end{align}
\end{corollary}
\begin{proof}
	According to \eqref{eq:U0infty}, the definition of $U_{0,\infty}$, we have
	\begin{align}
		f(u) = \lim_{n\to\infty}\sum_{x^n\in\mathbb{B}^n}{\Pr(X^n=x^n)\cdot\delta(u-l(x^n))}.\nonumber
	\end{align}
	For uniform binary sources, $\Pr(X^n=x^n)\equiv 2^{-n}$ and thus \eqref{eq:futrivial} holds.
\end{proof}

The analytical form of $f(u)$ is unknown in general. However, it was proved in \cite{FangTC13} that if $r$ is the inverse of an integer, then the analytical form of $f(u)$ exists. For example, if $r=1/2$, the asymptotic CCS is given by
\begin{align}\label{eq:closedForm_halfRate}
	f(u) = 
	\begin{cases}
		\frac{u}{3\sqrt{2}-4}, 	&0 \leq u < \sqrt{2}-1\\
		\frac{1}{2-\sqrt{2}}, 	&\sqrt{2}-1 \leq u < 2-\sqrt{2}\\
		\frac{1-u}{3\sqrt{2}-4},&2-\sqrt{2} \leq u < 1%
	\end{cases}.
\end{align}
This is a classic example of $f(u)$. Except those special cases, \eqref{eq:ccs} should be numerically implemented to calculate $f(u)$. A primitive numerical algorithm was proposed in \cite{FangTC13} to implement the backward recursion \eqref{eq:ccs} for uniform binary sources, and then it was generalized to nonuniform binary sources in \cite{FangTIT20}. However, the numerical algorithm proposed in \cite{FangTC13,FangTIT20} is theoretically imperfect. For this reason, after a strict theoretical analysis, \cite{Fang2023FairNA} proposed a novel numerical algorithm, which perfectly overcomes the drawbacks of \cite{FangTC13,FangTIT20}.

Regarding the physical meaning of CCS, \cite{FangTC14} revealed that the encoder of overlapped arithmetic codes is actually a many-to-one nonlinear mapping $\mathbb{B}^n\rightarrow[0:2^{nr})$ that unequally partitions source space $\mathbb{B}^n$ into $2^{nr}$ cosets. The $m$-th coset, where $m\in[0:2^{nr})$, contains roughly $f(m2^{-nr})\cdot2^{n(1-r)}$ codewords. In the asymptotic sense, 
\begin{align}\label{eq:fu_physical}
	\lim_{n\to\infty}\frac{\left|{\cal C}_{\lfloor{u2^{nr}}\rfloor}\right|}{2^{n(1-r)}} = f(u),
\end{align}
where $|\cdot|$ denotes the cardinality of a set and $\lfloor\cdot\rfloor$ denotes the flooring function. Obviously, the $0$-th coset includes only one codeword $0^n$, so we will ignore ${\cal C}_0$ in the following analysis.

\section{Coexisting Interval}\label{sec:coexist}
The \textit{Coexisting Interval}, formally defined in \cite{FangTCOM16b}, is a concept originated from the {\em Risky Interval} defined in \cite{FangTCOM16a}. It is a powerful analysis tool for overlapped arithmetic codes and has found two applications:
\begin{itemize}
	\item In \cite{FangTCOM16a,FangTCOM16b}, the coexisting interval was used to derive the \textit{Hamming Distance Spectrum} (HDS); and
	\item In \cite{FangCL21}, the coexisting interval was exploited to deduce the block error rate.
\end{itemize}
Due to its extreme importance, we dedicate this whole section to coexisting interval, which will lay the theoretical foundation for the analyses in the following sections. Though the concept of coexisting interval has been introduced and discussed in \cite{FangTCOM16a,FangTCOM16b}, many new things will be added below. Before our discussion, let us introduce the concept of {\em Equivalent Random Variables} for convenience. 
\begin{definition}[Equivalent Random Variables]
	If two random variables $X$ and $Y$ have the same distribution, we say that $X$ and $Y$ are equivalent to each other and denote this equivalence as $X\simeq Y$.
\end{definition}
According to this definition, if $X^n$ is an i.i.d. random process, then 
\begin{align}
	\sum_{i=1}^{n}{X_i\cdot2^{(n-(i-1))r}} \simeq \sum_{i=1}^{n}{X_i2^{ir}}.\nonumber
\end{align}
For conciseness, we will no longer distinguish equivalent random variables and simply redefine $\ell(X^n)$ as 
\begin{align}
	\ell(X^n) \triangleq (1-2^{-r})\sum_{i=1}^{n}{X_i2^{ir}} = (2^r-1)\sum_{i=0}^{n-1}{X_{i+1}2^{ir}}.\nonumber
\end{align}
For any $x^n\in\mathbb{B}^n$, it is easy to know
\begin{align}
	0 = \ell(0^n) \leq \ell(x^n) \leq \ell(1^n) = (2^r-1)\frac{2^{nr}-1}{2^r-1} = 2^{nr} - 1.\nonumber
\end{align}
The following proposition holds obviously.
\begin{proposition}[Relation Between $\ell$-Function and Coset]
	The necessary and sufficient condition for the event that $x^n\in\mathbb{B}^n$ belongs to the $m$-th coset ${\cal C}_m$ is $\ell(x^n)\in(m-1,m]$, where $m\in[0:2^{nr})$. It can be written as $\{x^n\in{\cal C}_m\} \leftrightarrow \{\ell(x^n)\in(m-1,m]\}$, where $\{\cdot\}\leftrightarrow\{\cdot\}$ denotes the equivalence between two events.
\end{proposition}

\subsection{Definition of Coexisting Interval}\label{subsec:coexist}
For $0\leq d\leq n$, we define $j^d\triangleq\{j_1,\dots,j_d\}\subseteq[n]\triangleq\{1,\dots,n\}$, where $1\leq j_1<j_2<\dots<j_d\leq n$. Let $b^d\triangleq(b_1,\dots,b_d)\in\mathbb{B}^d$. Further, we define
\begin{align}\label{eq:Jnd}
	{\cal J}_{n,d} \triangleq \{j^d:1\leq j_1<j_2<\dots<j_d\leq n\}.
\end{align}
The following properties of $j^d$ and ${\cal J}_{n,d}$ are obvious:
\begin{itemize}
	\item $j^0=\emptyset$, ${\cal J}_{n,0}=\{\emptyset\}$, ${\cal J}_{n,1}=\{\{1\},\dots,\{n\}\}$, $j^n\equiv[n]$, and ${\cal J}_{n,n}=\{[n]\}$;
	\item The cardinality of ${\cal J}_{n,d}$ is $|{\cal J}_{n,d}| = \binom{n}{d}$; and 
	\item If we define $|{\cal J}_{n,0}|=1$, then $\sum_{d=0}^{n}{|{\cal J}_{n,d}|} = 2^n$ and $\sum_{d=0}^{n}{\left(2^d\cdot|{\cal J}_{n,d}|\right)} = 3^n$.
\end{itemize}

\begin{definition}[Shift Function]
For $j^d\in\mathcal{J}_{n,d}$ and $b^d\in\mathbb{B}^d$, we define the shift function as \cite{FangTCOM16b} 
\begin{align}\label{eq:tau}
	\tau(j^d,b^d) \triangleq (1-2^{-r})\sum_{d'=1}^d{(1-2b_{d'})2^{rj_{d'}}}\in\mathbb{R}.
\end{align}
\end{definition}

\begin{lemma}[Properties of Shift Function]
	We have $\tau(j^d,b^d\oplus1^d) = -\tau(j^d,b^d)$ and $-2^{nr}<\tau(j^d,b^d)<2^{nr}$.
\end{lemma}
\begin{proof}
	Due to the symmetry, we have $\tau(j^d,b^d\oplus1^d) = -\tau(j^d,b^d)$. Let us rewrite \eqref{eq:tau} as
	\begin{align}\label{eq:tauvar}
		\tau(j^d,b^d)
		&= \underbrace{(1-2^{-r})\sum_{d'=1}^{d}{2^{rj_{d'}}}}_{c(j^d)} - 2\underbrace{(1-2^{-r})\sum_{d'=1}^{d}{b_{d'}2^{rj_{d'}}}}_{v(j^d,b^d)}\nonumber\\
		&= c(j^d) - 2v(j^d,b^d).
	\end{align}
	Since $b^d\in\mathbb{B}^d$, we have 
	\begin{align}
		0 = v(j^d,0^d) \leq v(j^d,b^d)\leq v(j^d,1^d) = c(j^d).\nonumber
	\end{align}
	Therefore, 
	\begin{align}
		-c(j^d) = \tau(j^d,1^d) \leq \tau(j^d,b^d) \leq \tau(j^d,0^d) = c(j^d).\nonumber
	\end{align}
	It is easy to know $c(j^d)\leq c(j^n)=\ell(1^n)=(2^{nr}-1)$. Then we have
	\begin{align}
		-2^{nr}<-(2^{nr}-1) = \tau(j^n,1^n) \leq \tau(j^d,b^d) \leq \tau(j^n,0^n) = (2^{nr}-1)<2^{nr}.\nonumber
	\end{align}
	Therefore, $\tau(j^d,b^d)$ is defined over $(-2^{nr},2^{nr})$.
\end{proof}

\begin{lemma}[Physical Meaning of Shift Function]\label{prop:tau}
	Let $x^n\in\mathbb{B}^n$, $y^n\in\mathbb{B}^n$, and $z^n=x^n\oplus y^n\in\mathbb{B}^n$. We define $x_{j^d}\triangleq(x_{j_1},\dots,x_{j_d})\in\mathbb{B}^d$. If $z_{j^d}\triangleq(z_{j_1},\dots,z_{j_d})=1^d$ and $z_{[n]\setminus j^d}=0^{n-d}$, then
	\begin{align}\label{eq:ell_yn}
		\ell(y^n) = \ell(x^n) + \tau(j^d,x_{j^d}).
	\end{align}
\end{lemma}
\begin{proof}
	For $1\leq i\leq n$, the following two branches hold obviously:
	\begin{itemize}
		\item If $x_i=0$ and $y_i=1$, then $x_i2^{ir}=0$ and $y_i2^{ir}=2^{ir}=x_i2^{ir}+2^{ir}$;  
		\item If $x_i=1$ and $y_i=0$, then $x_i2^{ir}=2^{ir}$ and $y_i2^{ir}=0=x_i2^{ir}-2^{ir}$. 
	\end{itemize}
	The above two branches can be merged as $y_i2^{ir}=x_i2^{ir} + (1-2x_i)2^{ir}$. Then \eqref{eq:ell_yn} follows immediately.
\end{proof}

\begin{definition}[Coexisting Interval]
	Let $\{(a,b]+\tau\} \triangleq (a+\tau,b+\tau]\subset\mathbb{R}$. For $m\in[1:2^{nr})$, the $m$-th coexisting interval associated with $j^d\in{\cal J}_{n,d}$ and $b^d\in\mathbb{B}^d$ is defined as
	\begin{align}\label{eq:Im}
		{\frak I}_m{(j^d,b^d)} \triangleq \{(m-1,m]-\tau(j^d,b^d)\} \cap (m-1,m],
	\end{align}
	where $(m-1,m]$ is called the $m$-th unit interval. Obviously, ${\frak I}_m{(j^d,b^d)}\subseteq (m-1,m]$.
\end{definition}

Of course, we can define the $0$-th unit interval and the $0$-th coexisting interval as $(-1,0]\cap[0,2^{nr})=[0,0]$, which includes only one point $0$ in the real field $\mathbb{R}$. As pointed out at the end of Sect.~\ref{sec:ccs}, there is only one codeword $0^n$ in the $0$-th coset ${\cal C}_0$, so we will no longer discuss the $0$-th coexisting interval in the following.

\begin{lemma}[Concrete Form of Coexisting Interval]
	Depending on the value of $\tau(j^d,b^d)$, the $m$-th coexisting interval associated with $j^d\in{\cal J}_{n,d}$ and $b^d\in\mathbb{B}^d$ has different forms:
	\begin{align}\label{eq:frakI}
		{\frak I}_m(j^d,b^d) = 
		\begin{cases}
			\emptyset, 				& |\tau(j^d,b^d)|\geq 1\\
			(m-1, m-\tau(j^d,b^d)],	& 0\leq\tau(j^d,b^d)<1\\
			(m-1-\tau(j^d,b^d), m],	& -1<\tau(j^d,b^d)\leq0
		\end{cases}.
	\end{align}
	Consequently, the complement of ${\frak I}_m(j^d,b^d)$ is
	\begin{align}\label{eq:frakI_bar}
	\bar{\frak I}_m(j^d,b^d) 
	&\triangleq (m-1,m]\setminus{\frak I}_m(j^d,b^d)\nonumber\\
	&= \begin{cases}
		(m-1,m], 					& |\tau(j^d,b^d)|\geq 1\\
		(m-\tau(j^d,b^d),m],	& 0\leq\tau(j^d,b^d)<1\\
		(m-1,m-1-\tau(j^d,b^d)],& -1<\tau(j^d,b^d)\leq0
	\end{cases}.
\end{align}
\end{lemma}

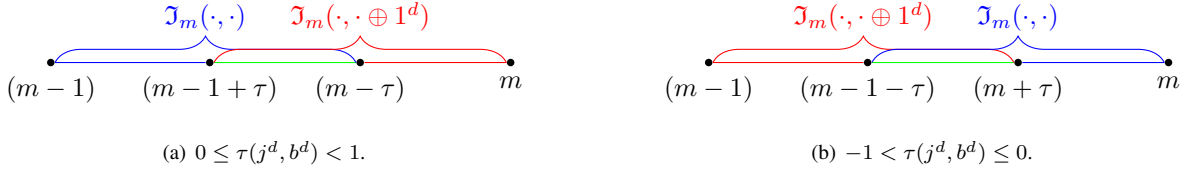
\begin{figure}
	\subfigure[$0\leq\tau(j^d,b^d)<1$.]{
		\begin{tikzpicture}
			\node[circle, fill, inner sep=1, label=below:$(m-1)$](c1){};
			\node[circle, fill, inner sep=1, label=below:$(m-1+\tau)$](c2)[right=2 of c1]{};
			\node[circle, fill, inner sep=1, label=below:$(m-\tau)$](c3)[right=4 of c1]{};
			\node[circle, fill, inner sep=1, label=below:$m$](c4)[right=6 of c1]{};
			\draw[blue](c1)--(c2);
			\draw[green](c2)--(c3);
			\draw[red](c3)--(c4);			
			\draw[blue, decorate, decoration={brace,amplitude=10pt}](c1) to node[above=0.3]{${\frak I}_m(\cdot,\cdot)$} (c3);
			\draw[red, decorate, decoration={brace,amplitude=10pt}](c2) to node[above=0.3]{${\frak I}_m(\cdot,\cdot\oplus 1^d)$} (c4);			
		\end{tikzpicture}
	}\qquad\qquad
	\subfigure[$-1<\tau(j^d,b^d)\leq0$.]{
		\begin{tikzpicture}
			\node[circle, fill, inner sep=1, label=below:$(m-1)$](c1){};
			\node[circle, fill, inner sep=1, label=below:$(m-1-\tau)$](c2)[right=2 of c1]{};
			\node[circle, fill, inner sep=1, label=below:$(m+\tau)$](c4)[right=4 of c1]{};
			\node[circle, fill, inner sep=1, label=below:$m$](c4)[right=6 of c1]{};
			\draw[red](c1)--(c2);
			\draw[green](c2)--(c3);
			\draw[blue](c3)--(c4);
			\draw[red, decorate, decoration={brace,amplitude=10pt}](c1) to node[above=0.3]{${\frak I}_m(\cdot,\cdot\oplus 1^d)$} (c3);
			\draw[blue, decorate, decoration={brace,amplitude=10pt}](c2) to node[above=0.3]{${\frak I}_m(\cdot,\cdot)$} (c4);
		\end{tikzpicture}
	}
	\caption{The $m$-th pair of mirror coexisting intervals associated with $j^d\in{\cal J}_{n,d}$ and $b^d\in\mathbb{B}^d$, where $\tau$ is a shortened form of $\tau(j^d,b^d)$, ${\frak I}_m(\cdot,\cdot)$ is a shortened form of ${\frak I}_m(j^d,b^d)$, and ${\frak I}_m(\cdot,\cdot\oplus 1^d)$ is a shortened form of ${\frak I}_m(j^d,b^d\oplus 1^d)$ for a more pleasing appearance.}
	\label{fig:coexist}
\end{figure}

\begin{lemma}[Length of Coexisting Interval]
	Let $|{\cal I}|$ be the length of continuous interval ${\cal I}$. Then 
	\begin{align}\label{eq:len}
		|{\frak I}_m(j^d,b^d)| = \left(1-|\tau(j^d,b^d)|\right)^+ \triangleq \max(0,1-|\tau(j^d,b^d)|)\in[0,1].
	\end{align}
\end{lemma}

\begin{definition}[Mirror Coexisting Interval]
	We call ${\frak I}_m(j^d,b^d)$ and ${\frak I}_m(j^d,b^d\oplus1^d)$ as the $m$-th pair of mirror coexisting intervals associated with $j^d\in{\cal J}_{n,d}$ and $b^d\in\mathbb{B}^d$.
\end{definition}

A pair of mirror coexisting intervals associated with $j^d\in{\cal J}_{n,d}$ and $b^d\in\mathbb{B}^d$ is given in Fig.~\ref{fig:coexist} for $|\tau(j^d,b^d)|<1$. With the help of Fig.~\ref{fig:coexist}, it is easy to prove the following properties of (mirror) coexisting intervals.
 
\begin{lemma}[Relations between Mirror Coexisting Intervals]\label{prop:mirror} 
	The following three points hold.
	\begin{itemize}
		\item The mirror of a coexisting interval can be obtained by a simple shift:
		\begin{align}\label{eq:mirror}
			{\frak I}_m(j^d,b^d\oplus1^d) = {\frak I}_m(j^d,b^d) + \tau(j^d,b^d) \subseteq (m-1,m].
		\end{align}
		\item A pair of mirror coexisting intervals must belong to the same unit interval. More concretely, the $m$-th pair of mirror coexisting intervals must belong to the $m$-th unit interval.
		\item The $m$-th pair of mirror coexisting intervals are almost symmetric around the point $(m-0.5)$ (except at the end points of coexisting intervals). 
\end{itemize}
\end{lemma}
\begin{proof}
	If $\tau(j^d,b^d)\geq0$, then $\tau(j^d,b^d\oplus1^d)=-\tau(j^d,b^d)\leq0$. According to \eqref{eq:frakI}, we have 
	\begin{align}\label{eq:pos}
		\begin{cases}
			{\frak I}_m(j^d,b^d) = (m-1, m-\tau(j^d,b^d)]\\
			{\frak I}_m(j^d,b^d\oplus1^d) = (m-1-\tau(j^d,b^d\oplus1^d), m] = (m-1+\tau(j^d,b^d), m]
		\end{cases}.
	\end{align}
	If $\tau(j^d,b^d)\leq0$, then $\tau(j^d,b^d\oplus1^d)=-\tau(j^d,b^d)\geq0$. According to \eqref{eq:frakI}, we have 
	\begin{align}\label{eq:neg}
		\begin{cases}
			{\frak I}_m(j^d,b^d) = (m-1-\tau(j^d,b^d),m]\\
			{\frak I}_m(j^d,b^d\oplus1^d) = (m-1,m-\tau(j^d,b^d\oplus1^d)] = (m-1,m+\tau(j^d,b^d)]
		\end{cases}.
	\end{align}
	Combining the above two branches, we will obtain \eqref{eq:mirror} immediately. 
	
	The second bullet point holds obviously, so its proof is omitted to avoid verbosity. As for the symmetry between ${\frak I}_m(j^d,b^d)$ and ${\frak I}_m(j^d,b^d\oplus1^d)$, it can be found from \eqref{eq:pos} and \eqref{eq:neg} that the sum of the lower bound of ${\frak I}_m(j^d,b^d)$ and the upper bound of ${\frak I}_m(j^d,b^d\oplus1^d)$ is always $(2m-1)=2(m-0.5)$, and so is the sum of the upper bound of ${\frak I}_m(j^d,b^d)$ and the lower bound of ${\frak I}_m(j^d,b^d\oplus1^d)$.
\end{proof}

\begin{definition}[Coexisting Interval Set] The set of coexisting intervals associated with $j^d\in{\cal J}_{n,d}$ and $b^d\in\mathbb{B}^d$ is
	\begin{align}
		{\frak I}(j^d,b^d) \triangleq \left\{{\frak I}_m(j^d,b^d):m\in[1:2^{nr})\right\}.
	\end{align}
\end{definition}

\begin{theorem}[Necessary and Sufficient Condition for Coexistence]\label{thm:equiv}
	Consider two binary blocks $x^n\in\mathbb{B}^n$ and $y^n\in\mathbb{B}^n$. Let $z^n=x^n\oplus y^n\in\mathbb{B}^n$. If $z_{j^d}=1^d$ and $z_{[n]\setminus j^d}=0^{n-d}$, the necessary and sufficient condition for the event that $x^n$ and $y^n$ coexist in the same coset is $\ell(x^n)\in{\frak I}(j^d,x_{j^d})$, or equivalently $\ell(y^n)\in{\frak I}(j^d,y_{j^d})={\frak I}(j^d,x_{j^d}\oplus 1^d)$.
\end{theorem}
\begin{proof}
	According to \eqref{eq:ell_yn} of Lemma~\ref{prop:tau}, we have $\ell(y^n)=\ell(x^n)+\tau(j^d,x_{j^d})$, where $j^d\in{\cal J}_{n,d}$ and $x_{j^d}\in\mathbb{B}^d$. If $\ell(x^n)\in{\frak I}_m(j^d,x_{j^d})\subseteq(m-1,m]$, then according to \eqref{eq:mirror} of Lemma~\ref{prop:mirror}, we have
	\begin{align}
		\ell(y^n) \in {\frak I}_m(j^d,x_{j^d}) + \tau(j^d,x_{j^d}) = {\frak I}_m(j^d,x_{j^d}\oplus 1^d) = {\frak I}_m(j^d,y_{j^d})\subseteq(m-1,m].
	\end{align}
	Since $\ell(x^n)\in(m-1,m]$ and $\ell(y^n)\in(m-1,m]$, it is clear that both $x^n$ and $y^n$ belong to the $m$-th coset. After generalization for all $m\in[1:2^{nr})$, if $\ell(x^n)\in{\frak I}(j^d,x_{j^d})$, or equivalently $\ell(y^n)\in{\frak I}(j^d,y_{j^d})={\frak I}(j^d,x_{j^d}\oplus 1^d)$, then both $x^n$ and $y^n$ belong to the same coset.

	{\bf Converse}. If $\ell(x^n)\in\bar{\frak I}_m(j^d,x_{j^d})\subseteq(m-1,m]$, where $\bar{\frak I}_m(j^d,x_{j^d})$ is the complement of ${\frak I}_m(j^d,x_{j^d})$ given by \eqref{eq:frakI_bar}, then according to $\ell(y^n)=\ell(x^n)+\tau(j^d,x_{j^d})$, we have
	\begin{align}
		\ell(y^n) \in
		\begin{cases}
			(m-1+\tau(j^d,x_{j^d}),m+\tau(j^d,x_{j^d})], & |\tau(j^d,x_{j^d})|\geq 1\\
			(m,m+\tau(j^d,x_{j^d})],	& 0\leq\tau(j^d,x_{j^d})<1\\
			(m-1+\tau(j^d,x_{j^d}),m-1],& -1<\tau(j^d,x_{j^d})\leq0
		\end{cases}.\nonumber
	\end{align}
	Obviously, in all of the above three cases, $\ell(y^n)\notin(m-1,m]$. While as we know, $\ell(x^n)\in(m-1,m]$. According to the relation between $\ell$-function and coset, it is evident that $x^n$ and $y^n$ do not coexist in the same coset.
\end{proof}

\begin{corollary}[Necessary Condition for Coexistence]\label{corol:neccoe}
Let $x^n\in\mathbb{B}^n$, $y^n\in\mathbb{B}^n$, and $z^n=x^n\oplus y^n$. If $z_{j^d}=1^d$ and $z_{[n]\setminus j^d}=0^{n-d}$, the event that $x^n$ and $y^n$ coexist in the same coset happens only if
$|\tau(j^d,x_{j^d})|<1$.
\end{corollary}
\begin{proof}
	This corollary is a direct result of Theorem~\ref{thm:equiv}.
\end{proof}

\subsection{Probability of Coexisting Interval}
Now we ponder over the following important problem: Given $X_{j^d}=b^d\in\mathbb{B}^d$, how possible will $\ell(X^n)$ fall into ${\frak I}(j^d,b^d)$? To answer this question, we should introdude the following important random variable.

\begin{definition}[Conditional Value of $\ell(X^n)$]
	For every $j^d\in{\cal J}_{n,d}$ and every $b^d\in\mathbb{B}^d$, we define the conditional value of $\ell(X^n)$ given $X_{j^d}=b^d$ as
	\begin{align}\label{eq:E}
		E{(j^d,b^d)}\triangleq\ell(X^n|X_{j^d}=b^d).
	\end{align}	
\end{definition}

According to the definition of $\ell(X^n)$, we have
\begin{align}\label{eq:Evar}
	E{(j^d,b^d)} 
	&= 
	\underbrace{(1-2^{-r})\sum_{d'=1}^{d}{b_{d'}2^{rj_{d'}}}}_{{\rm deterministic~constant}~c(j^d,b^d)} + \underbrace{(1-2^{-r})\sum_{i\in[n]\setminus j^d}{X_i2^{ir}}}_{{\rm random~variable}~\Upsilon([n]\setminus j^d)~{\rm for}~d<n}\nonumber\\
	&= c(j^d,b^d) + \Upsilon([n]\setminus j^d),
\end{align}
where $c(j^d,b^d)$ is a deterministic constant parameterized by $j^d$ and $b^d$, while $\Upsilon([n]\setminus j^d)$ is a function w.r.t. $(n-d)$ random variables $X_{[n]\setminus j^d}$. For conciseness, we abbreviate $c(j^d,1^d)$ to $c(j^d)$, just as defined by \eqref{eq:tauvar}. Consequently, $E{(j^d,b^d)}$ is also a function w.r.t $(n-d)$ random variables $X_{[n]\setminus j^d}$. If $d<n$, then $E{(j^d,b^d)}$ is itself a random variable; otherwise, if $d=n$, then $E{(j^d,b^d)}$ actually degenerates into a deterministic constant parameterized by $b^n$, \textit{i.e.}, $E{(j^n,b^n)}=\ell(b^n) = c(j^n,b^n) = (1-2^{-r})\sum_{i=1}^{n}{b_i2^{ir}}$.

\begin{remark}[Comparison of $E(j^d,b^d)$ with $\tau(j^d,b^d)$]
	After observing \eqref{eq:tauvar} and \eqref{eq:Evar}, the reader may find that $E(j^d,b^d)$ and $\tau(j^d,b^d)$ are very similar to each other. However, the reader should notice the differences between them. Most importantly, $\tau(j^d,b^d)$ is a deterministic constant defined over $(-2^{nr},2^{nr})$, while $E(j^d,b^d)$ is a random variable defined over $[0,2^{nr})$. In addition, $c(j^d)$ defined by \eqref{eq:tauvar} is a special case of $c(j^d,b^d)$ defined by \eqref{eq:Evar} when $b^d=1^d$. Finally, $v(j^d,b^d)$ defined by \eqref{eq:tauvar} is the sum of $d$ terms, while $\Upsilon([n]\setminus j^d)$ is the sum of $(n-d)$ terms. 
\end{remark}

\begin{theorem}[Asymptotic Probability of Coexisting Interval]\label{thm:prob}
	If the sequence $(2^r,2^{2r},\dots)$ satisfies Wilms' condition, then we have
	\begin{align}\label{eq:Eprob}
		\lim_{(n-d)\to\infty}\Pr\left\{E{(j^d,b^d)}\in{\frak I}(j^d,b^d)\right\} = \left(1-|\tau(j^d,b^d)|\right)^+.
	\end{align}
\end{theorem}
\begin{proof}	
	As shown by \eqref{eq:Evar}, $E(j^d,b^d)$ is the sum of a deterministic constant $c(j^d,b^d)$ and a random variable $\Upsilon([n]\setminus j^d)$. Let us pay attention to $\Upsilon([n]\setminus j^d)$. The cardinality of $[n]\setminus j^d$ will go to infinity as $(n-d)\to\infty$. Since the sequence $(2^r,2^{2r},\dots)$ satisfies Wilms' condition, according to Theorem~\ref{thm:genud}, it is sure that $\Upsilon([n]\setminus j^d)$ will be u.d. mod 1 as $(n-d)\to\infty$, and in turn $E(j^d,b^d)$ will also be u.d. mod 1. In other words, if $\lceil{E{(j^d,b^d)}}\rceil=m$, then $E(j^d,b^d)$ will be u.d. over $(m-1,m]$ as $(n-d)\to\infty$. Therefore,
	\begin{align}
		\lim_{(n-d)\to\infty}\Pr\left\{E{(j^d,b^d)}\in{\frak I}_m(j^d,b^d)\middle|\lceil{E{(j^d,b^d)}}\rceil=m\right\} 
		&= \frac{|{\frak I}_m(j^d,b^d)|}{|(m-1,m]|} \nonumber\\
		&\stackrel{(a)}{=} \left(1-|\tau(j^d,b^d)|\right)^+,
	\end{align}
	where $(a)$ comes from \eqref{eq:len}. After averaging over all coexisting intervals associated with $j^d\in{\cal J}_{n,d}$ and $b^d\in\mathbb{B}^d$, we will obtain \eqref{eq:Eprob}. 
\end{proof}

\begin{remark}[A Note about the Proof of Theorem~\ref{thm:prob}]
	If $(n-d)=\infty$, then $|[n]\setminus j^d|=\infty$, where $|[n]\setminus j^d|$ denotes the cardinality of $[n]\setminus j^d$, and hence $\Upsilon([n]\setminus j^d)$ will be a continuous random variable u.d. mod 1, and so is $E(j^d,b^d)$, provided with that the sequence $(2^r,2^{2r},\dots)$ satisfies Wilms' condition. If $(n-d)<\infty$, then $\Upsilon([n]\setminus j^d)$ will be a discrete random variable, and so is $E(j^d,b^d)$. Finally, if $d=n$, then $\Upsilon([n]\setminus j^d)$ does not exist and hence $E(j^d,b^d)$ is a deterministic constant rather than a random variable.   
\end{remark}

\section{Hamming Distance Spectrum}\label{sec:hds1}
This section will formally define the concept of \textit{Hamming Distance Spectrum} (HDS) and prove its convexity. Then three methods will be given to calculate HDS: 
\begin{itemize}
	\item The \textit{Exhaustive Enumeration} is the most accurate, but it has the highest complexity;
	\item The \textit{Binomial Approximation} is the simplest, but its accurateness is poor; and
	\item With the help of \textit{coexisting intervals} defined in Sect.~\ref{sec:coexist}, we derive the \textit{Soft Approximatio}n of HDS, which is accurate for $d\ll n$ and has a low complexity for small $d$, where $d$ is Hamming distance and $n$ is code length.
\end{itemize}
On the basis of \textit{Soft Approximation}, this section will also discuss the convergence of HDS:
\begin{itemize}
	\item For $d=1$ and $2$, the closed form of HDS is derived, showing that the HDS is always convergent;
	\item For $d=3$, the necessary and sufficient condition is given for the convergence of HDS; and
	\item For $d=3$, the closed form of HDS is derived in two divergent cases.  
\end{itemize}

\subsection{Definition of HDS}
\begin{definition}[Codeword HDS]
	For $d\in[0:n]$, the HDS of codeword $x^n\in\mathbb{B}^n$ is defined as 
	\begin{align}
		k(x^n,d) \triangleq \left|\left\{y^n: y^n\in\mathbb{B}^n{\rm~and~}m(x^n)=m(y^n){\rm~and~}|x^n\oplus y^n|=d\right\}\right|.\nonumber
	\end{align}
\end{definition}

\begin{remark}[Properties of Codeword HDS]
	In plain words, $k(x^n,d)$ is the number of codewords in the coset containing $x^n$ that are $d$-away (in Hamming distance) from $x^n$. Clearly, $0\leq k(x^n,d)\leq\binom{n}{d}$. Let ${\cal C}_m$ be the coset containing $x^n$, whose cardinality is denoted by $|{\cal C}_m|$, where $m\in[0:2^{nr})$. It is easy to obtain the following properties:
	\begin{itemize}
		\item For $1\leq d\leq n$, $\sum_{x^n\in{\cal C}_m}{k(x^n,d)}$ is twice the number of $d$-away codeword-pairs in ${\cal C}_m$. Therefore,
		\begin{align}\label{eq:sumk1}
			\sum_{d=1}^{n}{\sum_{x^n\in{\cal C}_m}{k(x^n,d)}} = 2\binom{|{\cal C}_m|}{2} = |{\cal C}_m|^2 - |{\cal C}_m|.
		\end{align}
		\item If we define $k(x^n,0)=1$, then $\sum_{d=0}^{n}{k(x^n,d)} = |\mathcal{C}_{m}|$ and further
		\begin{align}\label{eq:sumk}
			\sum_{x^n\in{\cal C}_m}{\sum_{d=0}^{n}{k(x^n,d)}} = |{\cal C}_m|^2.
		\end{align}
	\end{itemize}	
\end{remark}

\begin{definition}[Code HDS]
	Let $\mathbb{E}[\cdot]$ denote the expectation. For $0\leq d\leq n$, the HDS of an overlapped arithmetic code is defined as
	\begin{align}
		\psi(d;n) \triangleq \mathbb{E}_{X^n}[k(X^n,d)] = \sum_{x^n\in\mathbb{B}^n}{\Pr(X^n=x^n)\cdot k(x^n,d)}.\nonumber
	\end{align}
\end{definition}

\begin{remark}[Code HDS for Uniform Binary Sources]
	If $\Pr(X^n=x^n)\equiv2^{-n}$, then
	\begin{align}\label{eq:psinddef}
		\psi(d;n) 
		= 2^{-n}\sum_{x^n\in\mathbb{B}^n}{k(x^n,d)} = 2^{-n}\sum_{m=0}^{2^{nr}-1}\sum_{x^n\in{\cal C}_m}{k(x^n,d)}.
	\end{align}
\end{remark}

\begin{definition}[Asymptotic Code HDS]
	For $0\leq d\leq n$, the asymptotic HDS of an overlapped arithmetic code is
	\begin{align}
		\psi(d) \triangleq \lim_{n\rightarrow\infty}{\psi(d;n)}.\nonumber
	\end{align}
\end{definition}

\begin{example}
	A good example for $n=4$ and $r=0.5$ is given in Table~I of \cite{FangTCOM16a} to explain the concept of HDS, which unequally partitions source space $\mathbb{B}^4$ into the following four cosets:
	\begin{itemize}
		\item ${\cal C}_0 = \{\underline{0000}\}$,
		\item ${\cal C}_1 = \{\underline{0001}, \underline{0010}, \underline{0011}, \underline{0100}\}$,
		\item ${\cal C}_2 = \{\underline{0101}, \underline{0110}, \underline{0111}, \underline{1000}, \underline{1001}, \underline{1010}, \underline{1100}\}$, and
		\item ${\cal C}_3 = \{\underline{1011}, \underline{1101}, \underline{1110}, \underline{1111}\}$,
	\end{itemize}
	where $\underline{b_1\cdots b_n}$ denotes a sequence of $n$ bits. Let us take ${\cal C}_1$ as an example. It is easy to get 
	\begin{itemize}
		\item $k(\underline{0001},1)=1$, $k(\underline{0001},2)=2$, $k(\underline{0001},3)=0$, and $k(\underline{0001},4)=0$; 
		\item $k(\underline{0010},1)=1$, $k(\underline{0010},2)=2$, $k(\underline{0010},3)=0$, and $k(\underline{0010},4)=0$; 
		\item $k(\underline{0011},1)=2$, $k(\underline{0011},2)=0$, $k(\underline{0011},3)=1$, and $k(\underline{0011},4)=0$; 
		\item $k(\underline{0100},1)=0$, $k(\underline{0100},2)=2$, $k(\underline{0100},3)=1$, and $k(\underline{0100},4)=0$.
	\end{itemize}
	If we define $k(x^n,0)=1$, then $\sum_{d=0}^{4}{k(x^n,d)}=|{\cal C}_1|=4$ for every $x^n\in{\cal C}_1$. In ${\cal C}_1$, there are two $1$-away codeword pairs, \textit{i.e.}, $(\underline{0001},\underline{0011})$ and $(\underline{0010}, \underline{0011})$, so $\sum_{x^n\in{\cal C}_1}{k(x^n,1)}=1+1+2+0=4$ is twice the number of $1$-away codeword-pairs; there are three $2$-away codeword-pairs, \textit{i.e.}, $(\underline{0001},\underline{0010})$, $(\underline{0001},\underline{0100})$, and $(\underline{0010}, \underline{0100})$, so $\sum_{x^n\in{\cal C}_1}{k(x^n,2)}=2+2+0+2=6$ is twice the number of $2$-away codeword-pairs; and there is one $3$-away codeword-pair, \textit{i.e.}, $(\underline{0011},\underline{0100})$, so $\sum_{x^n\in{\cal C}_1}{k(x^n,3)}=0+0+1+1=2$ is twice the number of $3$-away codeword-pairs. Finally, we have
	\begin{align}
		\sum_{d=1}^{4}\sum_{x^n\in{\cal C}_1}{k(x^n,d)}=4+6+2+0=12=2\binom{|{\cal C}_1|}{2}=|{\cal C}_1|^2-|{\cal C}_1|,\nonumber
	\end{align}
	verifying \eqref{eq:sumk1}, and 
	\begin{align}
		\sum_{d=0}^{4}\sum_{x^n\in{\cal C}_1}{k(x^n,d)}=4+12=16=|{\cal C}_1|^2,\nonumber
	\end{align}
	verifying \eqref{eq:sumk}. According to \eqref{eq:psinddef}, we can easily obtain $\psi(0;4)=16/16=1$, $\psi(1;4)=20/16=5/4$, $\psi(2;4)=28/16=7/4$, $\psi(3;4)=12/16=3/4$, and $\psi(4;4)=6/16=3/8$. Hence, $\sum_{d=0}^{4}{\psi(d;4)}=82/16=5.125>4=2^{4(1-0.5)}$. In this example, we find $\sum_{d=0}^{n}{\psi(d;n)} > 2^{n(1-r)}$, which is not a coincidence, as shown below.
\end{example}

\begin{lemma}[Sum of HDS]\label{lem:sumhds}
	Let $f(u)$ be the asymptotic CCS defined by \eqref{eq:asympccs}. Then
	\begin{align}
		\lim_{n\to\infty}\frac{\sum_{d=0}^{n}{\psi(d;n)}}{2^{n(1-r)}} = \int_{0}^{1}{f^2(u)du}.
	\end{align}
\end{lemma}
\begin{proof}
	According to \eqref{eq:sumk} and \eqref{eq:psinddef}, we will obtain
	\begin{align}
		\sum_{d=0}^{n}{\psi(d;n)} 
		&= 2^{-n}\sum_{m=0}^{2^{nr}-1}{|{\cal C}_m|^2}\nonumber\\
		&= 2^{n(1-r)}\sum_{m=0}^{2^{nr}-1}{\left(\frac{|{\cal C}_m|}{2^{n(1-r)}}\right)^2}2^{-nr}.\nonumber
	\end{align}
	It can be written as
	\begin{align}
		\lim_{n\to\infty}\frac{\sum_{d=0}^{n}{\psi(d;n)}}{2^{n(1-r)}} 
		= \lim_{n\to\infty}\sum_{m=0}^{2^{nr}-1}{\left(\frac{|{\cal C}_m|}{2^{n(1-r)}}\right)^2}2^{-nr}
		\stackrel{(a)}{=} \int_{0}^{1}{f^2(u)du},\nonumber
	\end{align}
	where $(a)$ comes from \eqref{eq:fu_physical}.
\end{proof}	

\begin{lemma}[Convexity of HDS]
	For an overlapped arithmetic code with length $n$ and rate $r$, we have 
	\begin{align}\label{eq:conv}
		\sum_{d=0}^{n}{\psi(d;n)} \geq 2^{n(1-r)},
	\end{align}
	where the equality holds {\em if and only if} (iff) source space $\mathbb{B}^n$ is equally partitioned into $2^{nr}$ cosets.
\end{lemma}
\begin{proof}
	Since $\int_{0}^{1}{f^2(u)du}$ is a nonnegative and convex function in $f(u)$, we have $\int_{0}^{1}{f^2(u)du}\geq1$ and the equality holds iff $f(u)$ is uniform over $[0,1)$, \textit{i.e.}, source space $\mathbb{B}^n$ is equally partitioned into $2^{nr}$ cosets of cardinality $2^{n(1-r)}$. This lemma then follows Lemma~\ref{lem:sumhds} immediately.
\end{proof}

\subsection{Calculation of HDS}\label{subsec:hdscal}
\begin{remark}[Exhaustive Enumeration]
As shown by \eqref{eq:psinddef}, to obtain $\psi(d;n)$ for all $d\in[0:n]$, we should try every $x^n\in\mathbb{B}^n$ and every $y^n\in\mathbb{B}^n$, so the total complexity is ${\cal O}(2^n\times2^n)={\cal O}(4^n)$. More concretely, the computing complexity of $\psi(d;n)$ varies for different $d$. Let $z^n=x^n\oplus y^n$ and $|x^n\oplus y^n|=|z^n|=d$. For convenience, we define $z_{j^d}^n\in\mathbb{B}^n$ as a length-$n$ binary sequence with $z_{j^d}=1^d$ and $z_{[n]\setminus j^d}=0^{n-d}$. According to \eqref{eq:psinddef}, we have
\begin{align}\label{eq:psind}
	\psi(d;n) = 2^{-n}\sum_{j^d\in{\cal J}_{n,d}} \sum_{x^n\in\mathbb{B}^n}{\bf 1}_{m(x^n)=m(x^n\oplus z_{j^d}^n)},
\end{align}
where ${\bf 1}_A$ is the well-known indicator function equal to $1$ if $A$ is true, or $0$ if $A$ is false. Obviously, the computing complexity of \eqref{eq:psind} is ${\cal O}(2^n\binom{n}{d})$, following the binomial distribution w.r.t. $d$, extremely huge and unacceptable for every $d\in[0:n]$. Hence, we are badly in need of a fast method to calculate $\psi(d;n)$.
\end{remark}

\begin{theorem}[Binomial Approximation of HDS]\label{thm:coarsehds}
	If we take overlapped arithmetic codes as random codes, then $\psi(d;n)$ for $0\leq d\leq n$ obeys the following binomial distribution 
	\begin{align}\label{eq:coarsehds}
		\psi(d;n) \approx \binom{n}{d}\cdot 2^{-nr} \cdot \int_{0}^{1}{f^2(u)\,du}.
	\end{align}
\end{theorem}
\begin{proof}
	As shown by Lemma~\ref{lem:sumhds},
	\begin{align}
		\sum_{d=0}^{n}{\psi(d;n)} \approx 2^{n(1-r)}\cdot \int_{0}^{1}{f^2(u)du}.\nonumber
	\end{align}
	If $\psi(d;n)$ for $0\leq d\leq n$ obeys the binomial distribution, then
	\begin{align}
		\psi(d;n) \approx \frac{\binom{n}{d}\cdot 2^{n(1-r)}\cdot \int_{0}^{1}{f^2(u)du}}{\sum_{d=0}^{n}\binom{n}{d}}.\nonumber
	\end{align}
	As we know, $\sum_{d=0}^{n}\binom{n}{d}=2^n$. Now \eqref{eq:coarsehds} follows immediately.
\end{proof}

Actually, \eqref{eq:coarsehds} has been given as Equation (56) in sub-Sect.~VII.C of \cite{FangTCOM16a} directly without any proof. However, notice that there is a typo in Equation (56) of \cite{FangTCOM16a}, where $2^{n(1-R)}$ should be $2^{-nR}$. According to \eqref{eq:coarsehds}, the complexity of $\psi(d;n)$ is ${\cal O}(1)$, hence the total complexity is ${\cal O}(n)$ for all $d\in[0:n]$.

Despite its low complexity, one serious drawback of \eqref{eq:coarsehds} is its low accurateness. Because overlapped arithmetic codes are not random codes, $\psi(d;n)$ for $0\leq d\leq n$ does not strictly follow the binomial distribution, especially at the two ends of $[0:n]$. Hence with \eqref{eq:coarsehds}, we can obtain only a coarse approximation of $\psi(d;n)$. For example, \eqref{eq:coarsehds} returns $2^{-nr}\int_{0}^{1}{f^2(u)\,du}$ for $d=0$, while $\psi(0;n)=1$ actually. As verified in Sect.~\ref{sec:example}, \eqref{eq:coarsehds} works well only for $d\approx n/2$, while performs poorly in other cases, which motivates us to look for a more accurate method. 

\begin{theorem}[Asymptotic HDS]\label{thm:hds}
	Let us define ${\cal J}_{n,d}$ as \eqref{eq:Jnd}. As $(n-d)\to\infty$, we have
	\begin{align}\label{eq:psidbd}
		\psi(d)
		&= 2^{-d}\sum_{b^d\in\mathbb{B}^d}\underbrace{\left(\sum_{j^d\in{\cal J}_{\infty,d}}{\left(1-|\tau(j^d,b^d)|\right)^+}\right)}_{\psi(d|b^d)}\nonumber\\
		&= 2^{-d}\sum_{b^d\in\mathbb{B}^d}{\psi(d|b^d)}.
	\end{align}
\end{theorem}
\begin{proof}
	See Appendix~\ref{prf:hds} for the proof.
\end{proof}

On knowing the asymptotic HDS $\psi(d)$ for $n=\infty$ given by \eqref{eq:psidbd}, it is straightforward to obtain an approximation of $\psi(d;n)$ for finite $n$, as given by the following corollary.
\begin{corollary}[Soft Approximation of HDS for $d\ll n$]\label{corol:softhds}
	For $d\ll n$, $\psi(d;n)$ can be approximated by
	\begin{align}\label{eq:softhds}
		\psi(d;n) \approx 2^{-d} \sum_{b^d\in\mathbb{B}^d}\sum_{j^d\in{\cal J}_{n,d}}{\left(1-|\tau(j^d,b^d)|\right)^+}.
	\end{align}
\end{corollary}

We name \eqref{eq:softhds} as the \textit{Soft Approximation} of $\psi(d;n)$, to distinguish it from the so-called \textit{Hard Approximation} of $\psi(d;n)$ in the next section. Actually, \eqref{eq:psidbd} and \eqref{eq:softhds} have been given as Equations (47) and (46) in sub-Sect.~VI.D of \cite{FangTCOM16a}, respectively. However, there lacks an explicit rigorous proof for Equations (47) and (46) in \cite{FangTCOM16a}. Thus, one contribution of this paper is giving a strict formal proof for \eqref{eq:psidbd} and \eqref{eq:softhds} (cf. Appendix~\ref{prf:hds}).

\begin{remark}[Accurateness]
	As shown in Appendix~\ref{prf:hds}, the key step in the proof of Theorem~\ref{thm:hds} is \eqref{eq:approx}, which is based on Theorem~\ref{thm:prob}. Since the prerequisite of Theorem~\ref{thm:prob} is $(n-d)\to\infty$, the accurateness of \eqref{eq:softhds} cannot be guaranteed when $d\approx n$. That is why Corollary~\ref{corol:softhds} requires $d\ll n$. However, as proved by the theoretical analyses in Sect.~\ref{sec:hdsn} and also verified by the experimental results in Sect.~\ref{sec:example}, \eqref{eq:softhds} is actually accurate for almost every $d$. The only exception is $\psi(n;n)$, which cannot be approximated by \eqref{eq:softhds} otherwise there will be a large deviation.
\end{remark}

\begin{remark}[Complexity]
	As we know, the cardinality of ${\cal J}_{n,d}$ is $\binom{n}{d}$, so the complexity to calculate $\psi(d;n)$ by \eqref{eq:softhds} is ${\cal O}(2^d\binom{n}{d})$. Compared with \eqref{eq:psind}, whose complexity is ${\cal O}(2^n\binom{n}{d})$, the complexity of \eqref{eq:softhds} is reduced by $2^{n-d}$ times. However, the complexity of \eqref{eq:softhds} goes up hyper-exponentially as $d$ increases. Therefore, even though \eqref{eq:softhds} is accurate for almost every $d$, it is feasible only for small $d$ in practice. It will be very desirable if the complexity of \eqref{eq:softhds} can be reduced for large $d$. Since $\sum_{d=0}^{n}2^d\binom{n}{d}=3^n$, the total complexity to calculate $\psi(d;n)$ for all $d\in[0:n]$ by \eqref{eq:softhds} is ${\cal O}(3^n)$. As a comparison, remember that the total complexity of the exhaustive enumeration \eqref{eq:psind} is ${\cal O}(4^n)$, and the total complexity of the binomial approximation \eqref{eq:coarsehds} is ${\cal O}(n)$. 
\end{remark}

\subsection{Convergence of HDS}\label{subsec:converge}
Given $1\leq j_1<j_2<\cdots<j_d$, another form of \eqref{eq:psidbd} is 
\begin{align}\label{eq:psid}
	\psi(d) = 2^{-d}\sum_{b^d\in\mathbb{B}^d}\sum_{j_d=d}^{\infty}\sum_{j_{d-1}=(d-1)}^{(j_d-1)}\cdots\sum_{j_1=1}^{(j_2-1)}{\left(1-|\tau(j^d,b^d)|\right)^+}.
\end{align}
It can be seen that $\psi(d)$ is the sum of infinite terms, so naturally there is an interesting and important problem: Is $\psi(d)<\infty$ or not? If $\psi(d;n)<\infty$ as $n\to\infty$, we say that $\psi(d;n)$ is \textit{convergent}; otherwise, we say that $\psi(d;n)$ is \textit{divergent}. It has been shown in \cite{FangTCOM16a} that $\psi(1)<\infty$ and $\psi(2)<\infty$, but for $d\geq3$, $\psi(d;n)$ may or may not converge as $n\to\infty$. Below, we will first derive the concrete closed forms of $\psi(1)$ and $\psi(2)$, and then give the necessary and sufficient condition for the convergence of $\psi(3)$.

\begin{corollary}[Convergence of $\psi(1)$]\label{corol:psi1}
	As $n\to\infty$, $\psi(1;n)$ will converge to
	\begin{align}
		\psi(1) = \sum_{i=1}^{J_1}{\left(1-(1-2^{-r})2^{ir}\right)}<\infty,\nonumber
	\end{align}	
	where 	
	\begin{align}\label{eq:J1}
		J_1 \triangleq -\left\lfloor\frac{\log_2{(2^r-1)}}{r}\right\rfloor < \infty.
	\end{align} 	
\end{corollary}
\begin{proof}
	As we know, ${\cal J}_{n,1}=\{\{1\},\dots,\{n\}\}$, so when $d=1$, \eqref{eq:psidbd} becomes
	\begin{align}
		\psi(1) = \left(\psi(1|0)+\psi(1|1)\right)/2,\nonumber
	\end{align}
	where 
	\begin{align}
		\psi(1|b) = \sum_{i=1}^{\infty}{\left(1-|\tau(i,b)|\right)^+}.\nonumber
	\end{align}
	Due to the symmetry, $\psi(1|0)=\psi(1|1)$, so we consider only $\psi(1|0)$ below. By \eqref{eq:tau},	
	\begin{align}
		\tau(i,0) = (1-2^{-r}){2^{ir}}>0,\nonumber
	\end{align}
	which is a monotonously increasing function w.r.t. $i$. After solving $\tau(i,0)<1$, we obtain
	\begin{align}
		i \leq J_1 \triangleq \left\lceil-\frac{\log_2{(1-2^{-r})}}{r}\right\rceil-1 = -\left\lfloor\frac{\log_2{(2^r-1)}}{r}\right\rfloor < \infty.\nonumber
	\end{align} 
	Finally, $\psi(1) = \psi(1|0) = \psi(1|1)$.
\end{proof}

\begin{corollary}[Convergence of $\psi(2)$]\label{corol:psi2}
	The closed form of $\psi(2|0^2)$ is
	\begin{align}
		\psi(2|0^2) = \sum_{i=1}^{J_{2,1}}\sum_{k=1}^{\kappa_1(i)}{(1-(1-2^{-r})2^{ir}(2^{kr}+1))},\nonumber
	\end{align}
	where 
	\begin{align}\label{eq:J21}
		J_{2,1} \triangleq -\left\lfloor\frac{\log_2{(2^{2r}-1)}}{r}\right\rfloor
	\end{align}
	and $\kappa_1(i)$ is a function w.r.t. $i$ defined as
	\begin{align}
		\kappa_1(i) \triangleq 
		\left\lceil\frac{\log_2{(2^{-ir}-1+2^{-r})} - \log_2{(2^r-1)}}{r}\right\rceil \leq 
		\kappa_1(1) = \left\lceil\frac{\log_2{(2^{1-r}-1)}-\log_2{(2^r-1)}}{r}\right\rceil.\nonumber
	\end{align}
	The closed form of $\psi(2|\underline{10})$ is
	\begin{align}
		\psi(2|\underline{10}) = \sum_{i=1}^{J_{2,2}}\sum_{k=1}^{\kappa_2(i)}{(1-(1-2^{-r})2^{ir}(2^{kr}-1))},\nonumber
	\end{align}
	where
	\begin{align}\label{eq:J22}
		J_{2,2} \triangleq -\left\lfloor\frac{2\log_2{(2^r-1)}}{r}\right\rfloor	
	\end{align}
	and $\kappa_2(i)$ is a function w.r.t. $i$ defined as
	\begin{align}
		\kappa_2(i) \triangleq 
		\left\lceil\frac{\log_2{(2^{-ir}+1-2^{-r})} - \log_2{(2^r-1)}}{r}\right\rceil \leq
		\kappa_2(1) = J_1.\nonumber
	\end{align}
	As $n\to\infty$, $\psi(2;n)$ will converge to
	\begin{align}
		\psi(2) = \frac{\psi(2|0^2) + \psi(2|\underline{10})}{2}<\infty.\nonumber
	\end{align}
\end{corollary}
\begin{proof}
	See Appendix~\ref{prf:psi2} for the proof.
\end{proof}

To understand the above two corollaries, the curves of $J_1$, $J_{2,1}$, and $J_{2,2}$ w.r.t. $r$ are plotted in sub-Fig.~\ref{subfig:J12} and the curves of $\psi(1)$ and $\psi(2)$ w.r.t. $r$ are plotted in sub-Fig.~\ref{subfig:psi12}. It is easy to verify that $J_1=J_{2,2}=0$ when $r=1$. As $r$ decreases from $1$, $J_1$ and $J_{2,2}$ will jump from $0$ to $1$ immediately. Then at $r\approx0.8114$, which corresponds to $2^r-1=2^{-r/2}$, $J_{2,2}$ jumps from $1$ to $2$; at $r=\log_2{\varphi}\approx 0.6942$, where $\varphi\approx1.618$ is the golden ratio, $J_1$ jumps from $1$ to $2$ and $J_{2,2}$ jumps from $2$ to $3$. Actually, there is a simple relation between $J_1$ and $J_{2,2}$, \textit{e.g.}, $J_{2,2}=2J_1$ or $2J_1-1$. As for $J_{2,1}$, according to \eqref{eq:J21}, we have $J_{2,1}=-1$ for $r\geq\log_2{\varphi}$ and $J_{2,1}=0$ for $0.5\leq r<\log_2{\varphi}$. However, negative $J_{2,1}$ makes no sense, so we lower bound $J_{2,1}$ by $0$ in sub-Fig.~\ref{subfig:J12}.

As for $\psi(d)$, it can be found from sub-Fig.~\ref{subfig:psi12} that as $r$ decreases, $\psi(d)$ will strictly go up, coinciding with our intuition. However, the curves of $\psi(d)$ are not always smooth and there are many turning points, roughly corresponding to the jump points of $J_1$, $J_{2,1}$, and $J_{2,2}$. By intuition, there should be $\psi(2)>\psi(1)$. However, surprisingly, we find that $\psi(2)>\psi(1)$ does not hold always, which is counterintuitive.

\begin{corollary}[An Exception]
	Let $r_0\approx0.8114$ be the root of the equation $2^{-r/2}=2^r-1$. Then at least for $r_0\leq r<1$, we have $\psi(1)>\psi(2)$.
\end{corollary}
\begin{proof}
	For $r_0\leq r<1$, it can be found from sub-Fig.~\ref{subfig:J12} that $J_1=J_{2,2}=1$ and $J_{2,1}\leq 0$, so we have
	\begin{align}
		\psi(1) = 1-(1-2^{-r})2^{r} = 1-(2^r-1) = 1-x\nonumber
	\end{align}	
	and 
	\begin{align}
		\psi(2) = \frac{\psi(2|\underline{10})}{2} 
		&= \frac{1}{2}\sum_{k=1}^{\kappa_2(1)}{(1-(1-2^{-r})2^{r}(2^{kr}-1))}\nonumber\\
		&= \frac{1-(1-2^{-r})2^{r}(2^{r}-1)}{2} = \frac{1-(2^{r}-1)^2}{2} = \frac{1-x^2}{2},\nonumber
	\end{align}
	where $x\triangleq(2^r-1)\in[0,1]$. The ratio between $\psi(2)$ and $\psi(1)$ is
	\begin{align}
		\frac{\psi(2)}{\psi(1)} = \frac{1+x}{2} = 2^{r-1} \leq 1,\nonumber
	\end{align}
	 and the equality holds iff $x=r=1$.
\end{proof}

\begin{figure*}[!t]
	\centering
	\subfigure[]{\includegraphics[width=.5\linewidth]{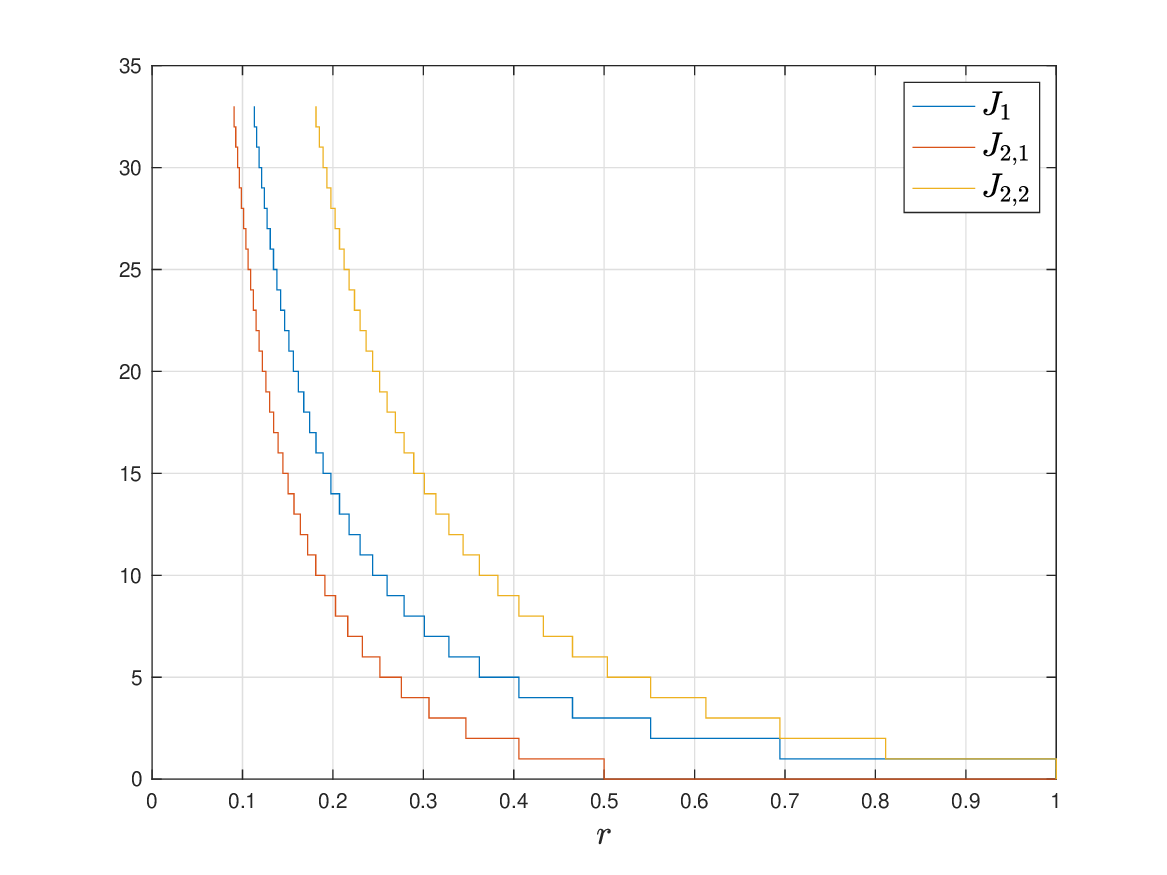}\label{subfig:J12}}%
	\subfigure[]{\includegraphics[width=.5\linewidth]{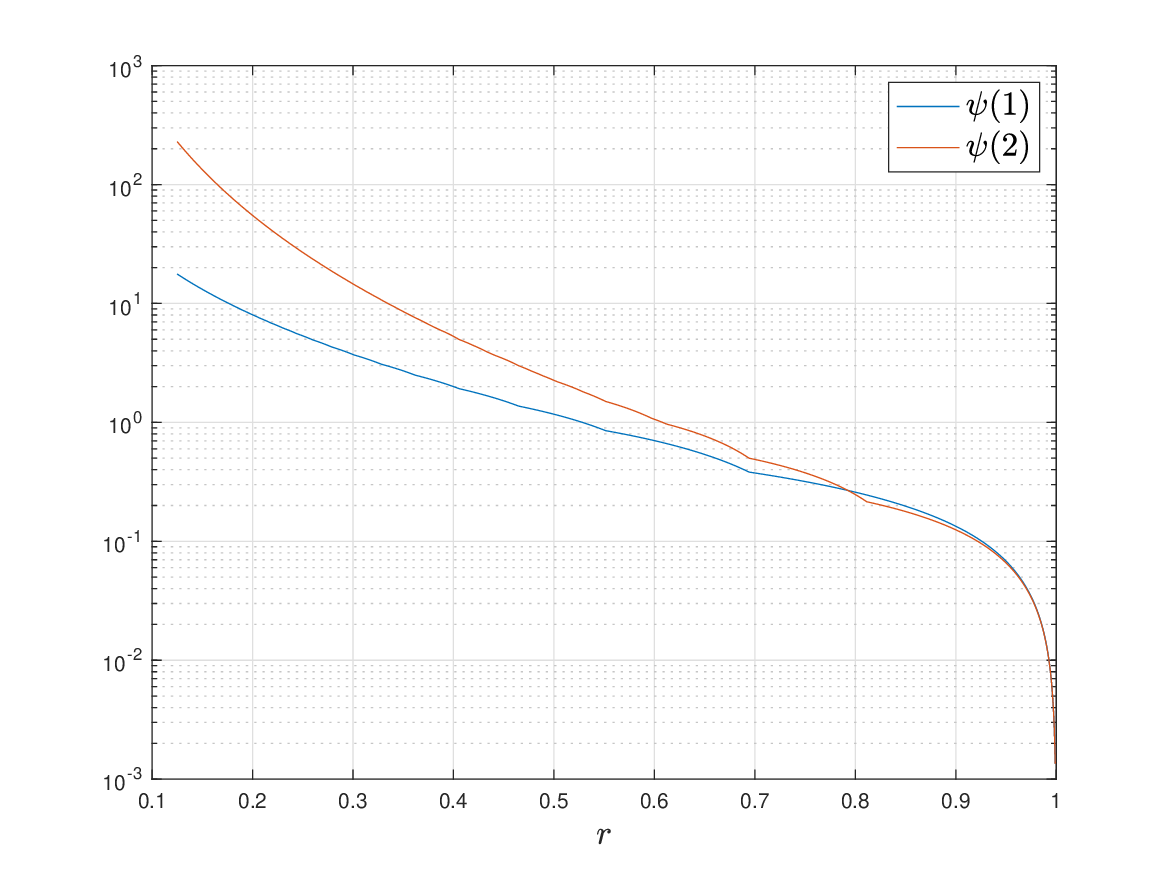}\label{subfig:psi12}}
	\caption{(a) Variations of $J_1$, $J_{2,1}$, and $J_{2,2}$ w.r.t. overlapping factor $r$, where $J_1$, $J_{2,1}$, and $J_{2,2}$ are defined by \eqref{eq:J1},  \eqref{eq:J21}, and \eqref{eq:J22}, respectively. Note that $J_{2,1}$ is lower bounded by $0$ because negative $J_{2,1}$ makes no sense. (b) Variations of $\psi(1)$ and $\psi(2)$ w.r.t. $r$, which are given by Corollary~\ref{corol:psi1} and Corollary~\ref{corol:psi2}, respectively.}
	\label{fig:hds12}
\end{figure*}

In theory, if only $\psi(d)<\infty$, it can be calculated following the same methodology developed for $\psi(1)$ and $\psi(2)$. However, as shown in Appendix~\ref{prf:psi2}, the procedure will be more and more complex. What's worse, $\psi(d)=\infty$ often happens for $d\geq 3$, as observed in \cite{FangTCOM16a}. The following theorem gives the necessary and sufficient condition for the convergence of $\psi(3)$, which is one of main contributions of this paper.
 
\begin{theorem}[Necessary and Sufficient Condition for the Convergence of $\psi(3)$]\label{thm:psi3}
	If there is no pair of integers $i\geq1$ and $j\geq1$ such that $2^{ir}(2^{jr}-1)=1$, then $\psi(3)<\infty$; otherwise, $\psi(3)=\infty$.
\end{theorem}
\begin{proof}
	See Appendix~\ref{prf:psi3} for the proof.
\end{proof}

\begin{corollary}[Sufficient Condition for the Convergence of $\psi(3)$]
	If $2^r$ is a transcendental number, $\psi(3)<\infty$.
\end{corollary}
\begin{proof}
	This is a direct result of Theorem~\ref{thm:psi3}.
\end{proof}

We have given the necessary and sufficient condition for the convergence of $\psi(3)$. Similarly, it is possible to deduce the necessary and sufficient condition for the convergence of $\psi(d)$ for $d>3$. However, the procedure will become more and more complex, so we would like to stop here. Actually, we strongly believe that if $2^r$ is a transcendental number, then $\psi(d)<\infty$ for any $d$, not just for $d=3$. However, we are not able to provide a strict proof at present, so we remain it as future work.

\begin{figure*}[!t]
	\centering
	\subfigure[]{\includegraphics[width=.5\linewidth]{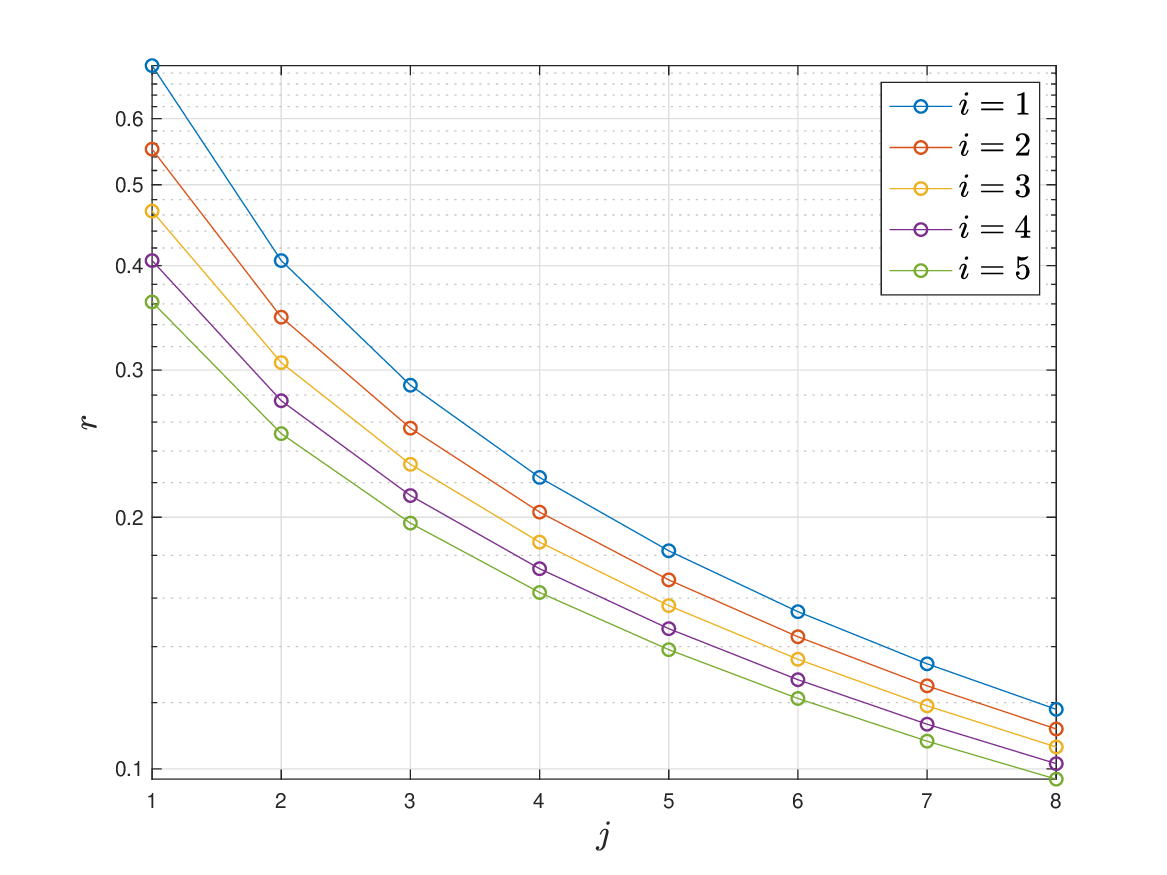}\label{subfig:d3a}}%
	\subfigure[]{\includegraphics[width=.5\linewidth]{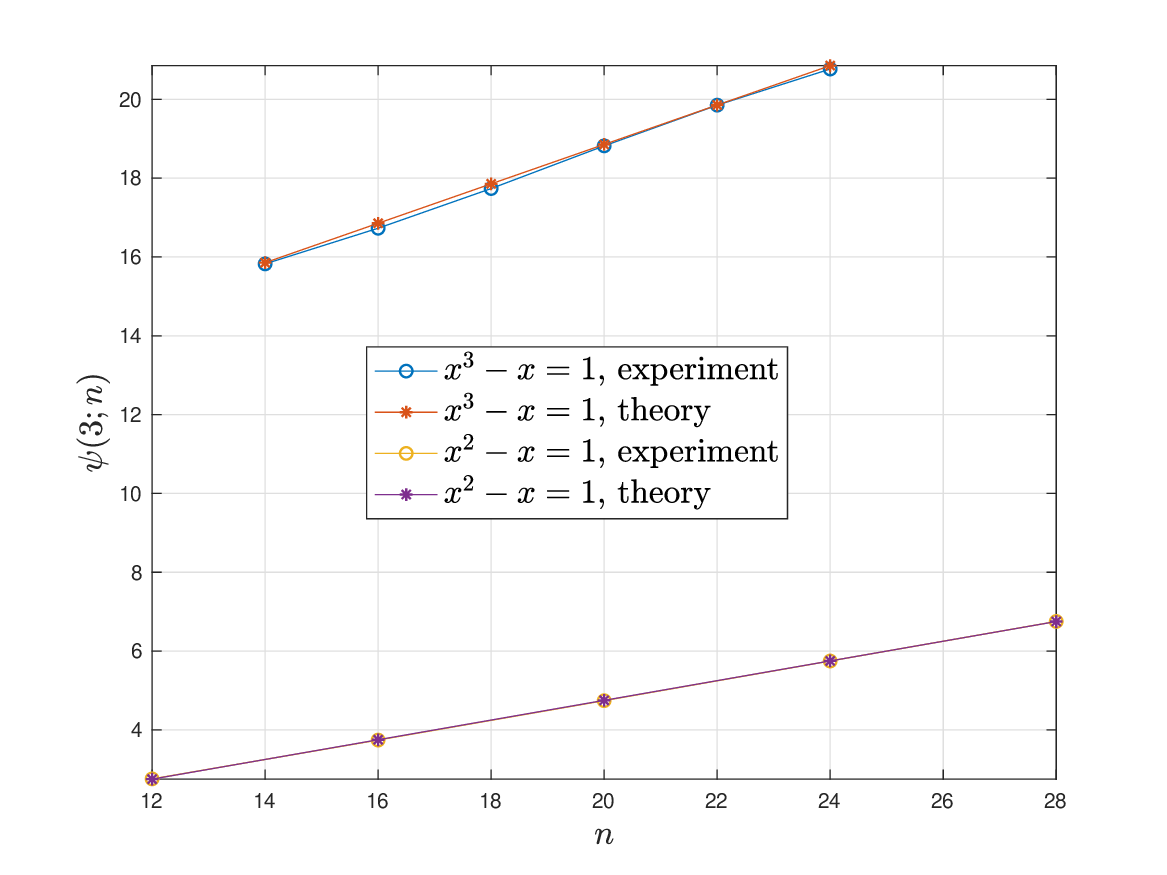}\label{subfig:d3b}}
	\caption{(a) Each circle corresponds to a tuple of $(i,j,r)$ satisfying $2^{ir}(2^{jr}-1)=1$, where $i,j\in\mathbb{Z}$ and $r\in[0,1]$. Note that for some specific $r$, there may be more than one $(i,j)$-pairs satisfying $2^{ir}(2^{jr}-1)=1$. For example, if $2^r(2^{2r}-1)=1$, then $2^{4r}(2^r-1)=1$. (b) Theoretical and experimental results of $\psi(3;n)$ for two divergent settings of $x=2^r$.}
	\label{fig:d3}
\end{figure*}

\subsection{Analytical Form of Divergent $\psi(3;n)$}\label{subsec:psi3}
For every pair of $i\geq1$ and $j\geq1$, there must be an overlapping factor $r\in[0,1]$ such that $2^{ir}(2^{jr}-1)=1$. Some examples of $(i,j)$-pairs and the corresponding overlapping factors satisfying $2^{ir}(2^{jr}-1)=1$ are given in sub-Fig.~\ref{subfig:d3a}. It can be seen that when $2^r$ is the golden ratio, $i=j=1$ will make $2^{ir}(2^{jr}-1)=1$. Actually, this is the largest $r$ that results in $\psi(3)=\infty$, as proved by the following lemma. As $i$ or $j$ increases, the overlapping factor making $2^{ir}(2^{jr}-1)=1$ will be smaller. 

\begin{lemma}[Sign of $(2^{j_3r}-2^{j_2r}-2^{j_1r})$]\label{lem:signum}
	Let $\varphi$ denote the golden ratio. Depending on the relation between $2^r$ and $\varphi$, the value of $(2^{j_3r}-2^{j_2r}-2^{j_1r})$ will have different signs:
 	\begin{itemize}
		\item $2^r=\varphi$: We have $(2^{j_3r}-2^{j_2r}-2^{j_1r})\geq0$ and the equality holds iff $(j_3-j_2)=(j_2-j_1)=1$; 
		\item $2^r>\varphi$: We have $(2^{j_3r}-2^{j_2r}-2^{j_1r})>0$ for every $1\leq j_1<j_2<j_3$ and hence $\psi(3)<\infty$; 
		\item $2^r<\varphi$: The value of $(2^{j_3r}-2^{j_2r}-2^{j_1r})$ can be positive or negative, and for some special $2^r$, may be $0$.
	\end{itemize}
\end{lemma}
\begin{proof}
	We can write $(2^{j_3r}-2^{j_2r}-2^{j_1r})$ as
	\begin{align}
		(2^{j_3r}-2^{j_2r}-2^{j_1r}) 
		&= 2^{j_1r}(2^{(j_3-j_1)r}-2^{(j_2-j_1)r}-1)\nonumber\\
		&= 2^{j_1r}(2^{(j_2-j_1)r}(2^{((j_3-j_1)-(j_2-j_1))r}-1)-1)\nonumber\\
		&= 2^{ir}(2^{jr}(2^{kr}-1)-1),\nonumber
	\end{align}
	where $i=j_1\geq1$, $j=(j_2-j_1)\geq1$, and $k=((j_3-j_1)-(j_2-j_1))=(j_3-j_2)\geq1$. Clearly, $2^{jr}(2^{kr}-1)>0$ and it is strictly increasing w.r.t. both $j$ and $k$. Thus $2^{jr}(2^{kr}-1)\geq2^r(2^r-1)$. After solving $2^r(2^r-1)=1$, we obtain $2^r=\varphi\approx1.618$ and $r\approx 0.6942$. Therefore, if $2^r=\varphi$, then
	\begin{align}
		(2^{j_3r}-2^{j_2r}-2^{j_1r}) = 2^{ir}(2^{jr}(2^{kr}-1)-1) \geq (2^{jr}(2^{kr}-1)-1) \geq (2^r(2^r-1)-1) = 0,\nonumber
	\end{align}
	where the equality holds iff $k=(j_3-j_2)=j=(j_2-j_1)=1$.
	
	Since $2^r(2^r-1)$ is strictly increasing w.r.t. $2^r$ and $r$, if $2^r>\varphi$, then $(2^{j_3r}-2^{j_2r}-2^{j_1r})\geq (2^r(2^r-1)-1)>0$ for any $1\leq j_1<j_2<j_3$ and consequently $\psi(3)<\infty$.
	
	The third bullet point holds obviously.
\end{proof}

Now we consider the following interesting problem: Is it possible to find the analytical form of $\psi(3;n)$ even if $\psi(3;n)$ does not converge? Below we first give the general form of $\psi(3;n)$ when it does not converge, and then derive the concrete forms in two special cases. It was observed in Fig.~5(b) of \cite{FangTCOM16a} that when $2^r$ is the golden ratio, $\psi(3;n)$ will increase linearly w.r.t. $n$. This is not a coincidence, as proved by the following theorem. 

\begin{theorem}[Linearity of $\psi(3;n)$ for Large $n$]\label{thm:psi3linear}
	If there exist one or more pairs of integers $i\geq1$ and $j\geq1$ such that $2^{ir}(2^{jr}-1)=1$, then for large $n$, $\psi(3;n)$ will increase linearly. Let ${\cal P}$ denote the set of all pairs of integers $i\geq1$ and $j\geq1$ satisfying $2^{ir}(2^{jr}-1)=1$. Every $(i,j)\in{\cal P}$ will cause a divergent term $(n-(i+j))$, and
	\begin{align}\label{eq:psi3ngeneral}
		\psi(3;n) \approx c_0 + (1/4)\sum_{(i,j)\in{\cal P}}{(n-(i+j))} = c_1 + n|{\cal P}|/4,
	\end{align}	
	where $c_0$ is the sum of convergent terms, $|{\cal P}|$ is the cardinality of ${\cal P}$, and $c_1=c_0-\frac{1}{4}\sum_{(i,j)\in{\cal P}}{(i+j)}$.
\end{theorem}
\begin{proof}
	See Appendix~\ref{prf:psi3linear} for the proof.
\end{proof}

As shown by \eqref{eq:psi3ngeneral}, to obtain $\psi(3;n)$ for a specific $r$, the key is to determine $c_0$, the sum of convergent terms. Following are two examples to show how to calculate $c_0$ for a specific $r$. In the first example, the set ${\cal P}$ includes only one $(i,j)$-pair, while in the second example, the set ${\cal P}$ includes two $(i,j)$-pairs. 
\begin{corollary}[$\psi(3;n)$ for Golden Ratio]\label{corol:psi3}
	If $(2^{2r}-2^r-1)=0$, then $\psi(3;n)\approx (n-1)/4$ for $n\geq5$. 
\end{corollary}
\begin{proof}
	See Appendix~\ref{prf:psi3golden} for the proof.
\end{proof}

Continue our discussion. As shown by the third bullet point of Lemma~\ref{lem:signum}, if $2^r$ is smaller than the golden ratio, then $(2^{j_3r}-2^{j_2r}-2^{j_1r})$ will have zero-crossing points. An interesting thing is that for some special overlapping factors, there may be more than one pairs of $i\geq1$ and $j\geq1$ such that $2^{ir}(2^{jr}-1)=1$. For example, in sub-Fig.~\ref{subfig:d3a}, we find an overlapping factor $r\approx0.4057$ that results in $2^r(2^{2r}-1)=2^{4r}(2^r-1)=1$. To prove this point, let us define $x\triangleq 2^r\in(1,2)$ for convenience. Then $2^r(2^{2r}-1)=x(x^2-1)=1$, which is followed by $\alpha\triangleq(x^3-x-1)=0$. We call $\alpha$ as the \textbf{zero element} just as in finite field. With polynomial division,
\begin{align}
	2^{4r}(2^r-1) = x^4(x-1) = (x^5-x^4)\bmod\alpha = 1.\nonumber
\end{align}
The following corollary gives the closed form of $\psi(3;n)$ for $\alpha\triangleq(x^3-x-1)=0$.

\begin{corollary}[$\psi(3;n)$ for $(x^3-x-1)=0$]\label{corol:psi3b}
	Let $x\triangleq 2^r\in(1,2)$. If $\alpha\triangleq(x^3-x-1)=0$, then for $n\geq14$,
	\begin{align}
		\psi(3;n)\approx \frac{-12x^2-17x+79}{4} + n/2.
	\end{align}
\end{corollary}
\begin{proof}
	See Appendix~\ref{prf:psi3b} for the proof.
\end{proof}

With the above two examples, we show that the analytical form of $\psi(3;n)$ is calculable for every $n$ even though it does not converge. Of course, following the methodology developed in Appendix~\ref{prf:psi3b}, the reader can also derive the mathematical expression of $\psi(3;n)$ in other divergent cases. However, the procedure is very complex and there is nothing exciting, so we would like to stop here.

To verify Corollary~\ref{corol:psi3} and Corollary~\ref{corol:psi3b}, some results are given in sub-Fig.~\ref{subfig:d3b}. It can be seen that theoretical curves almost coincide with experimental curves, perfectly confirming the correctness of these two corollaries.

\subsection{Propagation of Divergence}
Finally, we would like to end this section with an interesting phenomenon. If for some overlapping factor, there are two or more pairs of $i\geq1$ and $j\geq1$ such that $2^{ir}(2^{jr}-1)=1$, then $\psi(4;n)$ and $\psi(6;n)$ may not converge. For example, if
$x^3-x-1=x^5-x^4-1=0$, then 
\begin{align}
	x^2(x^5-x^4-1)+(x^3-x-1) = x^7-x^6+x^3-x^2-x-1=0,\nonumber
\end{align}
which will cause a divergent term $(n-7)$ and further make $\psi(6;n)$ divergent. Meanwhile, we have
\begin{align}
	(x^5-x^4)-(x^3-x) = x^5-x^4-x^3+x = x(x^4-x^3-x^2+1) = 0,\nonumber
\end{align}
and thus $(x^4-x^3-x^2+1)=0$, which will cause a divergent term $(n-4)$ and further make $\psi(4;n)$ divergent. By repeatedly combining the above zero elements, we can even observe that $\psi(d;n)$ is divergent for other $d$. We name this phenomenon as the \textit{Propagation of Divergence}. This topic is amazing but also very difficult meanwhile, so we remain it as future work.

\section{Hard Approximation of HDS for $d\gg 1$}\label{sec:hdsn}
As we know, Corollary~\ref{corol:softhds} is an approximate version of Theorem~\ref{thm:hds} for finite $n$. Just as analyzed at the end of sub-Sect.~\ref{subsec:hdscal}, only for small $d$, we can use Corollary~\ref{corol:softhds} to calculate $\psi(d;n)$. There are two reasons:
\begin{itemize}
	\item \textbf{Accurateness}: As shown in Appendix~\ref{prf:hds}, the proof of Theorem~\ref{thm:hds} is based on the assumption of $(n-d)\to\infty$, so it seems that Corollary~\ref{corol:softhds} is inaccurate for $d\approx n$. 
	\item \textbf{Complexity}: The computing complexity of $\psi(d;n)$ by Corollary~\ref{corol:softhds} is ${\cal O}(2^d\binom{n}{d})$, going up very quickly as $d$ increases, so this complexity is unacceptable for $d\gg 1$. 
\end{itemize}
This section will investigate how to calculate $\psi(d;n)$ for large $d$. We propose a variant of \eqref{eq:softhds} to calculate $\psi(d;n)$, which is called \textit{Hard Approximation}, just to distinguish it from the \textit{Soft Approximation} defined by \eqref{eq:softhds}. In essence, the {\em Hard Approximation} is an approximation of the {\em Soft Approximation}. According to the deduction of \textit{Hard Approximation}, both {\em Hard Approximation} and {\em Soft Approximation} are very accurate even for $d\approx n$. However, there is an exception when $d=n$, which will be particularly discussed. 

Note that though the {\em Hard Approximation} developed in this section looks very simple and straightforward, it will lay a foundation for the next section, which is the core of this paper to bridge HDS with CCS. 
 
\subsection{Discussion for $1\ll d<n$}
\begin{definition}[Active Set]
	The active set associated with $j^d\in{\cal J}_{n,d}$ is defined as
	\begin{align}
		{\cal B}_{j^d} \triangleq \{b^d: b^d\in\mathbb{B}^d{\rm~and~}|\tau(j^d,b^d)|<1\}.
	\end{align}
\end{definition}

\begin{lemma}[Asymptotic Conditional Probability of Coexisting Interval]
	We define $E{(j^d,b^d)}$ as \eqref{eq:E}. Then
	\begin{align}\label{eq:Econd}
		\lim_{(n-d)\to\infty}\Pr\left\{E{(j^d,b^d)}\in{\frak I}(j^d,b^d) \middle| b^d\in{\cal B}_{j^d}\right\} = \left(1-|\tau(j^d,b^d)|\right).
	\end{align}
\end{lemma}
\begin{proof}
	This lemma is an immediate result of Theorem~\ref{thm:prob}.
\end{proof}

Now for $j^d\in{\cal J}_{n,d}$, we define such a sequence:
\begin{align}\label{eq:omegajdb}
	\omega_{j^d,{\cal B}} \triangleq \left(\tau(j^d,b^d)\right)_{b^d\in{\cal B}_{j^d}}.
\end{align}
According to the above definition, the sequence $\omega_{j^d,{\cal B}}$ includes $|{\cal B}_{j^d}|$ terms, where $|{\cal B}_{j^d}|$ is the cardinality of ${\cal B}_{j^d}$. Obviously, the terms of $\omega_{j^d,{\cal B}}$ are drawn from the interval $(-1,1)$. In other words, $\omega_{j^d,{\cal B}}\in(-1,1)^{|{\cal B}_{j^d}|}$. The following lemma holds naturally. 

\begin{lemma}[Conditional Distribution of Shift Function]\label{lem:mod1sft}
If the sequence $(2^r,2^{2r},\dots)$ satisfies Wilms' condition, then the sequence $\omega_{j^d,{\cal B}}$ will be u.d. over $(-1,1)$ as $d\to\infty$. Therefore, for $1\ll d<n$,
\begin{align}\label{eq:sumBjd}
	\frac{\sum_{b^d\in{\cal B}_{j^d}}{\Pr\left\{E{(j^d,b^d)}\in{\frak I}(j^d,b^d)\right\}}}{|{\cal B}_{j^d}|}
	\approx \frac{\sum_{b^d\in{\cal B}_{j^d}}{\left(1-|\tau(j^d,b^d)|\right)}}{|{\cal B}_{j^d}|} \approx 1/2,
\end{align}
where $|{\cal B}_{j^d}| = \sum_{b^d\in\mathbb{B}^d}{{\bf 1}_{|\tau(j^d,b^d)|<1}}$ is the cardinality of ${\cal B}_{j^d}$. 
\end{lemma}

About Lemma~\ref{lem:mod1sft}, we would like to remind the reader about the following two points. 
\begin{itemize}
	\item As $d\to\infty$, the left hand side of \ref{eq:sumBjd} will converge to $1/2$. For small $d$, there may be a large difference between the left hand side of \ref{eq:sumBjd} and $1/2$. Hence $d$ should be large enough in practice to make \eqref{eq:sumBjd} hold.
	\item As evidenced by \eqref{eq:Evar}, $E{(j^d,b^d)}$ is a random variable for $d<n$ but will degenerate into a deterministic constant when $d=n$. Hence \eqref{eq:sumBjd} cannot be applied to the case of $d=n$. 
\end{itemize}

\begin{lemma}[Average Conditional Probability of Coexisting Interval]
	Given $1\ll d<n$, for any $j^d\in{\cal J}_{n,d}$, if the sequence $(2^r,2^{2r},\dots)$ satisfies Wilms' condition, then
	\begin{align}\label{eq:approxhalf}
		\Pr\left\{\ell(X^n)\in{\frak I}(j^d,X_{j^d})\middle| |\tau(j^d,X_{j^d})|<1 \right\} \approx 1/2.
	\end{align}	
\end{lemma}
\begin{proof}
	Since $E{(j^d,b^d)}=\ell(X^n|X_{j^d}=b^d)$, we have
	\begin{align}
		&\Pr\left\{\ell(X^n)\in{\frak I}(j^d,X_{j^d})\middle| |\tau(j^d,X_{j^d})|<1\right\}\nonumber\\
		&\qquad= \sum_{b^d\in{\cal B}_{j^d}}\underbrace{\Pr(X_{j^d}=b^d)}_{=1/|{\cal B}_{j^d}|}\cdot
		\Pr\left\{E{(j^d,b^d)}\in{\frak I}(j^d,b^d)\right\}
		\nonumber\\
		&\qquad= \frac{1}{|{\cal B}_{j^d}|}\sum_{b^d\in{\cal B}_{j^d}}\Pr\left\{E{(j^d,b^d)}\in{\frak I}(j^d,b^d)\right\} \stackrel{(a)}{\approx} 1/2,\nonumber
	\end{align}
	where $|{\cal B}_{j^d}|$ is the cardinality of ${\cal B}_{j^d}$ and $(a)$ comes from Lemma~\ref{lem:mod1sft}.
\end{proof}

\begin{theorem}[Hard Approximation of HDS for $1\ll d<n$]\label{thm:th3a}
	If the sequence $(2^r,2^{2r},\dots)$ satisfies Wilms' condition, then for $1\ll d<n$,
	\begin{align}\label{eq:hardhds}
		\psi(d;n)
		\approx 2^{-(d+1)}\sum_{j^d\in{\cal J}_{n,d}}|{\cal B}_{j^d}| 
		= 2^{-(d+1)} \sum_{j^d\in{\cal J}_{n,d}} \sum_{b^d\in\mathbb{B}^d} {{\bf 1}_{|\tau(j^d,b^d)|<1}}.
	\end{align}	
\end{theorem}
\begin{proof}
	The proof includes two folds. The first fold discusses the case of $1\ll d\ll n$. In this case, Corollary~\ref{corol:softhds} holds. According to \eqref{eq:softhds}, we have
	\begin{align}\label{eq:softhdsvar}
		\psi(d;n) \approx 2^{-d}\sum_{j^d\in{\cal J}_{n,d}}\sum_{b^d\in{\cal B}_{j^d}}\left(1-|\tau(j^d,b^d)|\right).
	\end{align}
	If the sequence $(2^r,2^{2r},\dots)$ satisfies Wilms' condition, then the sequence $\omega_{j^d,{\cal B}}$ is u.d. over $(-1,1)$ as $d\to\infty$. On substituting \eqref{eq:sumBjd} into \eqref{eq:softhdsvar}, we will obtain \eqref{eq:hardhds}.
	
	The second fold discusses the case of $1\ll d\approx n$ but $d<n$. In this case, since $(n-d)<\infty$, we cannot apply Corollary~\ref{corol:softhds} directly. Instead, according to the definition of $E{(j^d,b^d)}$ and ${\frak I}(j^d,b^d)$, we have
	\begin{align}\label{eq:psidnhard}
		\psi(d;n) \approx 2^{-d}\sum_{j^d\in{\cal J}_{n,d}}\sum_{b^d\in{\cal B}_{j^d}}\Pr\left\{E{(j^d,b^d)}\in{\frak I}(j^d,b^d)\right\}.
	\end{align}
	On substituting \eqref{eq:sumBjd} into \eqref{eq:psidnhard}, we will obtain \eqref{eq:hardhds}.
\end{proof}

\begin{remark}[Comparison of Theorem~\ref{thm:th3a} with Corollary~\ref{corol:softhds}]
	It can be seen that \eqref{eq:softhds} and \eqref{eq:hardhds} are very similar to each other, and \eqref{eq:hardhds} can be taken as a binary approximation of \eqref{eq:softhds}. We would like to highlight the following points.
	\begin{itemize}
		\item They have the same complexity ${\cal O}(2^d\binom{n}{d})$.
		\item Though Corollary~\ref{corol:softhds} requires $d\ll n$, it is easy to know from \eqref{eq:sumBjd} that disregarding the issue of complexity, \eqref{eq:softhds} is accurate for every $d\in[1:n)$. Actually, as $d$ increases, both \eqref{eq:softhds} and \eqref{eq:hardhds} will approach each other gradually, according to the law of large numbers. 
		\item Both \eqref{eq:softhds} and \eqref{eq:hardhds} do not apply to $d=n$ because $E{(j^d,b^d)}$ will degenerate into a deterministic constant. 
		\item For small $d$, the accurateness of \eqref{eq:hardhds} is poor because \eqref{eq:sumBjd}, the cornerstone for the proof of \eqref{eq:hardhds}, does not hold; while \eqref{eq:softhds} is always accurate no matter for small $d$ or large $d$ (except $d=n$).
	\end{itemize}
\end{remark}

\subsection{Discussion for $d=n$}
As analyzed before, neither \eqref{eq:softhds} nor \eqref{eq:hardhds} applies to $d=n$ because $E{(j^d,b^d)}$ will degenerate into a deterministic constant. Now we want to see what will happen when $d=n$.

\begin{lemma}[$n$-away Codeword Pairs]\label{lem:naway}
	Let $n$ be the length of an overlapped arithmetic code. If a pair of $n$-away codewords coexist in the same coset, then this pair of codewords must belong to the $2^{nr-1}$-th coset. Conversely speaking, the $2^{nr-1}$-th coset includes all pairs of $n$-away codewords.
\end{lemma}
\begin{proof}
	According to the definition of $\tau$-function,
	\begin{align}
		\tau(j^n,x^n)
		&= \underbrace{(1-2^{-r})\sum_{i=1}^{n}{2^{ir}}}_{2^{nr}-1} - 2\underbrace{(1-2^{-r})\sum_{i=1}^{n}{x_i2^{ir}}}_{\ell(x^n)}\nonumber\\
		&= (2^{nr}-1) - 2\ell(x^n),\nonumber
	\end{align}
	which is followed by
	\begin{align}
		\ell(x^n) = 2^{nr-1}-\frac{1+\tau(j^n,x^n)}{2}.\nonumber
	\end{align}
	For $y^n=x^n\oplus1^n$, we have
	\begin{align}
		\ell(y^n) = \ell(x^n\oplus1^n) = \ell(x^n) + \tau(j^n,x^n) = 2^{nr-1} - \frac{1-\tau(j^n,x^n)}{2}.\nonumber
	\end{align}
	If $|\tau(j^n,x^n)|<1$, then $0<\frac{1\pm\tau(j^n,x^n)}{2}<1$ and
	\begin{align}
		\begin{cases}
			(2^{nr-1}-1)<\ell(x^n)<2^{nr-1}\\
			(2^{nr-1}-1)<\ell(y^n)<2^{nr-1}	
		\end{cases}.\nonumber
	\end{align} 
	That means, given $|\tau(j^n,x^n)|<1$, we have $\lceil\ell(x^n)\rceil\equiv\lceil\ell(y^n)\rceil\equiv2^{nr-1}$. In other words, $x^n$ and $y^n=x^n\oplus1^n$ must coexist in the $2^{nr-1}$-th coset.
\end{proof}

\begin{lemma}[Particularity of $d=n$]\label{corol:d=n}
	The necessary and sufficient condition for the event that $x^n$ and $(x^n\oplus1^n)$ coexist in the same coset is $|\tau(j^n,x^n)|<1$.
\end{lemma}
\begin{proof}
	This lemma is a direct result of Lemma~\ref{lem:naway}.
\end{proof}

In order to better understand Lemma~\ref{corol:d=n}, let us recall Corollary~\ref{corol:neccoe}. According to Corollary~\ref{corol:neccoe}, $|\tau(j^d,b^d)|<1$ is the necessary condition for the coexistence of $x^n$ and $y^n=x^n\oplus z^n$ in the same coset, where $x_{j^d}=b^d$, $z_{j^d}=1^d$, and $z_{[n]\setminus j^d}=0^{n-d}$. Especially, Lemma~\ref{corol:d=n} says that, if $d=n$, then the necessary condition $|\tau(j^n,b^n)|<1$ is also the sufficient condition for the coexistence of $x^n=b^n$ and $y^n=x^n\oplus 1^n$ in the same coset.

\begin{theorem}[Calculation of $\psi(n;n)$]\label{thm:th3b}
	Let $j^n=\{1,\dots,n\}$. Then
	\begin{align}\label{eq:psinapprox}
		\psi(n;n) 
		= 2^{-n}\cdot|{\cal B}_{j^n}| 
		= 2^{-n}\sum_{b^n\in\mathbb{B}^n}{{\bf 1}_{|\tau(j^n,b^n)|<1}}.
	\end{align}	
\end{theorem}
\begin{proof}
	There is only one sequence $j^n=\{1,\dots,n\}$ in the set ${\cal J}_{n,n}$. For every $b^n\in\mathbb{B}^n$, if $|\tau(j^n,b^n)|<1$, then both $b^n$ and $b^n\oplus1^n$ coexist in the same coset. Hence, \eqref{eq:psinapprox} holds naturally.
\end{proof}

Let us compare \eqref{eq:psinapprox} with \eqref{eq:hardhds}. Following \eqref{eq:hardhds}, we have
\begin{align}
	\psi(n;n) \approx 
	2^{-(n+1)}\cdot|{\cal B}_{j^n}| = 
	2^{-(n+1)}\sum_{b^n\in\mathbb{B}^n}{{\bf 1}_{|\tau(j^n,b^n)|<1}},
\end{align}
which is just half of \eqref{eq:psinapprox}. To find why this phenomenon happens, let us recall the definition of $E(j^d,b^d)$. According to \eqref{eq:E}, if $d<n$, $E(j^d,b^d)$ is a random variable; however for $d=n$, $E(j^n,b^n)$ will degenerate into a deterministic constant. Accordingly, \eqref{eq:hardhds} does not apply to the case of $d=n$. In addition, also please notice another subtle difference between \eqref{eq:psinapprox} and \eqref{eq:hardhds}. That is, \eqref{eq:hardhds} gives only an approximate value of $\psi(d;n)$ for $d<n$, while \eqref{eq:psinapprox} gives an exact value of $\psi(n;n)$.

For clarity, we integrate Theorem~\ref{thm:th3a} and Theorem~\ref{thm:th3b} into the following theorem.
\begin{theorem}[Hard Approximation of HDS]\label{thm:hard}
	Let $\alpha={\bf 1}_{(d=n)}$. For $1\ll d\leq n$, we have
	\begin{align}\label{eq:psith3}
		\psi(d;n) \approxeq 2^{\alpha-d-1}\sum_{j^d\in{\cal J}_{n,d}}\sum_{b^d\in\mathbb{B}^d}{{\bf 1}_{|\tau(j^d,b^d)|<1}},
	\end{align}
	where the approximation will become equality if $d=n$.	
\end{theorem}

Meanwhile, by taking the case of $d=n$ into consideration, we will get a variant of Corollary~\ref{corol:softhds} as below.
\begin{theorem}[Soft Approximation of HDS]\label{thm:soft}
	Let $\alpha={\bf 1}_{(d=n)}$. For $1\leq d\leq n$, we have
	\begin{align}\label{eq:psith2}
		\psi(d;n) \approx 2^{\alpha-d}\sum_{j^d\in{\cal J}_{n,d}}\sum_{b^d\in\mathbb{B}^d}{\left(1-|\tau(j^d,b^d)|\right)^+}.
	\end{align}	
\end{theorem}

\section{Fast Approximation of HDS for $d\approx n$}\label{sec:hdsfast}
As show by Theorem~\ref{thm:soft} and Theorem~\ref{thm:hard}, for both {\em Hard Approximation} and {\em Soft Approximation}, the complexity is too high to be acceptable for large $d$. This section will derive a fast method to calculate $\psi(d;n)$ for $d\approx n$ based on the close affinity between CCS and HDS, which is named as \textit{Fast Approximation}, whose complexity is ${\cal O}(1)$, the same as that of the {\em Binomial Approximation} defined by \eqref{eq:coarsehds}. Through this work, we bridge HDS with CCS, which is the most important contribution of this paper.

\subsection{Normalized Shift Function}
By observing \eqref{eq:psith3}, it can be found that to calculate $\psi(d;n)$, we should know the number of codewords making the shift function $\tau(j^d,b^d)$ fall into the interval $(-1,1)$. Actually, \eqref{eq:psith3} suggests to calculate $\psi(d;n)$ with exhaustive enumeration, whose complexity is ${\cal O}(2^d\binom{n}{d})$, which is very high for large $d$. In the following, instead of exhaustive enumeration, we try to find a simple method to calculate $\psi(d;n)$ for large $d$. Our analysis is based on the close affinity between CCS and HDS. More concretely, for every $d\leq n$, we can derive the asymptotic distribution of $\tau(j^d,b^d)$ according to CCS. This is a very interesting new finding.

We have defined the sequence $\omega_{j^d,{\cal B}}$ by \eqref{eq:omegajdb}. Further, we define 
\begin{align}\label{eq:omegadb}
	\omega_{d,{\cal B}} \triangleq \left(\tau(j^d,b^d)\right)_{j^d\in{\cal J}_{n,d},b^d\in{\cal B}_{j^d}}.
\end{align}
According to the above definition, the sequence $\omega_{d,{\cal B}}$ includes $\sum_{j^d\in{\cal J}_{n,d}}|{\cal B}_{j^d}|$ terms and every term is drawn from the interval $(-1,1)$. Actually, the sequence $\omega_{j^d,{\cal B}}$ defined by \eqref{eq:omegajdb} is a sub-sequence of the sequence $\omega_{d,{\cal B}}$ defined by \eqref{eq:omegadb}. If the sequence $(2^r,2^{2r},\dots)$ satisfies Wilms' condition, then the sequence $\omega_{d,{\cal B}}$ should also be u.d. over the interval $(-1,1)$ as $\sum_{j^d\in{\cal J}_{n,d}}|{\cal B}_{j^d}|\to\infty$. In addition, we define two more sequences:
\begin{align}\label{eq:omegajd}
	\omega_{j^d} \triangleq \left(\tau(j^d,b^d)\right)_{b^d\in\mathbb{B}^d},
\end{align} 
which denotes a sequence including $2^d$ elements, and
\begin{align}\label{eq:omegad}
	\omega_d \triangleq \left(\tau(j^d,b^d)\right)_{j^d\in{\cal J}_{n,d},b^d\in\mathbb{B}^d},
\end{align} 
which denotes a sequence including $2^d\binom{n}{d}$ elements. Obviously, $\omega_{j^d}$ is a sub-sequence of $\omega_d$. After comparing \eqref{eq:omegad} and \eqref{eq:omegajd} with \eqref{eq:omegadb} and \eqref{eq:omegajdb}, it can be seen that $\omega_{j^d,{\cal B}}$ is a sub-sequence of $\omega_{j^d}$, while $\omega_{d,{\cal B}}$ is a sub-sequence of $\omega_d$. Before deriving the distribution of $\omega_d$, we should know the distribution of $\omega_{j^d}$. To derive the distribution of $\omega_{j^d}$, let us define a random variable $\tau(j^d,X^d)$ according to \eqref{eq:tau}. Obviously, for uniform binary sources, the distribution of the random variable $\tau(j^d,X^d)$ is just the distribution of the sequence $\omega_{j^d}$. 

Note that $\tau(j^d,X^d)$ is discretely distributed over the interval $(-2^{nr},2^{nr})$. As $n\to\infty$, the shift function $\tau(j^d,X^d)$ will not be well defined because the interval $(-2^{nr},2^{nr})$ will become the real field $\mathbb{R}$. Therefore, it is very difficult to derive the distribution of $\tau(j^d,X^d)$ directly. To overcome this difficulty, for a given $j^d\in{\cal J}_{n,d}$, according to \eqref{eq:tau}, we define a normalized random variable for $\tau(j^d,X^d)$.
\begin{definition}[Normalized Shift Function]
	For every $j^d\in{\cal J}_{n,d}$, we define the normalized shift function as
\begin{align}
	W(j^d) \triangleq 2^{-nr}\tau(j^d,X^d) = 2^{-nr}c(j^d) - 2V(j^d),
\end{align}
where $c(j^d)$ is defined by \eqref{eq:tauvar} and
\begin{align}\label{eq:Vjd}
	V(j^d) \triangleq 2^{-nr}(1-2^{-r})\sum_{d'=1}^{d}{X_{d'}2^{rj_{d'}}}.
\end{align}
\end{definition}
It can be seen that $c(j^d)$ is a constant, while $V(j^d)$ is a random variable for $d\geq1$. Clearly, $0\leq 2^{nr}V(j^d) \leq c(j^d)$ and $-c(j^d)\leq 2^{nr}W(j^d) \leq c(j^d)$. Since $0\leq c(j^d)<2^{nr}$, we have $0\leq 2^{-nr}c(j^d)<1$. Therefore, $V(j^d)$ is defined over $[0,1)$ and $W(j^d)$ is defined over $(-1,1)$. 
\begin{remark}[Comparison of $V$ with $\Upsilon$]
	After comparing \eqref{eq:Evar} with \eqref{eq:Vjd}, it can be seen that $V(j^d)$ and $\Upsilon([n]\setminus j^d)$ are similar to each other. However, $\Upsilon([n]\setminus j^d)$ is defined over $[0,2^{nr})$, while $V(j^d)$ is defined over $[0,1)$. In addition, $\Upsilon([n]\setminus j^d)$ is the sum of $(n-d)$ terms, while $V(j^d)$ is the sum of $d$ terms.
\end{remark}

Let $f_{W|j^d}(w)$, where $-1<w<1$, be the pdf of $W(j^d)$, and $f_{V|j^d}(v)$, where $0\leq v<1$, be the pdf of $V(j^d)$. According to the property of pdf, it is easy to obtain
\begin{align}\label{eq:wvpdf}
	f_{W|j^d}(w) = \tfrac{1}{2}f_{V|j^d}(\tfrac{2^{-nr}c(j^d)-w}{2}).
\end{align}
Both $V(j^d)$ and $W(j^d)$ are tractable because they are well defined. Once $f_{W|j^d}(w)$ is obtained, the distribution of $\omega_{j^d}$ can be easily derived. Further, we define
\begin{align}\label{eq:fWd}
	f_{W|d}(w) \triangleq \frac{\sum_{j^d\in{\cal J}_{n,d}}{f_{W|j^d}(w)}}{\binom{n}{d}}.
\end{align}
Once $f_{W|d}(w)$ is obtained, the distribution of $\omega_{d}$ can be easily derived. 

\subsection{Distribution of Shift Function}
Let us begin with $d=n$. Since there is only one sequence $j^n=\{1,\dots,n\}$ in ${\cal J}_{n,n}$, we abbreviate $V(j^n)$ to $V$, and $W(j^n)$ to $W$ for conciseness. Consequently, $f_{W|j^n}(w)$ is shortened to $f_W(w)$, and $f_{V|j^n}(v)$ to $f_V(v)$.
\begin{theorem}[Asymptotic Distribution of Normalized Shift Function for $d=n$]\label{thm:wn}
	As $n\to\infty$, the pdf of $V$ is $f_V(v)=f(v)$, where $f(u)$ is the asymptotic CCS defined by \eqref{eq:asympccs}, and the pdf of $W$ is 
	\begin{align}\label{eq:fWfV}
		f_W(w)=f_V(\tfrac{1-w}{2})/2 = f(\tfrac{1-w}{2})/2.
	\end{align}
\end{theorem}
\begin{proof}
	For $j^n=\{1,\dots,n\}\in{\cal J}_{n,n}$, we have 
	\begin{align}\label{eq:V}
		V = (1-2^{-r})\sum_{i=1}^{n}{X_i2^{(i-n)r}} \simeq (2^r-1)\sum_{i=1}^{n}{X_i2^{-ir}}.
	\end{align}
	Surprisingly, after comparing \eqref{eq:V} with \eqref{eq:U0infty}, we find $\lim_{n\to\infty}{V} = U_{0,\infty}$. Therefore, as $n\to\infty$, the pdf of $V$ will have exactly the same shape as the asymptotic CCS $f(u)$ defined by \eqref{eq:asympccs}. On knowing the pdf of $V$, the pdf of $W$ can be easily obtained from \eqref{eq:wvpdf}. As we know $c(j^n)=2^{nr}-1$, it is easy to see $\lim_{n\to\infty}2^{-nr}c(j^n)=1$. Thus as $n\to\infty$, we have $f_W(w) = f(\frac{1-w}{2})/2$. 
\end{proof}

We proceed to the case of $d=(n-1)$. There are $n$ sequences in ${\cal J}_{n,n-1}$, and for every sequence, the pdf of $W(j^{n-1})$ is different. We have the following theorem for this issue.
\begin{theorem}[Asymptotic Distribution of Normalized Shift Function for $d=(n-1)$]\label{thm:wn1}
	Let
	\begin{align}\label{eq:jn1}
		j^{n-1}=\{1,\dots,n-k,n-k+2,\dots,n\}\in {\cal J}_{n,n-1}, 
	\end{align}
	where $1\leq k\leq n$. Let $f(u)$ be the asymptotic CCS defined by \eqref{eq:asympccs}. Then
	\begin{itemize}
		\item For $k<\infty$, we have
		\begin{align}\label{eq:Vjn1}
			\lim_{n\to\infty}f_{V|j^{n-1}}(v)
			= 2^{kr-(k-1)}\sum_{x^{k-1}\in\mathbb{B}^{k-1}}f((v-l(x^{k-1}))2^{kr})
		\end{align}
		and
		\begin{align}
			\lim_{n\to\infty}f_{W|j^{n-1}}(w)
			= 2^{-k(1-r)}\sum_{x^{k-1}\in\mathbb{B}^{k-1}}f((\tfrac{1-(2^r-1)2^{-kr}-w}{2}-l(x^{k-1}))2^{kr}),
		\end{align}
		where $l(X^i)$ is defined by \eqref{eq:lXi}. 
		
		\item As $k\to\infty$, we have $\lim_{n\to\infty}f_{V|j^{n-1}}(v)=f(v)$ and $\lim_{n\to\infty}f_{W|j^{n-1}}(w)=f(\frac{1-w}{2})/2$.  
	\end{itemize}
\end{theorem}
\begin{proof}
	See Appendix~\ref{prf:wn1} for the proof.
\end{proof}

\begin{corollary}[A Simple Relation Between $f_V(v)$ and $f(u)$]
	Let $f(u)$ be the asymptotic CCS defined by \eqref{eq:asympccs}. Let $j^{n-1}=(1,\dots,n-k,n-k+2,\dots,n)\in {\cal J}_{n,n-1}$, where $1\leq k\leq n$. Then
	\begin{align}\label{eq:fvfVv}
		f(v) = \lim_{n\to\infty}\left(f_{V|j^{n-1}}(v) + f_{V|j^{n-1}}(v-(2^r-1)2^{-kr})\right)/2.
	\end{align}
\end{corollary}
\begin{proof}
	According to \eqref{eq:futrivial} and the properties of pdf, we have
	\begin{align}
		f(u) = 2^{-k}\sum_{x^k\in\mathbb{B}^k} 2^{kr}f((u-l(x^k))2^{kr}).\nonumber
	\end{align}
	We can write $l(x^k)=l(x^{k-1})+x_k(2^r-1)2^{-kr}$. Hence, 
	\begin{align}
		f(u) 
		= 2^{kr-(k-1)}\sum_{x^{k-1}\in\mathbb{B}^{k-1}}\left(f((u-l(x^{k-1}))2^{kr}) +
		f((u-l(x^{k-1})-(2^r-1)2^{-kr})2^{kr})\right)/2.\nonumber
	\end{align}
	Then according to \eqref{eq:Vjn1}, it is easy to see that \eqref{eq:fvfVv} holds obviously.
\end{proof}

Finally, let us discuss $f_{W|(n-1)}(w)$, $-1<w<1$, which is defined by \eqref{eq:fWd}, where $d=(n-1)$. The following corollary gives an interesting result about $f_{W|(n-1)}(w)$.
\begin{corollary}[Asymptotic Form of $f_{W|{(n-1)}}(w)$]\label{corol:fWn1w}
	As $n\to\infty$, we have $f_{W|(n-1)}(w)=f(\frac{1-w}{2})$, where $f(u)$ is the asymptotic CCS defined by \eqref{eq:asympccs}. In other words, for large $n$, we have $f_{W|(n-1)}(w)\approx f(\frac{1-w}{2})$.
\end{corollary}
\begin{proof}
	There are $n$ sequences in ${\cal J}_{n,n-1}$, and each sequence can be written as $j^{n-1}=\{1,\dots,n-k,n-k+2,\dots,n\}\in {\cal J}_{n,n-1}$, where $1\leq k\leq n$. As stated by Theorem~\ref{thm:wn1}, as $k\to\infty$, we have $f_{W|j^{n-1}}(w)=f(\frac{1-w}{2})/2$. According to the definition \eqref{eq:fWd}, $f_{W|(n-1)}(w)$ is actually the average of $n$ functions, and these functions will converge to $f(\frac{1-w}{2})/2$. Hence, this corollary holds obviously.
\end{proof}

So far, we have solved the asymptotic distribution problem of $W(j^d)$ for $d=n$ and $(n-1)$. It is certain that the above methodology can be easily extended to the general case of $d\approx n$. There are $\binom{n}{d}$ sequences in ${\cal J}_{n,d}$. For each sequence $j^d\in{\cal J}_{n,d}$, we calculate the asymptotic pdf of $V(j^d)$ and then derive the asymptotic pdf of $W(j^d)$. Since the procedure is very complex and boring, we would like to stop here for this issue. However, it deserves being spotted out that Corollary~\ref{corol:fWn1w} can be extended to the general case $d\approx n$.
\begin{corollary}[Asymptotic Form of $f_{W|d}(w)$]\label{corol:fWd}
	Given $(n-d)<\infty$, as $n\to\infty$, we have $f_{W|d}(w)=f(\frac{1-w}{2})$, where $f(u)$ is the asymptotic CCS defined by \eqref{eq:asympccs}. In other words, for $d\approx n$, we have $f_{W|d}(w)\approx f(\frac{1-w}{2})$.
\end{corollary}
\begin{proof}
	Without loss of generality, any $j^{d}\in{\cal J}_{n,d}$ can be written as
	\begin{align}
		j^d = (1,\dots,n-k_{n-d},n-k_{n-d}+2,\dots,n-k_2,n-k_2+2,\dots,n-k_1,n-k_1+2,\dots,n),
	\end{align}
	where $k_1<\dots<k_{n-d}$. Then
	\begin{align}
		V(j^d) 
		&\simeq \underbrace{(2^r-1)\sum_{i=1}^{k_1-1}{X_i2^{-ir}}}_{V_1} + \underbrace{(2^r-1)\sum_{i=k_1+1}^{k_2-1}{X_i2^{-ir}}}_{V_2} + \dots + \underbrace{(2^r-1)\sum_{i=k_{n-d}+1}^{n}{X_i2^{-ir}}}_{V_{n-d+1}}\nonumber\\
		&= V_1+V_2+\dots+V_{n-d+1}.\nonumber
	\end{align}
	As $k_1\to\infty$, we have $V(j^d)=V_1\simeq U_{0,\infty}$ and $V_2=\dots=V_{n-d+1}=0$. Therefore, if $k_1=\infty$, then $f_{W|j^d}(w)=f(\frac{1-w}{2})/2$. As we know, $f_{W|d}(w)$ is the average of $\binom{n}{d}=\binom{n}{n-d}$ functions, and these functions will converge to $f(\frac{1-w}{2})/2$. Hence, this corollary holds obviously.
\end{proof}

\begin{corollary}[A Trend of $f_{W|d}(w)$]\label{corol:fWdtrend}
	Given $d\gg 1$, as $(n-d)\to\infty$, there will be a trend for $f_{W|d}(w)$ to become a Gaussian function centered at $0$.
\end{corollary}
\begin{proof}
	According to \eqref{eq:fWd}, $f_{W|d}(w)$ is the average of $\binom{n}{d}$ functions. As we know, $f_{W|j^d}(w)$ is the pdf of $W(j^d)$, and in turn $W(j^d)$ is a function of $d$ random variables $X_{j^d}$, which are drawn from $X^n$. In general, as $d$ decreases, the correlation between different $W(j^d)$'s will be weaker. Hence, according to the central limit theorem, as $(n-d)\to\infty$, $f_{W|d}(w)$ will tend to be a Gaussian function centered at $0$.
\end{proof}

\subsection{Bridging HDS with CCS}
Once knowing the pdf of $W(j^d)$, we can easily derive how the sequence $\omega_{j^d}$ is distributed over $(-2^{nr},2^{nr})$. The following lemma answers how many elements in the sequence $\omega_{j^d}$ fall into the interval $(-1,1)$.
\begin{lemma}[Cardinality of $j^d$-Active Set]\label{lem:caractset}
	Given $d\approx n$, for $j^d\in{\cal J}_{n,d}$, the number of codewords $b^d\in\mathbb{B}^d$ making $|\tau(j^d,b^d)|<1$ is $|{\cal B}_{j^d}|\approx 2^{d+1-nr}f_{W|j^d}(0)$, where $f_{W|j^d}(w)$ is the pdf of $W(j^d)$.
\end{lemma}
\begin{proof}
	Every $j^d\in{\cal J}_{n,d}$ corresponds to $2^d$ codewords in total, and for every  $b^d\in\mathbb{B}^d$, the shift function $\tau(j^d,b^d)$ is distributed over $(-2^{nr},2^{nr})$. According to the definition of $W(j^d)$, we have
	\begin{align}
		|{\cal B}_{j^d}| = 2^d\int_{-2^{-nr}}^{2^{-nr}}{f_{W|j^d}(w)}\,dw 
		\stackrel{(a)}\approx 2^d\cdot 2^{1-nr} \cdot f_{W|j^d}(0) = 2^{d+1-nr}f_{W|j^d}(0),
	\end{align}
	where $(a)$ is because $f_{W|j^d}(w)$ tends to be uniform over $(-2^{-nr},2^{-nr})$ as $n\to\infty$. 
\end{proof}

In turn, after knowing $f_{W|d}(w)$, we can easily derive how the sequence $\omega_d$ is distributed over $(-2^{nr},2^{nr})$. The following lemma answers how many elements in the sequence $\omega_d$ fall into the interval $(-1,1)$.
\begin{lemma}[Cardinality of $d$-Active Set]\label{lem:caractsetd}
	For large $n$, given $d\approx n$, the number of codewords $b^d\in\mathbb{B}^d$ making $|\tau(j^d,b^d)|<1$ is  
	\begin{align}
		\sum_{j^d\in{\cal J}_{n,d}}{|{\cal B}_{j^d}|} \approx \binom{n}{d}2^{d+1-nr}f_{W|d}(0) \approx \binom{n}{d}2^{d-nr}f(1/2), 
	\end{align}	
	where $f_{W|d}(w)$ is defined by \eqref{eq:fWd} and $f(u)$ is the asymptotic CCS defined by \eqref{eq:asympccs}. 
\end{lemma}
\begin{proof}
	This lemma is a natural result of the definitions $f_{W|d}(w)$ and $\omega_d$.
\end{proof}

\begin{theorem}[Fast Approximation of HDS]\label{thm:fast}
	For $d\approx n$, we have 
	 \begin{align}\label{eq:fast}
		\psi(d;n) \approx \binom{n}{d}2^{\alpha-nr-1}f(1/2). 
	\end{align}
	where $\alpha = {\bf 1}_{(d=n)}$ and $f(u)$ is the asymptotic CCS defined by \eqref{eq:asympccs}.
\end{theorem}
\begin{proof}
	According to \eqref{eq:psinapprox}, 
	\begin{align}\label{eq:nn}
		\psi(n;n) \approx 2^{-n}\cdot|{\cal B}_{j^d}| \stackrel{(a)}\approx 2^{-n}\cdot2^{n+1-nr}f_W(0) = 2^{1-nr}f_W(0) = 2^{1-nr}f_V(1/2)/2 \stackrel{(b)}{=} 2^{-nr}f(1/2),
	\end{align}
	where $(a)$ comes from Lemma~\ref{lem:caractset} and $(b)$ comes from Theorem~\ref{thm:wn}. According to \eqref{eq:hardhds}, for $d<n$,
\begin{align}\label{eq:dn}
	\psi(d;n) 
	\approx 2^{-(d+1)}\sum_{j^d\in{\cal J}_{n,d}}{|{\cal B}_{j^d}|}
	\stackrel{(a)}{\approx} 2^{-(d+1)} \binom{n}{d}2^{d-nr}f(1/2) = \binom{n}{d}2^{-nr-1}f(1/2),
\end{align}	
where $(a)$ comes from Lemma~\ref{lem:caractsetd}. Combining \eqref{eq:nn} with \eqref{eq:dn}, we will obtain \eqref{eq:fast}. 	
\end{proof}

After comparing Theorem~\ref{thm:soft} and Theorem~\ref{thm:hard} with Theorem~\ref{thm:fast}, it can be found that the computing complexity of $\psi(d;n)$ is dramatically reduced from ${\cal O}(2^d\binom{n}{d})$ to ${\cal O}(1)$. 

\subsection{Summary on Approximate Formulas of HDS}
Finally, let us end this section with a summary on our proposed approximate formulas of $\psi(d;n)$. In Table~\ref{tab:summary}, there are four formulas, tagged with TH-1, TH-2, TH-3, and TH-4, respectively. The corresponding theorems of these formulas are given, together with their complexity. The contributions of this paper are also included. 

\begin{table*}[!t]
	\small\centering
	\caption{A Summary on Approximate Formulas of $\psi(d;n)$}
	\begin{tabular}{c||c||c||c||c||c}
		\hline
		Name &Formula &Work for &Complexity &cf. &Novelty\\
		\hline
		\hline
		TH-1/Binomial
		&$\displaystyle\binom{n}{d}\cdot 2^{-nr} \cdot \int_{0}^{1}{f^2(u)\,du}$
		&$d\approx n/2$ &${\cal O}(1)$ & Theorem~\ref{thm:coarsehds}  &strict proof\\
		\hline
		TH-2/Soft
		&$\displaystyle2^{\alpha-d} \sum_{b^d\in\mathbb{B}^d} \sum_{j^d\in{\cal J}_{n,d}}{\left(1-|\tau(j^d,b^d)|\right)^+}$
		&$1\leq d\leq n$ &${\cal O}(2^d\binom{n}{d})$ & Theorem~\ref{thm:soft} &strict proof\\
		\hline
		TH-3/Hard
		&$\displaystyle2^{\alpha-d-1}\sum_{b^d\in\mathbb{B}^d}\sum_{j^d\in{\cal J}_{n,d}}{{\bf 1}_{|\tau(j^d,b^d)|<1}}$
		&$d\gg 1$ &${\cal O}(2^d\binom{n}{d})$ & Theorem~\ref{thm:hard} &strict proof\\
		\hline
		TH-4/Fast
		&$\displaystyle\binom{n}{d}\cdot 2^{\alpha-nr-1}\cdot f(1/2)$ 
		&$d\approx n$ &${\cal O}(1)$ &Theorem~\ref{thm:fast} &brand new\\
		\hline
	\end{tabular}	
	\begin{tablenotes}
		\item If $d=n$, then $\alpha=1$; otherwise $\alpha=0$. 
	\end{tablenotes}
	\label{tab:summary}
\end{table*}

\section{Experimental Results}\label{sec:example}
We run experiments under the special setting of $n=20$ and $r=1/2$ to verify the analysis in Sect.~\ref{sec:hdsfast}, which is the main contribution of this paper. We choose $r=1/2$ because the asymptotic CCS $f(u)$ is given by \eqref{eq:closedForm_halfRate}. We choose $n=20$ because the complexity is too high for larger $n$. This section includes two parts. Part I will discuss the distribution of the shift function $\tau(j^d,b^d)$, and Part II will demonstrate $\psi(d;n)$ for $d\in[1:n]$.

\subsection{Distribution of Shift Function}\label{subsec:shift}
\subsubsection{Theoretical Results}
For conciseness, we give the theoretical results only for $d=n$ and $(n-1)$, while the theoretical results for any $d$ satisfying $d\approx n$ can be obtained in a similar way. For the case $d=n$, there is only one size-$n$ set $j^n=\{1,\dots,n\}\in{\cal J}_{n,n}$. According to {Theorem~\ref{thm:wn}}, we have $f_{V|n}(v) = f_{V|j^n} \approx f(v)$, where $f(u)$ is the asymptotic CCS given by \eqref{eq:closedForm_halfRate}, and 
\begin{align}\label{eq:fWn}
	f_{W|n}(w) = f_{W|j^n} \approx f(\tfrac{1-w}{2})/2 = 
	\begin{cases}
		\frac{1+w}{12\sqrt{2}-16}, 	&-1 \leq v < 2\sqrt{2}-3\\
		\frac{1}{4-2\sqrt{2}}, 		&2\sqrt{2}-3 \leq v < 3-2\sqrt{2}\\
		\frac{1-w}{12\sqrt{2}-16},	&3-2\sqrt{2} \leq v < 1%
	\end{cases}.
\end{align}
Then we discuss the case $d=(n-1)$. Though there are $\binom{n}{n-1}=n$ size-$(n-1)$ sets $j^{n-1}\in{\cal J}_{n,n-1}$, only two subcases are considered below for conciseness: 
\begin{itemize}
	\item If $j^{n-1}=\{1,\dots,n-1\}$, \textit{i.e.}, $k=1$ in {Theorem~\ref{thm:wn1}}, we have
	\begin{align}
		f_{V|j^{n-1}}(v) \approx \sqrt{2}f(\sqrt{2}v) = 
		\begin{cases}
			\frac{2v}{3\sqrt{2}-4}, 		&0 \leq v < 1-1/\sqrt{2}\\
			\frac{\sqrt{2}}{2-\sqrt{2}}, 	&1-1/\sqrt{2} \leq v < \sqrt{2}-1\\
			\frac{\sqrt{2}-2v}{3\sqrt{2}-4},&\sqrt{2}-1 \leq v < 1/\sqrt{2}%
		\end{cases}\nonumber
	\end{align}
	and
	\begin{align}\label{eq:fWn1k1}
		f_{W|j^{n-1}}(w) \approx 2^{-1/2}f(1/2-w/\sqrt{2}) = 
		\begin{cases}
			\frac{1/2+w/\sqrt{2}}{6-4\sqrt{2}}, 		&-1/\sqrt{2} \leq w < 2-3/\sqrt{2}\\
			\frac{1}{2\sqrt{2}-2}, 		&2-3/\sqrt{2} \leq w < 3/\sqrt{2}-2\\
			\frac{1/2-w/\sqrt{2}}{6-4\sqrt{2}},	&3/\sqrt{2}-2 \leq w < 1/\sqrt{2}%
		\end{cases}.
	\end{align}

	\item If $j^{n-1}=\{1,\dots,n-2,n\}$, \textit{i.e.}, $k=2$ in {Theorem~\ref{thm:wn1}}, it is easy to know $l(0)=0$ and $l(1)=(1-2^{-r})=(1-1/\sqrt{2})$. Hence, we have
	\begin{align}
		f_{V|j^{n-1}}(v) &\approx f(2v) + f(2(v-(1-1/\sqrt{2}))).\nonumber
	\end{align}
	After arrangement, we obtain
	\begin{align}
		f_{V|j^{n-1}}(v) &\approx 		
		\begin{cases}
			\frac{2v}{3\sqrt{2}-4}, 			&0 \leq v < (\sqrt{2}-1)/2\\
			\frac{1}{2-\sqrt{2}}, 				&(\sqrt{2}-1)/2 \leq v < (2-\sqrt{2})\\
			\frac{(3-\sqrt{2})-2v}{3\sqrt{2}-4},&(2-\sqrt{2}) \leq v < (3-\sqrt{2})/2%
		\end{cases}.\nonumber
	\end{align}
	As for $f_{W|j^{n-1}}(w)$, it is easy to get $1-(2^r-1)2^{-kr} = \frac{3-\sqrt{2}}{2}$. Therefore,
	\begin{align}
		f_{W|j^{n-1}}(w)
		&\approx (1/2)
		\left(f(\tfrac{3-\sqrt{2}}{2}-w) + f(\tfrac{3-\sqrt{2}}{2}-w-(2-\sqrt{2}))\right)\nonumber\\
		&= (1/2)
		\left(f(\tfrac{3-\sqrt{2}}{2}-w) + f(\tfrac{\sqrt{2}-1}{2}-w)\right).\nonumber
	\end{align}
	After arrangement, we obtain
	\begin{align}\label{eq:fWn1k2}
		f_{W|j^{n-1}}(w) \approx 
		\begin{cases}
			\frac{\tfrac{3-\sqrt{2}}{2}+w}{6\sqrt{2}-8}, & -\tfrac{3-\sqrt{2}}{2} < w\leq -\tfrac{5-3\sqrt{2}}{2}\\
			\frac{1}{4-2\sqrt{2}}, & -\tfrac{5-3\sqrt{2}}{2} \leq w \leq \tfrac{5-3\sqrt{2}}{2}\\
			\frac{\tfrac{3-\sqrt{2}}{2}-w}{6\sqrt{2}-8}, & \tfrac{5-3\sqrt{2}}{2} \leq w < \tfrac{3-\sqrt{2}}{2}
		\end{cases}.
	\end{align}
\end{itemize}
For other $j^{n-1}\in{\cal J}_{n,n-1}$, $f_{V|j^{n-1}}(v)$ and $f_{W|j^{n-1}}(w)$ can be derived according to {Theorem~\ref{thm:wn1}} in a similar way. However, the procedure will be more and more complex, so we stop here.

\subsubsection{Experimental Results}
To obtain experimental results, we run an exhaustive search, whose total complexity is ${\cal O}(3^n)$, where $3^n=\sum_{d=0}^{n}\binom{n}{d}2^d$. For every $j^d\in{\cal J}_{n,d}$ and every $b^n\in\mathbb{B}^d$, we define 
\begin{align}
	t(j^d,b^d)\triangleq{\rm sgn}(\tau(j^d,b^d)) \cdot \lceil|\tau(j^d,b^d)|\rceil,\nonumber
\end{align}
where ${\rm sgn}(\cdot)$ denotes the sign function. According to the properties of shift function, it is easy to know $t(j^d,b^d)\in(-2^{nr}:2^{nr})$. For a specific $j^d\in{\cal J}_{n,d}$, let $c(x;j^d)$, where $x\in(-2^{nr}:2^{nr})$, denote the number of $b^d\in\mathbb{B}^d$ such that $t(j^d,b^d)=x$. Obviously, $c(x;j^d)=c(-x;j^d)$ and $\sum_{x=-(2^{nr}-1)}^{2^{nr}-1}{c(x;j^d)}=2^d$, which can be rewritten as 
\begin{align}
	\sum_{x=-(2^{nr}-1)}^{2^{nr}-1}{c(x;j^d)2^{-d}} = 1.\nonumber
\end{align}
Now we need to build the connection between $c(x;j^d)$ for $x\in(-2^{nr}:2^{nr})$ and $f_{W|j^d}(w)$ for $w\in(-1,1)$. According to the definition of normalized shift function, we have
\begin{align}
	\sum_{x=-(2^{nr}-1)}^{2^{nr}-1}f_{W|j^d}(x2^{-nr})2^{-nr} \approx \int_{-1}^{1}f_{W|j^d}(w)\,dw=1.\nonumber
\end{align}
Therefore, $f_{W|j^d}(x2^{-nr})2^{-nr} \approx c(x;j^d)2^{-d}$ and hence for $x\in(-2^{nr}:2^{nr})$,
\begin{align}
	f_{W|j^d}(x2^{-nr})=f_{W|j^d}(-x2^{-nr})\approx c(x;j^d)2^{nr-d}.\nonumber
\end{align}
Further, we define $c(x;d)\triangleq\sum_{j^d\in{\cal J}_{n,d}}{c(x;j^d)}$. Obviously, $\sum_{x=-(2^{nr}-1)}^{2^{nr}-1}c(x;d)=\binom{n}{d}2^d$. Similarly, $f_{W|d}(w)$ for $w\in(-1,1)$ and $c(x;d)$ for $x\in(-2^{nr}:2^{nr})$ can be connected by
\begin{align}
	f_{W|d}(x2^{-nr}) = f_{W|d}(-x2^{-nr}) \approx \frac{c(x;d)2^{nr}}{\binom{n}{d}2^d}.\nonumber
\end{align}

\subsubsection{Comparison}
We plot the results for $n=20$ and $r=1/2$ in Fig.~\ref{fig:fW}. We do not try larger $n$ because the computing complexity is unacceptable. As we know, the complexity of a full search is ${\cal O}(3^n)$, going up sharply as $n$ increases. Let us check the correctness of theoretical analyses one by one. 
\begin{itemize}
	\item For $d=n$, we have $f_{W|n}(w)=f_{W|j^n}(w)=f(\frac{1-w}{2})/2$, where $j^n=\{1,\dots,n\}$ is the only subset in ${\cal J}_{n,n}$. The theoretical result of $f_{W|n}(w)$ (and also $f_{W|j^n}(w)$), which is given by \eqref{eq:fWn}, corresponds to the TH curve in sub-Fig.~\ref{subfig:fWn}, while the experimental result of $f_{W|n}(w)$ (and also $f_{W|j^n}(w)$) corresponds to the EXP curve in sub-Fig.~\ref{subfig:fWn}. It can be observed from sub-Fig.~\ref{subfig:fWn} that these two curves almost coincide with each other, strongly confirming the correctness of Theorem~\ref{thm:wn}. 
	
	\item For $d=(n-1)$, we have derived the theoretical results of $f_{W|j^{n-1}}(w)$ for $k=1$ and $k=2$ according to Theorem~\ref{thm:wn1}, as given by \eqref{eq:fWn1k1} and \eqref{eq:fWn1k2}, respectively. We compare the theoretical results of $f_{W|j^{n-1}}(w)$ for $k=1$ and $k=2$ with their experimental results in sub-Fig.~\ref{subfig:fWjn1_1} and sub-Fig.~\ref{subfig:fWjn1_2}, respectively. It can be observed that the theoretical curves almost coincide with the corresponding experimental curves, perfectly confirming the correctness of Theorem~\ref{thm:wn1}. 
	
	\item It is also declared by Theorem~\ref{thm:wn1} that $f_{W|j^{n-1}}(w)$ will converge to $f(\frac{1-w}{2})/2$, where $f(u)$ is the initial CCS, as $k$ increases. To confirm this prediction, we plot several curves of $f_{W|j^{n-1}}(w)$ for $k=4$, $8$, $12$, and $16$ in sub-Fig.~\ref{subfig:fWjn1}, where the TH curve is $f(\frac{1-w}{2})/2$. It can be seen that as $k$ increases, $f_{W|j^{n-1}}(w)$ does converge to the TH curve $f(\frac{1-w}{2})/2$, verifying the correctness of Theorem~\ref{thm:wn1}.
	
	\item We plot several curves of $f_{W|d}(w)$ for $d=(n-1)$, $(n-3)$, $(n-5)$, $(n-7)$, and $(n-9)$ in sub-Fig.~\ref{subfig:fWd}. It can be seen that as $d\to n$, $f_{W|d}(w)$ will approach $f(\frac{1-w}{2})/2$, as predicted by Corollary~\ref{corol:fWd}; while as $(n-d)$ increases, $f_{W|d}(w)$ will tend to be a bell-shaped function, as predicted by Corollary~\ref{corol:fWdtrend}.
	
	\item As shown by sub-Fig.~\ref{subfig:fWd}, $f_{W|d}(w)$ does not coincide with $f(\frac{1-w}{2})/2$ well for $(n-d)\ll d<n$, which can be attributed to small $n$. To verify this claim, we give the results for different code lengths in sub-Fig.~\ref{subfig:fWd_vs_n}. It can be seen that compared with the curves of $n=16$, the curves of $n=20$ are closer to the curve of $f(\frac{1-w}{2})/2$. Hence, we believe the correctness of Corollary~\ref{corol:fWd}, \textit{i.e.}, as $n$ increases, $f_{W|d}(w)$ for $(n-d)\ll d<n$ will converge to $f(\frac{1-w}{2})/2$.
\end{itemize}

\begin{figure*}[!t]
	\centering
	\subfigure[]{\includegraphics[width=.5\linewidth]{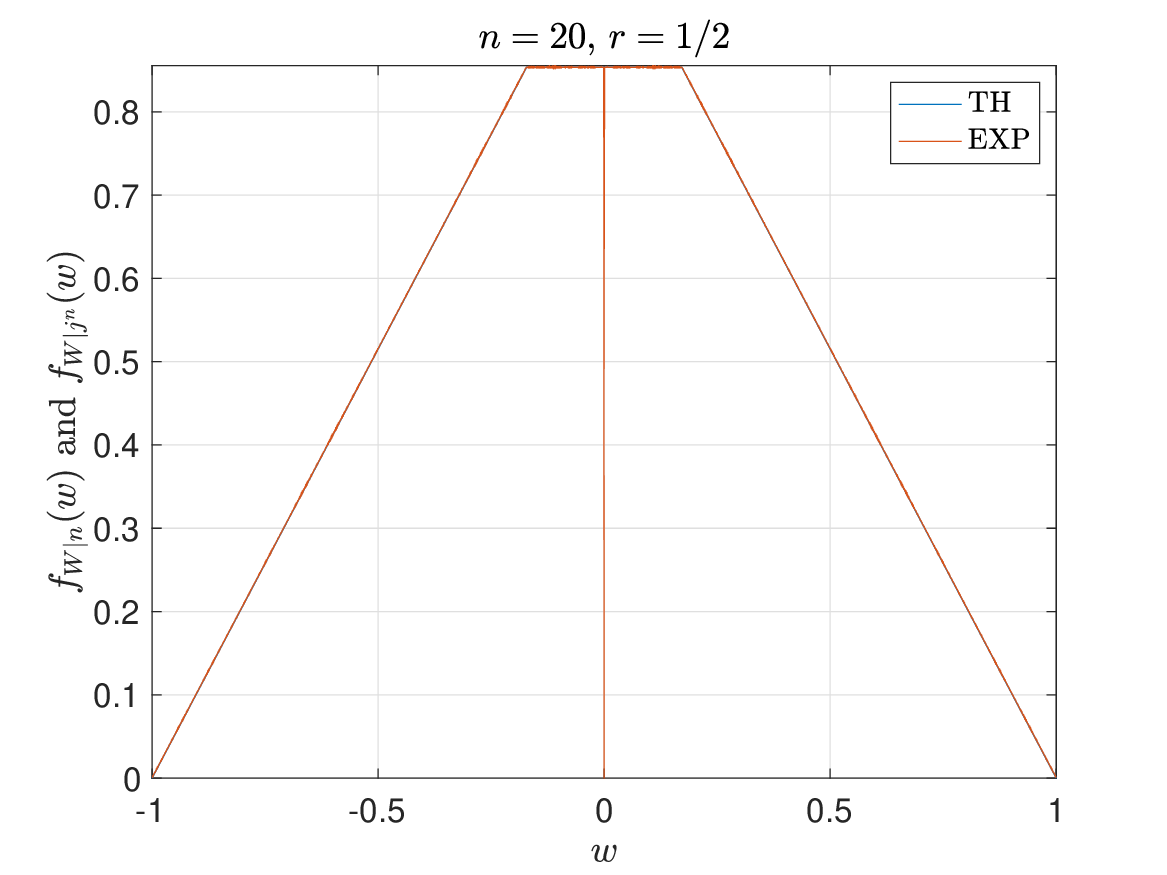}\label{subfig:fWn}}%
	\subfigure[]{\includegraphics[width=.5\linewidth]{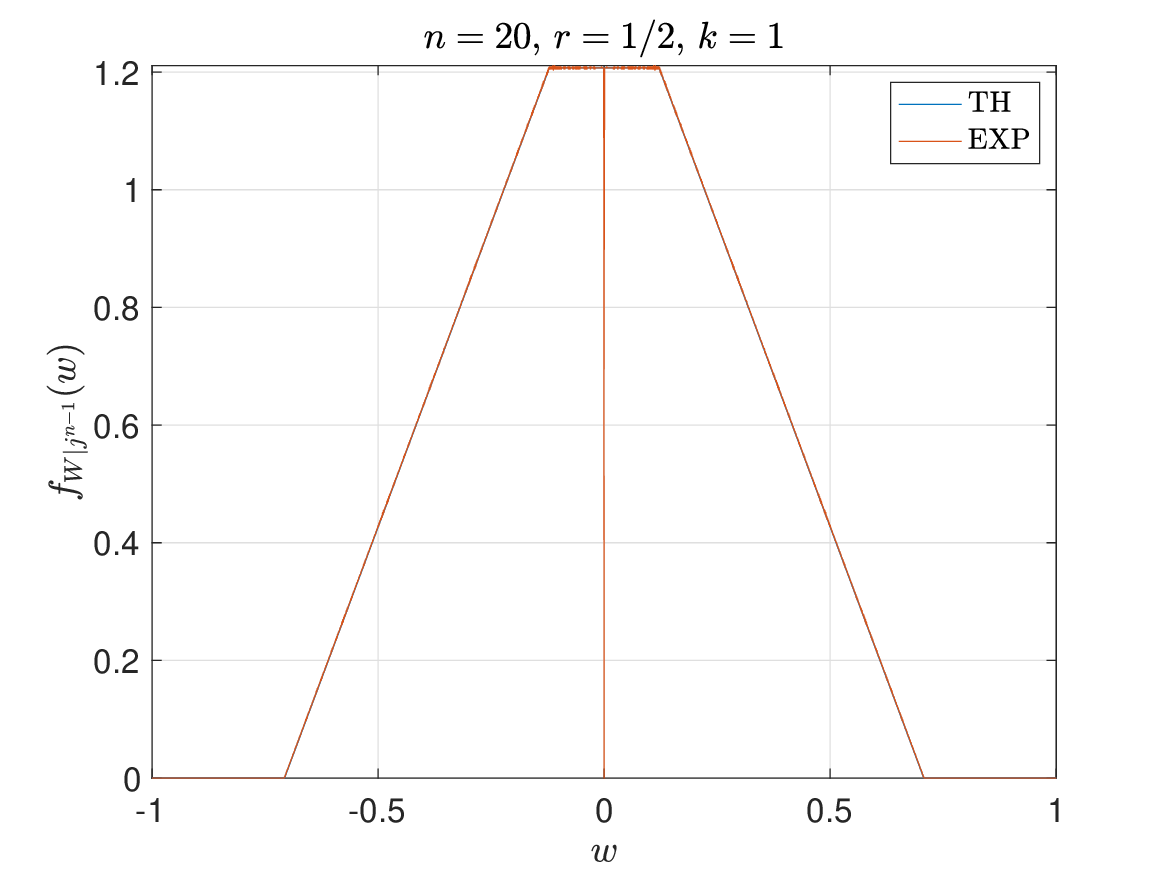}\label{subfig:fWjn1_1}}\\
	\subfigure[]{\includegraphics[width=.5\linewidth]{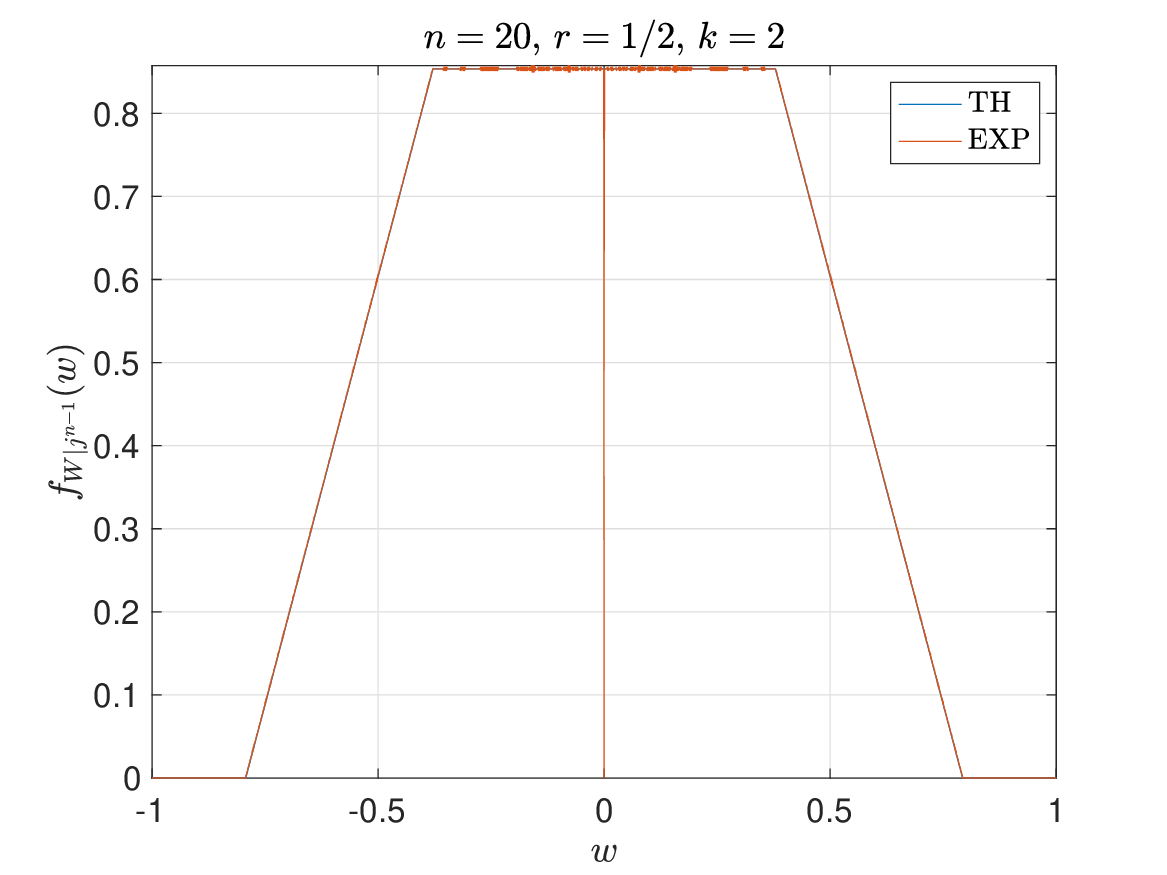}\label{subfig:fWjn1_2}}%
	\subfigure[]{\includegraphics[width=.5\linewidth]{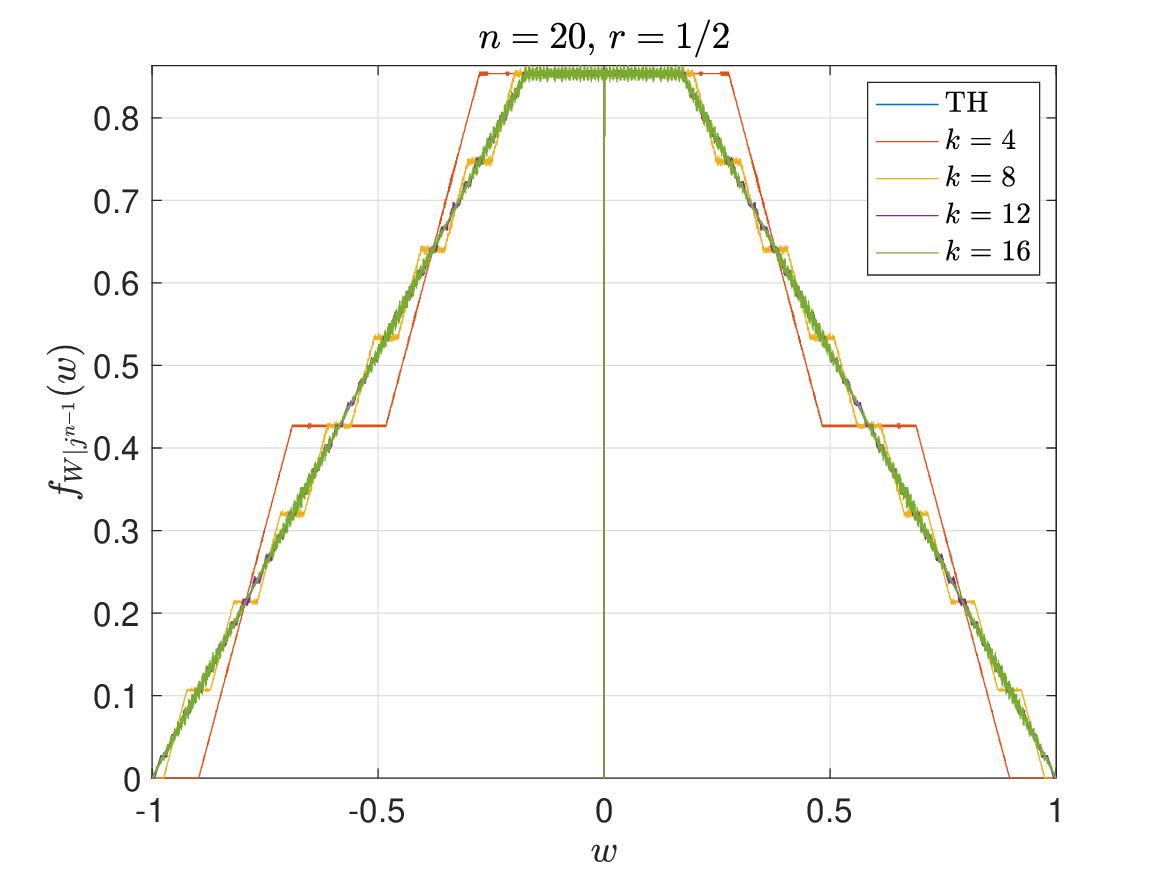}\label{subfig:fWjn1}}\\
	\subfigure[]{\includegraphics[width=.5\linewidth]{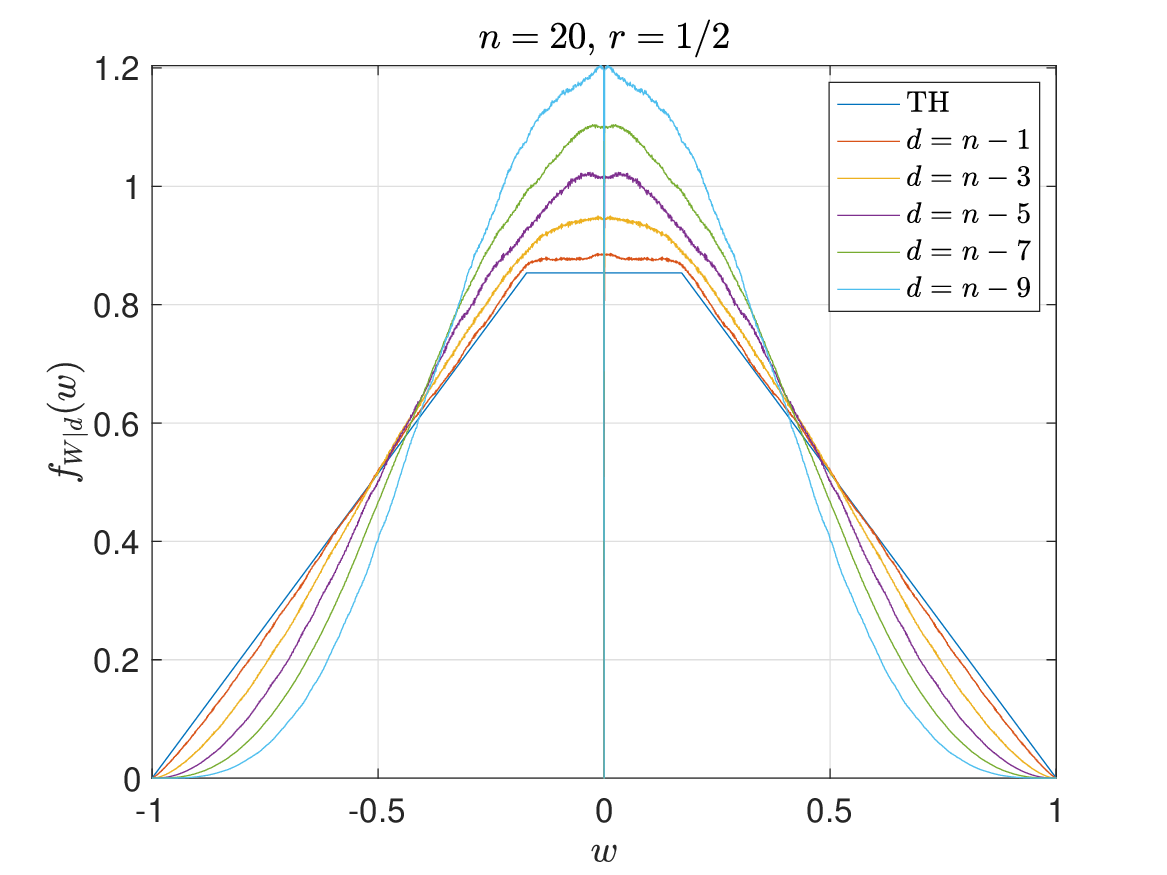}\label{subfig:fWd}}%
	\subfigure[]{\includegraphics[width=.5\linewidth]{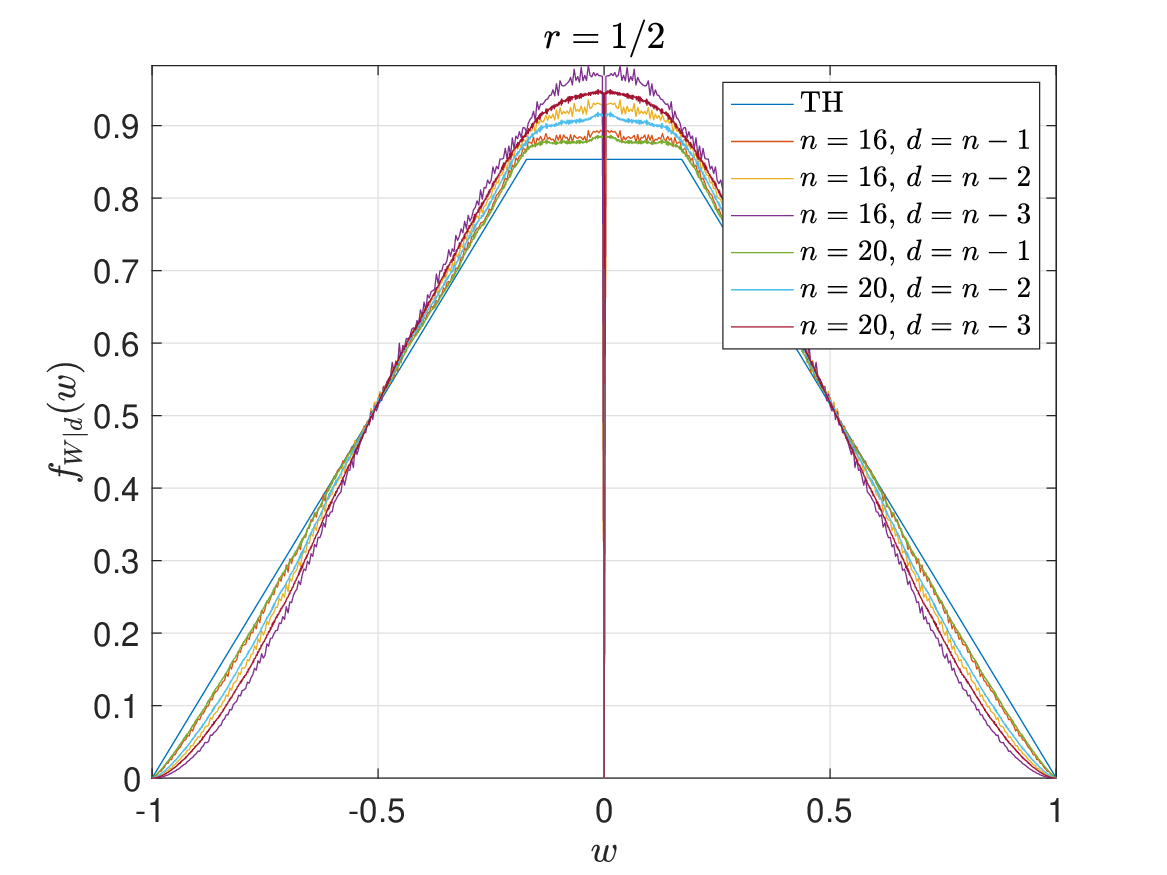}\label{subfig:fWd_vs_n}}
	\caption{Distribution of normalized shift function. The TH curves in (a), (d), (e), and (f) are given by \eqref{eq:fWn}, while the TH curves in (b) and (c) are given by \eqref{eq:fWn1k1} and \eqref{eq:fWn1k2}, respectively. In (b), (c), and (d), the parameter $k$ is defined by \eqref{eq:jn1}. Note that for every EXP/experimental curve, there is a zero point at $w=0$, because $\tau(j^d,b^d)\neq 0$ for every $j^d\in{\cal J}_{n,d}$ and every $b^d\in\mathbb{B}^d$ when $r=1/2$.}
	\label{fig:fW}
\end{figure*}

\subsection{HDS Obtained by Different Methods}\label{subsec:larged}
Now we verify the four approximate formulas of $\psi(d;n)$ listed in Table~\ref{tab:summary}. For $r=1/2$, we have $f(1/2)=\frac{1}{2-\sqrt{2}}\approx 1.7$, so for the TH-4/fast approximate formula \eqref{eq:fast}, $\psi(n;n) \approx 1.7\times 2^{-n/2}\binom{n}{n} \approx 0.0017$, where $n=20$, and $\psi(d;n) \approx 1.7\times 2^{-n/2-1}\binom{n}{d}$ for $d<n$; while for the TH-1/binomial approximate formula \eqref{eq:coarsehds}, 
\begin{align}
	\int_0^1{f^2(u)\,du} 
	&= \frac{2}{(3\sqrt{2}-4)^2}\int_{0}^{\sqrt{2}-1}u^2\,du + \frac{1}{(2-\sqrt{2})^2} \int_{\sqrt{2}-1}^{2-\sqrt{2}}1\,du\nonumber\\
	&= \frac{2(\sqrt{2}-1)^3}{3(3\sqrt{2}-4)^2} + \frac{3-2\sqrt{2}}{(2-\sqrt{2})^2}
	= \frac{1}{3(\sqrt{2}-1)} + 1/2 \approx 1.3,\nonumber
\end{align}
and hence $\psi(d;n)\approx 1.3\times 2^{-n/2}\binom{n}{d}$ for every $d$. The whole HDS for all $d\in[1:n]$ is shown by sub-Fig.~\ref{subfig:hds}, while for clarity, sub-Fig.~\ref{subfig:hds_zoom} zooms in on the partial HDS for large $d$. We have the following findings. 
\begin{itemize}
	\item The TH-1/binomial formula \eqref{eq:coarsehds} is only a coarse approximation of the EXP/experimental curve. As a rule of thumb, with \eqref{eq:coarsehds}, we will get a smaller value of $\psi(d;n)$ than its real value for $d<n/2$, and a larger value than its real value for $d>n/2$, implying that overlapped arithmetic codes are worse than random codes \cite{FangTCOM16a}. In other words, \eqref{eq:coarsehds} is relatively accurate only for $d\approx n/2$. 
	
	\item The TH-2/soft approximate formula \eqref{eq:psith2} perfectly coincides with the EXP/experimental curve for all $[1:n]$, however the cost is high complexity for large $d$. 
	
	\item The TH-3/hard approximate formula \eqref{eq:psith3} almost coincides with the EXP/experimental curve for large $d$, but there will be a larger deviation as $d$ decreases. 
	
	\item The TH-4/fast approximate formula \eqref{eq:fast} coincides with the EXP/experimental curve perfectly for $d\approx n$, but as $d$ decreases, the TH-4/fast curve will be much lower than the EXP/experimental curve. 
\end{itemize}
In one word, all theoretical predictions summarized in Table~\ref{tab:summary} are perfectly verified by Fig.~\ref{fig:hds}.

Finally, based on theoretical analyses and experimental results, we would like to end this section with the following suggestion. That is, if we want to calculate the HDS for an overlapped arithmetic code with a good compromise between accurateness and complexity, the best way is to use the TH-2/soft formula \eqref{eq:psith2} for $d\approx 1$ and the TH-4/fast formula \eqref{eq:fast} for $d\approx n$, while use the TH-1/binomial formula \eqref{eq:coarsehds} for other $d$.

\begin{figure*}[!t]
	\centering
	\subfigure[]{\includegraphics[width=.5\linewidth]{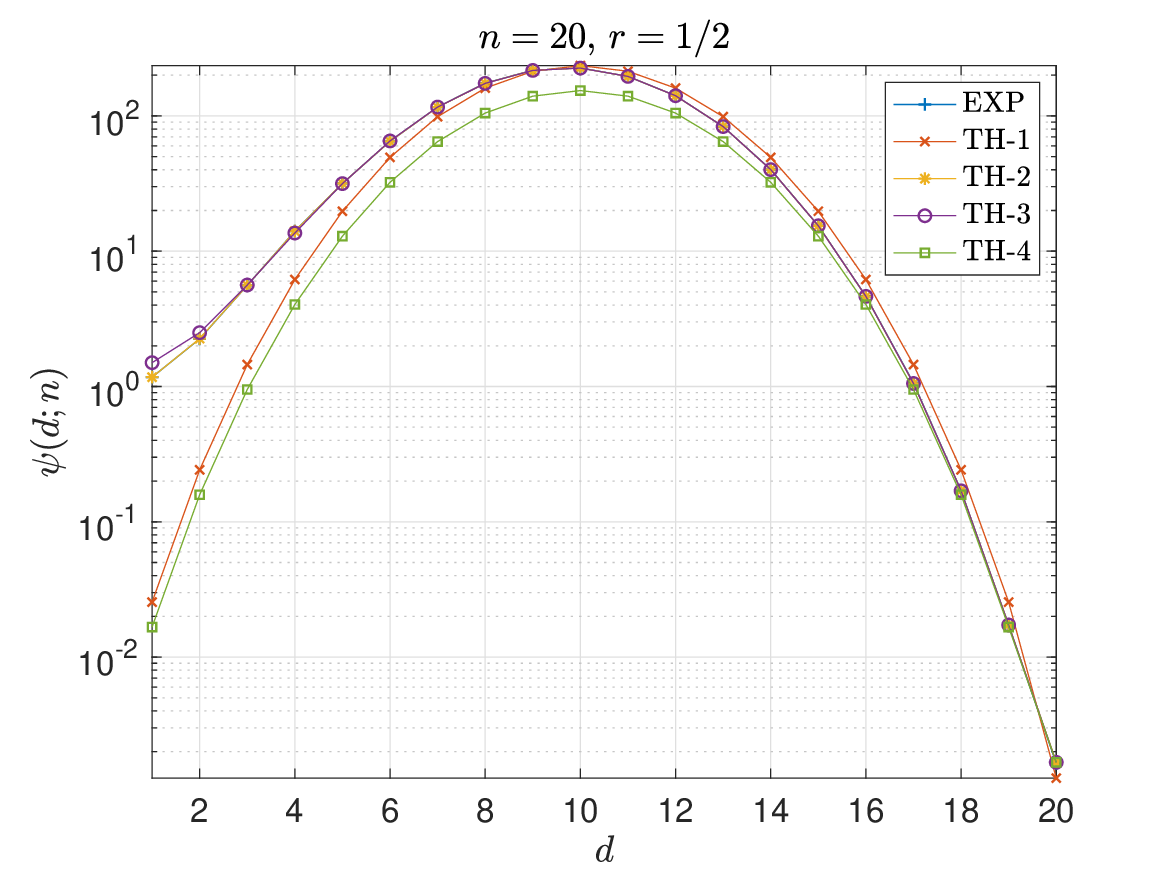}\label{subfig:hds}}%
	\subfigure[]{\includegraphics[width=.5\linewidth]{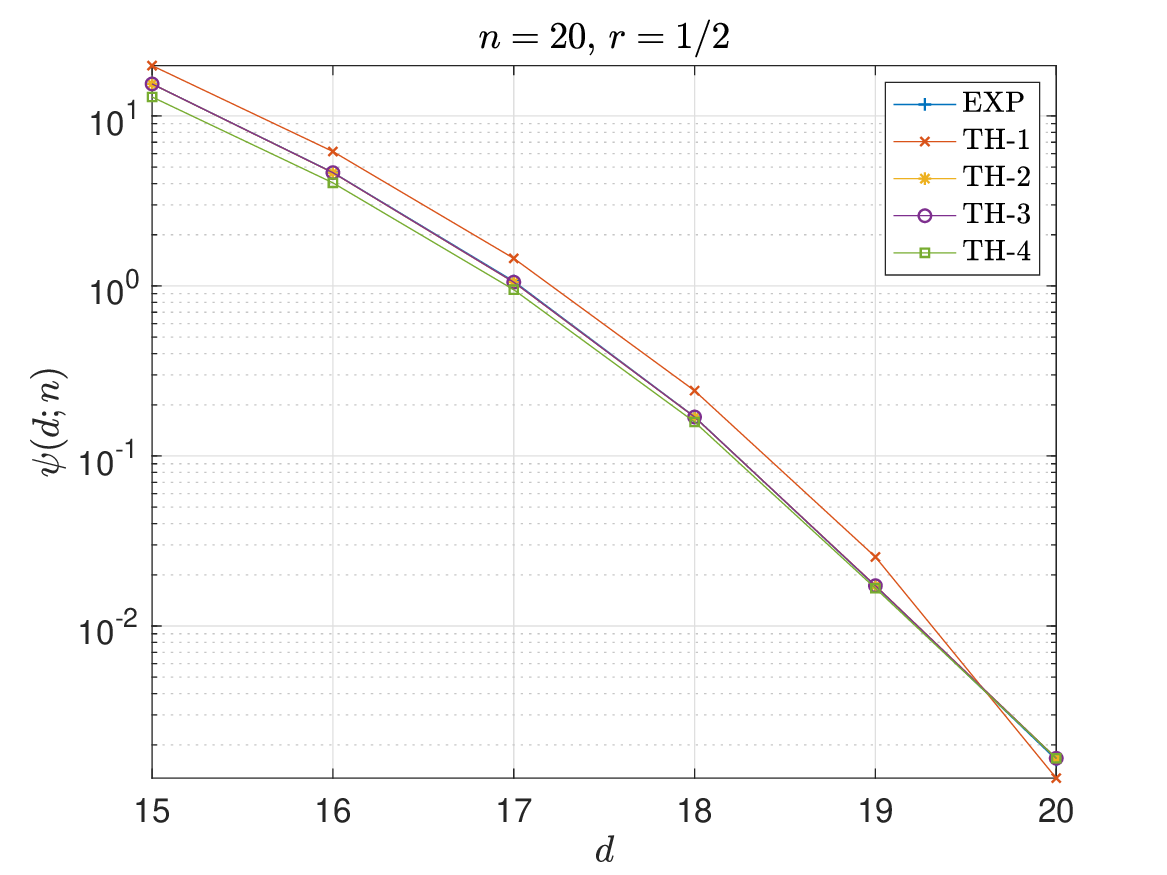}\label{subfig:hds_zoom}}
	\caption{Comparison of HDS obtained by different approximate formulas, where TH-1, TH-2, TH-3, and TH-4 are defined by Table~\ref{tab:summary}, and EXP refers to experimental curves.}
	\label{fig:hds}
\end{figure*}

\section{Conclusion}\label{sec:conclusion}
In this paper, we define the concept of HDS and derive four approximate formulas of HDS. The most significant novelty of this paper is bridging HDS and CCS, which were almost always separately treated ever before. We show that for $d\approx n$, the HDS $\psi(d;n)$ can be accurately and quickly calculated with the initial CCS $f(u)$. Another important novelty of this paper is bridging HDS with polynomial division. We reveal how to calculate the concrete analytical form of $\psi(3;n)$ when $2^r$ is an algebraic number, which may make $\psi(3;n)$ divergent. All these theoretical results will promote our understanding on overlapped arithmetic codes to a great extent.

In the future, we will move forward along the following directions. First, we will try to improve the HDS of overlapped arithmetic codes. As shown by the experimental results, the HDS of overlapped arithmetic codes is actually inferior to the HDS of random codes because $\psi(d;n)$ for small $d$ will not converge to $0$ as $n\to\infty$. It has been shown in \cite{FangTCOM16b} that with non-overlapping symbol-interval mapping, $\psi(d;n)$ for small $d$ may converge to $0$ as $n\to\infty$. We guess that with a more flexible symbol-interval mapping scheme, $\psi(d;n)$ for small $d$ may converge to $0$ even more quickly as $n\to\infty$. How to find such a good symbol-interval mapping scheme deserves an indepth exploration. Second, only uniform binary sources are considered in this paper. It will be interesting to define HDS for nonuniform binary sources and derive the calculation formula. Third, it will be very useful to extend the HDS of binary overlapped arithmetic codes to the \textit{Manhattan Distance Spectrum} (MDS) of nonbinary overlapped arithmetic codes, because nonbinary overlapped arithmetic codes significantly outperform nonbinary LDPC codes \cite{FangTIT23}.

All source codes to reproduce the results in this paper have been released in \cite{software}.

\appendices
\section{Proof of Theorem~\ref{thm:hds}}\label{prf:hds}
Let us define the conditional indicator function as
\begin{align}
	{\bf 1}_{A|B} \triangleq 
	\begin{cases}
		1, & A{\rm~is~true~given}~B\\
		0, & A{\rm~is~false~give}~B
	\end{cases}.\nonumber
\end{align}
Then \eqref{eq:psind} can be written as
\begin{align}
	\psi(d;n) = 2^{-d}\sum_{b^d\in\mathbb{B}^d}\sum_{j^d\in{\cal J}_{n,d}}{\psi(d;n|j^d,b^d)},\nonumber
\end{align}
where
\begin{align}
	\psi(d;n|j^d,b^d) \triangleq 
	2^{-(n-d)}\sum_{x_{[n]\setminus j^d}\in\mathbb{B}^{n-d}}
	{\bf 1}_{m(x^n)=m(x^n\oplus z_{j^d}^n) | x_{j^d}=b^d}.\nonumber
\end{align}
Let us define the following binary random variable
\begin{align}
	V(j^d,b^d) \triangleq {\bf 1}_{m(X^n)=m(X^n\oplus z_{j^d}^n) | X_{j^d}=b^d}.\nonumber
\end{align}
Then ${\bf 1}_{m(x^n)=m(x^n\oplus z_{j^d}^n) | x_{j^d}=b^d}$ is a realization of $V(j^d,b^d)$. According to Theorem~\ref{thm:equiv}, we have
\begin{align}
	\left\{m(x^n)=m(x^n\oplus z_{j^d}^n)\middle|x_{j^d}=b^d\right\} \leftrightarrow \left\{\ell(x^n|x_{j^d}=b^d) \in {\frak I}(j^d,b^d)\right\},\nonumber
\end{align}
where $\{\cdot\}\leftrightarrow\{\cdot\}$ denotes the equivalence between two events, so $V(j^d,b^d)$ is a binary random variable with bias probability $\Pr\left\{E(j^d,b^d)\in{\frak I}(j^d,b^d)\right\}$, where $E(j^d,b^d)\triangleq\ell(X^n|X_{j^d}=b^d)$, as defined by \eqref{eq:E}. As $(n-d)\to\infty$, we have infinite independent realizations of $V(j^d,b^d)$ and thus according to the law of large numbers,
\begin{align}\label{eq:approx}
	\lim_{(n-d)\to\infty}\psi(d;n|j^d,b^d) 
	&= \lim_{(n-d)\to\infty}\Pr\left\{E(j^d,b^d)\in{\frak I}(j^d,b^d)\right\}\nonumber\\
	&\stackrel{(a)}{=} \left(1-|\tau(j^d,b^d)|\right)^+,
\end{align}
where $(a)$ comes from Theorem~\ref{thm:prob}. In turn, we obtain
\begin{align}
	\psi(d|b^d) = \lim_{(n-d)\to\infty}\sum_{j^d\in{\cal J}_{n,d}}\psi(d;n|j^d,b^d) = \sum_{j^d\in{\cal J}_{\infty,d}}\left(1-|\tau(j^d,b^d)|\right)^+.\nonumber
\end{align}
Now \eqref{eq:psidbd} follows immediately.

\begin{remark}[A Note about the Proof of Theorem~\ref{thm:hds}]
	The key step in the proof of Theorem~\ref{thm:hds} is \eqref{eq:approx}, which is actually based on Monte-Carlo's principle, \textit{i.e.}, as the number of random trials (the indicator functions) goes to infinity, the average of the results will converge to the expectation. On the contrary, if $(n-d)<\infty$, \eqref{eq:approx} will not hold.
\end{remark}

\section{Proof of Corollary~\ref{corol:psi2}}\label{prf:psi2}
When $d=2$, \eqref{eq:psidbd} will become
\begin{align}
	\psi(2) = 2^{-2}\sum_{b^2\in\mathbb{B}^2}{\psi(2|b^2)},\nonumber
\end{align}
where 
\begin{align}
	\psi(2|b^2) = \sum_{j^2\in{\cal J}_{\infty,2}}{\left(1-|\tau(j^2,b^2)|\right)^+}.\nonumber
\end{align}
Since $\psi(2|0^2)=\psi(2|1^2)$ and $\psi(2|\underline{01})=\psi(2|\underline{10})$, only $\psi(2|0^2)$ and $\psi(2|\underline{10})$ will be tackled below. By \eqref{eq:tau},
\begin{align}\label{eq:tau_00_10}
	\begin{cases}
		\tau(j^2,0^2) = (1-2^{-r})(2^{rj_1}+2^{rj_2})>0\\
		\tau(j^2,\underline{10}) = (1-2^{-r})(-2^{rj_1}+2^{rj_2})>0
	\end{cases},
\end{align}
where $1\leq j_1<j_2$ as agreed. For convenience, the above equation can be written as
\begin{align}
	\begin{cases}
		\tau(j^2,0^2) = (1-2^{-r})2^{rj_1}(2^{kr}+1)>0\\
		\tau(j^2,\underline{10}) = (1-2^{-r})2^{rj_1}(2^{kr}-1)>0
	\end{cases},\nonumber
\end{align}
where $k\triangleq(j_2-j_1)\geq1$. Considering the monotonicity of $2^{rj_1}$ and $2^{kr}$ w.r.t. $j_1$ and $k$, respectively, given $\tau(j^2,0^2)<1$ and $\tau(j^2,\underline{10})<1$, there is
\begin{align}
	\begin{cases}
		2^r \leq 2^{rj_1} = \frac{\tau(j^2,0^2)}{(1-2^{-r})(2^{kr}+1)} < \frac{1}{(1-2^{-r})(2^{kr}+1)} \leq \frac{1}{(1-2^{-r})(2^r+1)} = \frac{2^r}{2^{2r}-1}\\
		2^r \leq 2^{rj_1} = \frac{\tau(j^2,\underline{10})}{(1-2^{-r})(2^{kr}-1)} < \frac{1}{(1-2^{-r})(2^{kr}-1)} \leq \frac{1}{(1-2^{-r})(2^r-1)} = \frac{2^r}{(2^r-1)^2}
	\end{cases}.\nonumber
\end{align}
After taking the base-2 log of both sides of the above equation, we will get
\begin{align}
	\begin{cases}
		r\leq rj_1 < -\log_2{(1-2^{-r})} - \log_2{(2^{kr}+1)}\leq r-\log_2{(2^{2r}-1)}\\
		r\leq rj_1 < -\log_2{(1-2^{-r})} - \log_2{(2^{kr}-1)}\leq r-2\log_2{(2^r-1)}
	\end{cases}.\nonumber
\end{align}

\begin{remark}[Discussion on $\tau(j^2,0^2)$]
	To make $\tau(j^2,0^2)<1$, the following constraints should be satisfied 
	\begin{align}
		\begin{cases}
			2^r < \frac{1}{(1-2^{-r})(2^{kr}+1)}\\
			rj_1 < r-\log_2{(2^{2r}-1)} 
		\end{cases}.\nonumber
	\end{align}
	After solving the above group of inequalities, we will get the upper bounds of $k$ and $j_1$:
	\begin{align}
		\begin{cases}
			k \leq K_{2,1} \triangleq \left\lceil\frac{\log_2{(2-2^r)}-\log_2{(2^r-1)}}{r}\right\rceil - 1 
			= \left\lceil\frac{\log_2{(2^{1-r}-1)}-\log_2{(2^r-1)}}{r}\right\rceil < \infty\\
			j_1 \leq J_{2,1} \triangleq  
			\left\lceil\frac{r-\log_2{(2^{2r}-1)}}{r}\right\rceil - 1 
			= -\left\lfloor\frac{\log_2{(2^{2r}-1)}}{r}\right\rfloor
			< \infty
		\end{cases}.\nonumber
	\end{align}
	Note that $k$ takes $K_{2,1}$ when $j_1=1$, and $j_1$ takes $J_{2,1}$ when $k=1$. The strict relation between $k$ and $j_1$ can be found by solving the following inequality
	\begin{align}
		2^{rj_1} < \left({(1-2^{-r})(2^{kr}+1)}\right)^{-1},\nonumber
	\end{align}
	which is equivalent to
	\begin{align}
		2^{kr} < \frac{2^{-rj_1}}{1-2^{-r}} - 1 = \frac{2^{-rj_1}-1+2^{-r}}{1-2^{-r}}.\nonumber
	\end{align}
	After taking the base-2 log of both sides of the above inequality, we will get
	\begin{align}
		kr < \log_2{(2^{-rj_1}-1+2^{-r})} - \log_2{(1-2^{-r})}.\nonumber
	\end{align}
	Hence, the conditional upper bound of $k$ given $j_1$ is
	\begin{align}
		k \leq \kappa_1(j_1) 
		&\triangleq \left\lceil\frac{\log_2{(2^{-rj_1}-1+2^{-r})} - \log_2{(1-2^{-r})}}{r}\right\rceil - 1\nonumber\\
		&= \left\lceil\frac{\log_2{(2^{-rj_1}-1+2^{-r})} - \log_2{(2^r-1)}}{r}\right\rceil,\nonumber		
	\end{align}
	which is a strictly decreasing function w.r.t. $j_1$, and thus $\kappa_1(j_1)\leq \kappa_1(1)=K_{2,1}$.
	
	Let us discuss $j_2$ now. According to the first branch of \eqref{eq:tau_00_10}, $j_2$ takes the maximum when $j_1=1$. Thus the upper bound of $j_2$ is given by
	\begin{align}
		j_2 \leq 1 + K_{2,1} = \left\lceil\frac{\log_2{(2-2^r)}-\log_2{(2^r-1)}}{r}\right\rceil < \infty.\nonumber
	\end{align}	
\end{remark}

\begin{remark}[Discussion on $\tau(j^2,\underline{10})$]
	Similarly, to make $\tau(j^2,\underline{10})<1$, the following constraints should be satisfied 
	\begin{align}
		\begin{cases}
			2^r < \frac{1}{(1-2^{-r})(2^{kr}-1)}\\
			rj_1 < r-2\log_2{(2^r-1)}
		\end{cases}.\nonumber
	\end{align}
	After solving the above group of inequalities, we will get the upper bounds of $k$ and $j_1$:
	\begin{align}
		\begin{cases}
			k \leq K_{2,2} \triangleq 
			\left\lceil-\frac{\log_2{(1-2^{-r})}}{r}\right\rceil-1
			= -\left\lfloor\frac{\log_2{(2^r-1)}}{r}\right\rfloor = J_1 < \infty\\
			j_1 \leq J_{2,2} \triangleq \left\lceil\frac{r-2\log_2{(2^r-1)}}{r}\right\rceil-1 
			= -\left\lfloor\frac{2\log_2{(2^r-1)}}{r}\right\rfloor
			< \infty
		\end{cases}.\nonumber
	\end{align}	
	Note that $k$ takes $K_{2,2}$ when $j_1=1$, and $j_1$ takes $J_{2,2}$ when $k=1$. The strict relation between $k$ and $j_1$ can be found by solving the following inequality
	\begin{align}
		2^{rj_1} < \left({(1-2^{-r})(2^{kr}-1)}\right)^{-1},\nonumber
	\end{align}		
	which is equivalent to
	\begin{align}
		2^{kr} < \frac{2^{-rj_1}}{1-2^{-r}} + 1 = \frac{2^{-rj_1}+1-2^{-r}}{1-2^{-r}}.\nonumber
	\end{align}
	After taking the base-2 log of both sides of the above inequality, we will get
	\begin{align}
		kr < \log_2{(2^{-rj_1}+1-2^{-r})} - \log_2{(1-2^{-r})}.\nonumber
	\end{align}
	Hence, the conditional upper bound of $k$ given $j_1$ is
	\begin{align}
		k \leq \kappa_2(j_1) 
		&\triangleq \left\lceil\frac{\log_2{(2^{-rj_1}+1-2^{-r})} - \log_2{(1-2^{-r})}}{r}\right\rceil - 1\nonumber\\
		&= \left\lceil\frac{\log_2{(2^{-rj_1}+1-2^{-r})} - \log_2{(2^r-1)}}{r}\right\rceil,\nonumber		
	\end{align}
	which is a strictly decreasing function w.r.t. $j_1$, and thus $\kappa_2(j_1)\leq \kappa_2(1)=K_{2,2}=J_1$.
	
	Finally, we discuss $j_2$. According to the second branch of \eqref{eq:tau_00_10}, $j_2$ takes the maximum when $k=j_2-j_1=1$. Thus the upper bound of $j_2$ is given by
	\begin{align}
		j_2 \leq 1+J_{2,2} = 1-\left\lfloor\frac{2\log_2{(2^r-1)}}{r}\right\rfloor < \infty.\nonumber
	\end{align}
\end{remark}

\section{Proofs of Theorem~\ref{thm:psi3} and Theorem~\ref{thm:psi3linear}}
\subsection{Proof of Theorem~\ref{thm:psi3}}\label{prf:psi3}
We first prove that if there exist one or more pairs of integers $j\geq1$ and $k\geq1$ such that $2^{jr}(2^{kr}-1)=1$, then $\psi(3)=\infty$. According to \eqref{eq:psid}, we have
\begin{align}
	\psi(3|b^3) = \sum_{j_3=3}^{\infty}\sum_{j_2=2}^{(j_3-1)}\sum_{j_1=1}^{(j_2-1)}{\left(1-|\tau(j^3,b^3)|\right)^+}.\nonumber	
\end{align}
Then according to the symmetry, 
\begin{align}
	\psi(3) 
	&= 2^{-3}\sum_{b^3\in\mathbb{B}^3}{\psi(3|b^3)}\nonumber\\
	&= \left(\psi(3|0^3)+\psi(3|\underline{100})+\psi(3|\underline{010})+\psi(3|\underline{110})\right)/4\nonumber\\
	&> \psi(3|\underline{110})/4.\nonumber
\end{align}
Let us focus on $\psi(3|\underline{110})$, which can be written as
\begin{align}
	\psi(3|\underline{110}) = \sum_{j_3=3}^{\infty}\sum_{j_2=2}^{(j_3-1)}\sum_{j_1=1}^{(j_2-1)}{\left(1-(1-2^{-r})|-2^{j_1r}-2^{j_2r}+2^{j_3r}|\right)^+},\nonumber
\end{align}
where $1\leq j_1<j_2<j_3$ as agreed. If there exists a pair of integers $j\geq1$ and $k\geq1$ such that $2^{jr}(2^{kr}-1)=1$, then $2^{(j+k)r}-2^{jr}-1=0$ and for any $i\geq1$,
\begin{align}
	2^{ir}(2^{(j+k)r}-2^{jr}-1)=(2^{(i+j+k)r}-2^{(i+j)r}-2^{ir})=0.\nonumber
\end{align}
Therefore,
\begin{align}
	\psi(3|\underline{110}) 
	> \sum_{i=1}^{\infty}{\left(1-(1-2^{-r})(2^{(i+j+k)r}-2^{(i+j)r}-2^{ir})\right)} = \sum_{i=1}^{\infty}{1} = \infty,\nonumber
\end{align}
which is immediately followed by $\psi(3)>\psi(3|\underline{110})/4=\infty$.

{\bf Converse.} Then we prove that if there is no pair of integers $j\geq1$ and $k\geq1$ such that $2^{jr}(2^{kr}-1)=1$, then $\psi(3)<\infty$. To begin with, let us prove $\psi(3|0^3)<\infty$, $\psi(3|\underline{100})<\infty$, and $\psi(3|\underline{010})<\infty$. Given $1\leq j_1<j_2<j_3$, the following inequality holds obviously
\begin{align}
	0 < (-2^{j_2r}+2^{j_3r}) < (2^{j_1r}-2^{j_2r}+2^{j_3r}) < (-2^{j_1r}+2^{j_2r}+2^{j_3r}) < (2^{j_1r}+2^{j_2r}+2^{j_3r}).\nonumber
\end{align}
Hence, we need to consider only
\begin{align}
	(2^{j_3r}-2^{j_2r}) = 2^{j_3r}(1-2^{-(j_3-j_2)r})\geq 2^{j_3r}(1-2^{-r}) > 0.\nonumber
\end{align}
Since $2^{j_3r}(1-2^{-r})$ is strictly increasing w.r.t. $j_3$, there must exist an integer $J$ such that for any $j_3>J$,
\begin{align}
	(1-2^{-r})2^{j_3r}(1-2^{-r})\geq 1.\nonumber
\end{align}
Therefore,
\begin{align}
	\begin{cases}
		\psi(3|0^3) = \sum_{j_3=3}^{J}\sum_{j_2=2}^{j_3-1}\sum_{j_1=1}^{j_2-1}{\left(1-(1-2^{-r})(2^{j_1r}+2^{j_2r}+2^{j_3r})\right)^+}<\infty\\
		\psi(3|\underline{100}) = \sum_{j_3=3}^{J}\sum_{j_2=2}^{j_3-1}\sum_{j_1=1}^{j_2-1}{\left(1-(1-2^{-r})(-2^{j_1r}+2^{j_2r}+2^{j_3r})\right)^+}<\infty\\
		\psi(3|\underline{010}) = \sum_{j_3=3}^{J}\sum_{j_2=2}^{j_3-1}\sum_{j_1=1}^{j_2-1}{\left(1-(1-2^{-r})(2^{j_1r}-2^{j_2r}+2^{j_3r})\right)^+}<\infty	
	\end{cases}.\nonumber
\end{align}

\begin{remark}[Convergence of $\psi(3|\underline{110})$.]
	Now the remaining thing is whether $\psi(3|\underline{110})<\infty$ if there is no pair of integers $i\geq1$ and $j\geq1$ such that $2^{ir}(2^{jr}-1)=1$. This is a much more difficult problem. Let us begin with 
	\begin{align}
		(-2^{j_1r}-2^{j_2r}+2^{j_3r}) 
		&= 2^{j_1r}(2^{(j_3-j_1)r}-2^{(j_2-j_1)r}-1)\nonumber\\
		&= 2^{j_1r}(2^{kr}-2^{ir}-1) = 2^{j_1r}\left(2^{ir}(2^{jr}-1)-1\right),\nonumber
	\end{align}
	where $k=(j_3-j_1)=(i+j)>i=(j_2-j_1)\geq 1$. Then $\psi(3|\underline{110})$ can be rewritten as
	\begin{align}
		\psi(3|\underline{110}) = \sum_{j_1=1}^{\infty}\sum_{k=2}^{\infty}\sum_{i=1}^{k-1}{\left(1-(1-2^{-r})2^{j_1r}|2^{kr}-2^{ir}-1|\right)^+}.\nonumber
	\end{align}
	Obviously, if the upper bounds of $j_1$ and $k$ exist, then $\psi(3|\underline{110})<\infty$.
	
	We consider $k$ first. Apparently, $(2^{kr}-2^{ir}-1)$ is a strictly decreasing function w.r.t. $i$. Hence,
	\begin{align}
		(2^{kr}-2^{ir}-1) \geq (2^{kr}-2^{(k-1)r}-1) = (2^{kr}(1-2^{-r})-1).\nonumber
	\end{align}
	Since $(1-2^{-r})>0$, it is clear that $(2^{kr}(1-2^{-r})-1)$ is a strictly increasing function w.r.t. $k$. As such, there must exist an integer $K$ such that for any $k>K$,
	\begin{align}
		(1-2^{-r})(2^{kr}-2^{ir}-1) \geq (1-2^{-r})(2^{kr}(1-2^{-r})-1) \geq 1.\nonumber
	\end{align}
	For any $k>K$, we have $(2^{kr}-2^{ir}-1)>0$ and hence $2^{j_1r}(2^{kr}-2^{ir}-1)$ is a strictly increasing function w.r.t. $j_1$. Since $j_1\geq1$, for any $k>K$, we have
	\begin{align}
		(1-2^{-r})2^{j_1r}(2^{kr}-2^{ir}-1) > (1-2^{-r})(2^{kr}-2^{ir}-1) \geq 1.\nonumber
	\end{align}
	That means, it is unnecessary to consider the case $k>K$ and thus $\psi(3|\underline{110})$ can be reduced to
	\begin{align}
		\psi(3|\underline{110}) = \sum_{j_1=1}^{\infty}\sum_{k=2}^{K}\sum_{i=1}^{k-1}{\left(1-(1-2^{-r})2^{j_1r}|2^{kr}-2^{ir}-1|\right)^+}.\nonumber
	\end{align}

	Then we consider $j_1$. Given $1\leq i<k\leq K$, there are $K(K-1)/2<\infty$ pairs of $i$ and $k$ in total. If there is no pair of $i$ and $j$ such that $(2^{kr}-2^{ir})=2^{ir}(2^{jr}-1)=1$, we have $|2^{kr}-2^{ir}-1|>0$ for every $1\leq i<k\leq K$. Further, there must exist an integer $J$ such that for any $j_1>J$, the following inequality holds for every pair of $i$ and $k$ satisfying $1\leq i<k\leq K$:	
	\begin{align}
		(1-2^{-r})2^{j_1r}|2^{kr}-2^{ir}-1| \geq 1.\nonumber
	\end{align} 
	Consequently, $\psi(3|\underline{110})$ can be further reduced to
	\begin{align}
		\psi(3|\underline{110}) = \sum_{j_1=1}^{J}\sum_{k=2}^{K}\sum_{i=1}^{k-1}{\left(1-(1-2^{-r})2^{j_1r}|2^{kr}-2^{ir}-1|\right)^+}<\infty.\nonumber
	\end{align}
 	So far, we have finished the proof.
\end{remark}

\subsection{Proof of Theorem~\ref{thm:psi3linear}}\label{prf:psi3linear}
According to the above analysis, it can be seen that if there exist some pairs of integers $i\geq1$ and $j\geq1$ such that $2^{ir}(2^{jr}-1)=1$, then $\psi(3;n)$ will continuously go up as $n$ increases. Let ${\cal P}$ denote the set of the pairs of integers $i\geq1$ and $j\geq1$ satisfying $2^{ir}(2^{jr}-1)=1$.	After $n>n_0$, where $n_0$ is a certain integer, all convergent terms, \textit{i.e.}, $\psi(3|0^3)$, $\psi(3|\underline{100})$, $\psi(3|\underline{010})$, \textit{etc.}, will stay the same, while those divergent terms caused by $2^{ir}(2^{jr}-1)=1$ will increase linearly. Let $c_0$ be the sum of convergent terms. Then
\begin{align}
	\psi(3;n) 
	&\approx c_0 + (1/4)\sum_{(i,j)\in{\cal P}}\sum_{j_1=1}^{n-(i+j)}{\left(1-(1-2^{-r})2^{j_1r}|2^{(i+j)r}-2^{ir}-1|\right)^+}\nonumber\\
	&\approx c_0 + (1/4)\sum_{(i,j)\in{\cal P}}{(n-(i+j))}.\nonumber
\end{align}

\section{Proof of Corollary~\ref{corol:psi3}}\label{prf:psi3golden}
Since $\psi(3;n|b^3)=\psi(3;n|(1^3\oplus b^3))$, we have
\begin{align}\label{eq:psi3n}
	\psi(3;n) = \frac{\psi(3;n|0^3)+\psi(3;n|\underline{100})+\psi(3;n|\underline{010})+\psi(3;n|\underline{110})}{4}.
\end{align}
According to \eqref{eq:psid}, we have
\begin{align}
	\begin{cases}
		\psi(3;n|0^3) \approx \sum_{j_3=3}^{n}\sum_{j_2=2}^{j_3-1}\sum_{j_1=1}^{j_2-1}{\left(1-(1-2^{-r})|2^{j_1r}+2^{j_2r}+2^{j_3r}|\right)^+}\\
		\psi(3;n|\underline{100}) \approx \sum_{j_3=3}^{n}\sum_{j_2=2}^{j_3-1}\sum_{j_1=1}^{j_2-1}{\left(1-(1-2^{-r})|-2^{j_1r}+2^{j_2r}+2^{j_3r}|\right)^+}\\
		\psi(3;n|\underline{010}) \approx \sum_{j_3=3}^{n}\sum_{j_2=2}^{j_3-1}\sum_{j_1=1}^{j_2-1}{\left(1-(1-2^{-r})|2^{j_1r}-2^{j_2r}+2^{j_3r}|\right)^+}	
	\end{cases}.\nonumber
\end{align}
Obviously,
\begin{align}\label{eq:order}
	0< (2^{j_1r}-2^{j_2r}+2^{j_3r}) < (-2^{j_1r}+2^{j_2r}+2^{j_3r}) < (2^{j_1r}+2^{j_2r}+2^{j_3r}).
\end{align}
The prerequisite of this corollary is $(2^{2r}-2^r-1)=0$. With polynomial division, it is easy to obtain
\begin{align}
	(2^{j_1r}-2^{j_2r}+2^{j_3r}) \geq (2^{r}-2^{2r}+2^{3r}) = (2^{r}-2^{2r}+2^{3r})\bmod(2^{2r}-2^r-1) = 2^{1+r}.\nonumber
\end{align}
Consequently,
\begin{align}
	(1-2^{-r})(2^{j_1r}-2^{j_2r}+2^{j_3r})\geq(1-2^{-r})(2^{r}-2^{2r}+2^{3r})=2(2^r-1)\approx1.2361>1.\nonumber
\end{align}
Hence $\psi(3;n|\underline{010})\equiv0$, and immediately it can be deduced from \eqref{eq:order} that $\psi(3;n|\underline{100})\equiv0$ and $\psi(3;n|0^3)\equiv0$. Now it is clear that $\psi(3;n) = \psi(3;n|\underline{110})/4$.

According to \eqref{eq:psid}, we have
\begin{align}
	\psi(3;n|\underline{110}) \approx \sum_{j_3=3}^{n}\sum_{j_2=2}^{j_3-1}\sum_{j_1=1}^{j_2-1}{\left(1-(1-2^{-r})|2^{j_3r}-2^{j_2r}-2^{j_1r}|\right)^+}.\nonumber
\end{align}
According to the first bullet of Lemma~\ref{lem:signum}, given $(2^{2r}-2^r-1)=0$, we have $(2^{j_3r}-2^{j_2r}-2^{j_1r})\geq0$. Hence,
\begin{align}
	(2^{j_3r}-2^{j_2r}-2^{j_1r})=2^{j_1r}(2^{kr}-2^{ir}-1)\geq2^r(2^{kr}-2^{ir}-1),\nonumber
\end{align}
where $k\triangleq(j_3-j_1)>i\triangleq(j_2-j_1)\geq1$. As such,
\begin{align}
	(1-2^{-r})(2^{j_3r}-2^{j_2r}-2^{j_1r}) \geq (2^r-1)(2^{kr}-2^{ir}-1) \geq (2^r-1)(2^{kr}-2^{(k-1)r}-1).\nonumber
\end{align}
For $k=4$, we have
\begin{align}
	(2^{4r}-2^{3r}-1) = (2^{4r}-2^{3r}-1)\bmod(2^{2r}-2^r-1) = 2^r,\nonumber
\end{align}
which is followed by $(2^r-1)(2^{4r}-2^{3r}-1)=1$. Therefore, $(2^r-1)(2^{kr}-2^{ir}-1)\geq1$ for any $k\geq4$, where $i<k$ as agreed. Now $\psi(3;n|\underline{110})$ can be reduced to
\begin{align}
	\psi(3;n|\underline{110}) 
	&\approx 
	\sum_{k=2}^{3}\sum_{j_1=1}^{n-k}\sum_{i=1}^{k-1}{\left(1-(1-2^{-r})2^{j_1r}|2^{kr}-2^{ir}-1|\right)^+}\nonumber\\
	&= \sum_{j_1=1}^{n-2}{\left(1-(1-2^{-r})2^{j_1r}\underbrace{(2^{2r}-2^{r}-1)}_{=0}\right)^+} +\nonumber\\ 
	&\sum_{j_1=1}^{n-3}{\left(\left(1-(1-2^{-r})2^{j_1r}|2^{3r}-2^{r}-1|\right)^+ + \left(1-(1-2^{-r})2^{j_1r}|2^{3r}-2^{2r}-1|\right)^+\right)}.\nonumber
\end{align}
With polynomial division, it is easy to get $(2^{3r}-2^{r}-1)=2^r$ and $(2^{3r}-2^{2r}-1)=2^{-r}$. Hence
\begin{align}
	\psi(3;n|\underline{110}) \approx (n-2) + \sum_{j_1=1}^{n-3}{\left(\left(1-(2^r-1)2^{j_1r}\right)^+ + \left(1-(2^r-1)2^{(j_1-2)r}\right)^+\right)}.\nonumber
\end{align}
Since $(2^r-1)2^{j_1r}\geq 1$ for $j_1\geq1$ and $(2^r-1)2^{(j_1-2)r}\geq1$ for $j_1\geq3$, given $n\geq5$,
\begin{align}
	\psi(3;n|\underline{110}) 
	&\approx (n-2) + \sum_{j_1=1}^{2}{\left(1-(2^r-1)2^{(j_1-2)r}\right)}\nonumber\\
	&= (n-2) + \underbrace{(1-(2^r-1)2^{-r})}_{2^{-r}} + (1-(2^r-1))\nonumber\\
	&= (n-2) + \underbrace{(2^{-r}+1-2^r)}_{2^{-r}(1+2^r-2^{2r})=0} + 1 = (n-1).\nonumber
\end{align}
Finally, we obtain $\psi(3;n) = \psi(3;n|\underline{110})/4 \approx (n-1)/4$ for $n\geq 5$.

\section{Proof of Corollary~\ref{corol:psi3b}}\label{prf:psi3b}
As shown by \eqref{eq:psi3n}, $\psi(3;n)$ is the average of four terms, which will be discussed in turn. For conciseness, we define $\alpha\triangleq(x^3-x-1)=0$, where $x\triangleq2^r\in(1,2)$. Then $(1-2^{-r})2^r=(1-x^{-1})x=(x-1)$ and
\begin{align}
	\psi(3;n) = (1/4) \sum_{(s_1,s_0)\in\{-1,+1\}^2}\sum_{k=2}^{n-1}\sum_{i=1}^{k-1}\sum_{j=1}^{n-k}{\left(1-(x-1)x^{j-1}|x^k+s_1x^i+s_0|\right)^+}.\nonumber
\end{align}

\begin{definition}[Species, Generation, and Genus]
	We call $(x^k+s_1x^i+s_0)$ a species, where $k>i\geq1$ and $s_0,s_1\in\{+1,-1\}$, whose $j$-th generation is denoted as $x^{j-1}(x^k+s_1x^i+s_0)$, where $j\geq1$. For a given pair of $s_0,s_1\in\{+1,-1\}$, the set $\{(x^k+s_1x^i+s_0): k>i\geq1\}$ is called the $(s_1,s_0)$-th genus.
\end{definition}

\begin{definition}[Extinct Species and Alive Species]
	Let $(x^k+s_1x^i+s_0)$ be a species, where $k>i\geq1$ and $s_0,s_1\in\{+1,-1\}$. If $|x^k+s_1x^i+s_0|\geq(x-1)^{-1}$, we call it an extinct species; otherwise, we call it an alive species. Further, for an alive species, if $|x^k + s_1x^i + s_0|=0$, we call it an immortal species; otherwise, we call it a mortal species. That is, an alive species may be a mortal species or an immortal species. If all species belonging to a genus are extinct, then we say that this genus is an extinct genus.
\end{definition}

\begin{definition}[Lifespan of Species]
	Let $(x^k+s_1x^i+s_0)$ be a mortal species, where $k>i\geq1$ and $s_0,s_1\in\{+1,-1\}$. There must be an integer $J\geq1$ such that $x^{J-1}|x^k+s_1x^i+s_0|\geq(x-1)^{-1}$ for every $j>J$, and $x^{J-1}|x^k+s_1x^i+s_0|<(x-1)^{-1}$ for every $j\leq J$. We call $J$ the lifespan of the species $(x^k+s_1x^i+s_0)$. Apparently, for an immortal species, its lifespan is $J=\infty$, while for an extinct species, its lifespan is $J=0$.
\end{definition}

\begin{lemma}[Properties of Species]
	Let $J_{s_1,s_0,k,i}$ be the lifespan of the species $(x^k+s_1x^i+s_0)$, then
	\begin{align}
		\psi(3;n) = (1/4) \sum_{(s_1,s_0)\in\{-1,+1\}^2}\sum_{k=2}^{n-1}\sum_{i=1}^{k-1}\sum_{j=1}^{J_{s_1,s_0,k,i}}{\left(1-(x-1)x^{j-1}|x^k+s_1x^i+s_0|\right)^+}.\nonumber
	\end{align}
	An extinct species has no contribution to $\psi(3;n)$, a mortal species has a convergent contribution to $\psi(3;n)$, and an immortal species has a divergent contribution to $\psi(3;n)$.
\end{lemma}

\subsection{Calculation of $\psi(3;n|0^3)$}
Given $1\leq j_1<j_2<j_3\leq n$, we have
\begin{align}
	(1-2^{-r})(2^{j_1r}+2^{j_2r}+2^{j_3r})\geq (x-1)(x^2+x+1)\bmod\alpha = x>1.\nonumber
\end{align}
Hence, all species belonging to the $(+1,+1)$-th genus are extinct species, so $\psi(3;n|0^3)\equiv 0$.

\subsection{Calculation of $\psi(3;n|\underline{100})$}
Given $1\leq j_1<j_2<j_3\leq n$, we have
\begin{align}
	(-2^{j_1r}+2^{j_2r}+2^{j_3r}) = 2^{j_1r}(2^{kr}+2^{ir}-1) \geq 2^r(2^{kr}+2^r-1),\nonumber
\end{align}
where $k\triangleq(j_3-j_1)>i\triangleq(j_2-j_1)\geq1$. After solving the following inequality
\begin{align}
	(1-2^{-r})2^r(2^{kr}+2^r-1) = (x-1)(x^k+x-1) < 1,\nonumber
\end{align}
we obtain $k\leq 3$. Therefore, 
\begin{align}
	\psi(3;n|\underline{100}) \approx \sum_{k=2}^{3}\sum_{i=1}^{k-1}\sum_{j=1}^{n-k}{\left(1-x^{j-1}(x-1)(x^k+x^i-1)\right)^+},\nonumber
\end{align}
including only three species: $(x^3+x^2-1)\bmod\alpha=(x^2+x)$, $(x^3+x-1)\bmod\alpha=2x$, and $(x^2+x-1)$: 

\begin{itemize}
	\item For the species $(x^2+x)$, it is easy to know $(x-1)(x^2+x)\bmod\alpha=1$ and then $x^{j-1}(x-1)(x^2+x)\bmod\alpha=x^{j-1}\geq 1$ for every $j\geq1$, so it is an extinct species and has no contribution to $\psi(3;n|\underline{100})$.	
	\item For the species $2x$, we find $x^{j-1}(x-1)2x>1$ for $j\geq 2$, so only the first generation $2x$ contributes to $\psi(3;n|\underline{100})$, \textit{i.e.}, its lifespan is $J=1$.	
	\item For the species $(x^2+x-1)$, we find $x^{j-1}(x-1)(x^2+x-1)>1$ for $j\geq 3$, so only the first two generations $(x^2+x-1)$ and $x(x^2+x-1)$ contribute to $\psi(3;n|\underline{100})$, \textit{i.e.}, its lifespan is $J=2$.
\end{itemize}
In summary, the $(+1,-1)$-th genus includes two mortal species: $2x$, whose lifespan is $J=1$, and $(x^2+x-1)$, whose lifespan is $J=2$. These two species and related information are included in Table~\ref{tab:genus}. The accumulated contribution from these two mortal species is
\begin{align}\label{eq:100}
	(-2x^2+2x+1) + (x^2-x) = -x^2+x+1.
\end{align}

\begin{table*}[!t]
	\small\centering
	\caption{Mortal Species of the $(+1,-1)$-th Genus and the $(-1,+1)$-th Genus}
	\begin{tabular}{c||c||c||c||c||c}
		\hline
		Genus &Species &Species $\bmod\alpha$ &$J$ &Contribution &Contribution $\bmod\alpha$\\
		\hline
		\hline
		\multirow{2}*{\shortstack[]{$s_1=+1$\\$s_0=-1$}} 
		&$(x^3+x-1)$ &$2x$ &$1$ &$1-2x(x-1)$ &$(-2x^2+2x+1)$\\
		\cline{2-6}	
		&$(x^2+x-1)$ &$(x^2+x-1)$ &$2$ &$2-(x^2+x-1)(x^2-1)$ &$(x^2-x)$\\
		\hline		
		\hline
		\multirow{10}*{\shortstack[]{$s_1=-1$\\$s_0=+1$}}
		&$(x^7-x^6+1)$ 	& \multirow{3}*{$(x^2+1)$} & \multirow{3}*{$1$} & \multirow{3}*{$1-(x^2+1)(x-1)$} & \multirow{3}*{$(x^2-2x+1)$}\\
		&$(x^5-x^3+1)$ 	& & & &\\
		&$(x^4-x+1)$ 	& & & &\\
		\cline{2-6}	
		&$(x^6-x^5+1)$ & \multirow{2}*{$(x+1)$} & \multirow{2}*{$1$} & \multirow{2}*{$1-(x+1)(x-1)$} & \multirow{2}*{$(-x^2+2)$}\\
		&$(x^4-x^2+1)$ &  &  &  & \\
		\cline{2-6}	
		&$(x^5-x^4+1)$ & \multirow{2}*{$2$} & \multirow{2}*{$2$} & \multirow{2}*{$2-2(x^2-1)$} & \multirow{2}*{$(-2x^2+4)$}\\
		&$(x^3-x+1)$ &  & & \\
		\cline{2-6}	
		&$(x^4-x^3+1)$ & $x^2$ & $2$ & $2-x^2(x^2-1)$ & $(-x+2)$\\
		\cline{2-6}	
		&$(x^3-x^2+1)$ & $(-x^2+x+2)$ & $3$ & $3-(-x^2+x+2)(x^3-1)$ & $(-x^2-x+4)$\\	
		\cline{2-6}	
		&$(x^2-x+1)$ & $(x^2-x+1)$ & $3$ & $3-(x^2-x+1)(x^3-1)$ & $(x^2-2x+2)$\\				
		\hline		
	\end{tabular}	
	\label{tab:genus}
\end{table*}

\subsection{Calculation of $\psi(3;n|\underline{010})$}
Given $1\leq j_1<j_2<j_3\leq n$, we have
\begin{align}
	(2^{j_1r}-2^{j_2r}+2^{j_3r}) = 2^{j_1r}(2^{kr}-2^{ir}+1) \geq 2^r(2^{kr}-2^{(k-1)r}+1),\nonumber
\end{align}
where $k\triangleq(j_3-j_1)>i\triangleq(j_2-j_1)\geq1$. After solving the following inequality
\begin{align}
	(1-2^{-r})2^r(2^{kr}-2^{(k-1)r}+1) = (x-1)(x^k-x^{k-1}+1) < 1,\nonumber
\end{align}
we obtain $k\leq 7$. Therefore, 
\begin{align}
	\psi(3;n|\underline{010}) \approx \sum_{k=2}^{7}\sum_{i=1}^{k-1}\sum_{j=1}^{n-k}{\left(1-x^{j-1}(x-1)(x^k-x^i+1)\right)^+},\nonumber
\end{align}
which includes $\binom{7}{2}=21$ species. However, after a careful calculation, it will be found that there are only $10$ mortal species, as included in Table~\ref{tab:genus}, while the remaining $11$ species are extinct species. For each mortal species, its lifespan $J$ and contribution to $\psi(3;n|\underline{010})$ are also included in Table~\ref{tab:genus}. For $n\geq8$, the accumulated contribution from mortal species converges to
\begin{align}\label{eq:010}
	3(x^2-2x+1) + 2(-x^2+2) + 2(-2x^2+4) + (-x+2) + \nonumber\\(-x^2-x+4) + (x^2-2x+2) = -3x^2-10x+23.
\end{align}

\subsection{Calculation of $\psi(3;n|\underline{110})$}
Given $1\leq j_1<j_2<j_3\leq n$, we have
\begin{align}
	(2^{j_1r}-2^{j_2r}-2^{j_3r}) = 2^{j_1r}(2^{kr}-2^{ir}-1) = 2^{j_1r}(2^{kr}(1-2^{-(k-i)r})-1),\nonumber
\end{align}
where $k\triangleq(j_3-j_1)>i\triangleq(j_2-j_1)\geq1$. Note that $(2^{kr}-2^{ir}-1)$ may be negative or positive, so the analysis is much more complex than that of $(2^{kr}+2^{ir}+1)$, $(2^{kr}+2^{ir}-1)$, and $(2^{kr}-2^{ir}+1)$. For $k>i$, we have $(1-2^{-(k-i)r})>0$, so for any given $k'\triangleq(k-i)$, the function $2^{kr}(1-2^{-k'r})$ is strictly increasing w.r.t. $k$. Moreover, $(1-2^{-k'r})$ is strictly increasing w.r.t. $k'$. Therefore, there exists an integer $K$ such that for any $k>K$,
\begin{align}
	(2^{kr}(1-2^{-(k-i)r})-1) \geq (2^{kr}-2^{(k-1)r}-1)>0,\nonumber
\end{align}
After solving the above inequality, we find $K=5$, \textit{i.e.}, 
\begin{align}
	(2^{kr}-2^{(k-1)r}-1)
	\begin{cases}
		<0, &k<5\\
		=0, &k=5\\
		>0, &k>5
	\end{cases}.\nonumber
\end{align}
Hence, there are at most $\binom{4}{2}=6$ negative species. They are $(x^4-x-1)$, $(x^4-x^2-1)$, $(x^4-x^3-1)$, $(x^3-x-1)$, $(x^3-x^2-1)$, and $(x^2-x-1)$. After a careful calculation, we find $(x^4-x-1)>(x^4-x^2-1)>0$ and $(x^3-x-1)=0$. As such, there are only three negative species. A simple further calculation will verify that all these three negative species are mortal species. 

Continue our analysis. By solving the following inequality
\begin{align}
	(1-2^{-r})2^r(2^{kr}-2^{(k-1)r}-1) = (x-1)(x^k-x^{k-1}-1) < 1,\nonumber
\end{align}
we obtain $k\leq 9$. Therefore, 
\begin{align}
	\psi(3;n|\underline{110}) \approx \sum_{k=2}^{9}\sum_{i=1}^{k-1}\sum_{j_1=1}^{n-k}{\left(1-x^{j_1-1}(x-1)(x^k-x^i-1)\right)^+},\nonumber
\end{align}
which includes $\binom{9}{2}=36$ species. However, after a careful calculation, it will be found that there are only $16$ mortal species (three negative species are counted) and $2$ immortal species, as included in Table~\ref{tab:genus2}, while the remaining $18$ species are extinct species. For each alive species, its lifespan and contribution to $\psi(3;n|\underline{110})$ are also included in Table~\ref{tab:genus2}. For $n\geq(6+8)=14$, where $(6+8)$ corresponds to the species $(x^6-x^5-1)$, whose lifespan is $J=8$, the accumulated contribution from mortal species converges to
\begin{align}\label{eq:110mortal}
	3(x^2-x) + 3(3-x^2) + 3(4-x) + (1-2x^2+2x) + 2(7-x^2) +\nonumber\\ (2-x) + (11-2x^2-x) + (7-2x) + (7-2x^2) = -8x^2-8x+63.
\end{align}
The contribution from the immortal species $(x^3-x-1)$ is $(n-3)$ for $n\geq3$, and the contribution from the immortal species $(x^5-x^4-1)$ is $(n-5)$ for $n\geq5$. They do not converge as $n$ increases.

\begin{table*}[!t]
	\small\centering
	\caption{Mortal and Immortal Species of the $(-1,-1)$-th Genus}
	\begin{tabular}{c||c||c||c||c||c}
		\hline
		Sign &Species & Species $\bmod\alpha$ & $J$ & Contribution & Contribution $\bmod\alpha$ \\
		\hline
		\hline
		\multirow{13}*{Positive} 
		&$(x^9-x^8-1)$ &\multirow{3}*{$(x^2+x-1)$} &\multirow{3}*{$2$} &\multirow{3}*{$2-(x^2+x-1)(x^2-1)$} &\multirow{3}*{$(x^2-x)$}\\		
		&$(x^7-x^5-1)$ & & & \\		
		&$(x^6-x^3-1)$ & & & \\		
		\cline{2-6}		
		&$(x^8-x^7-1)$ &\multirow{3}*{$x$} &\multirow{3}*{$3$} &\multirow{3}*{$3-x(x^3-1)$} &\multirow{3}*{$(-x^2+3)$}\\		
		&$(x^6-x^4-1)$ & & & \\		
		&$(x^5-x^2-1)$ & & & \\		
		\cline{2-6}		
		&$(x^7-x^6-1)$ &\multirow{3}*{$(x^2-1)$} &\multirow{3}*{$5$} &\multirow{3}*{$5-(x^2-1)(x^5-1)$} &\multirow{3}*{$(-x+4)$}\\		
		&$(x^5-x^3-1)$ & & & \\		
		&$(x^4-x-1)$  & & & \\		
		\cline{2-6}		
		&$(x^6-x^5-1)$ &\multirow{2}*{$(x-1)$} &\multirow{2}*{$8$} &\multirow{2}*{$8-(x-1)(x^8-1)$} &\multirow{2}*{$(-x^2+7)$}\\		
		&$(x^4-x^2-1)$ & & & \\		
		\cline{2-6}		
		&$(x^6-x^2-1)$ & $2x$ & $1$ & $1-2x(x-1)$ & $(-2x^2+2x+1)$\\		
		\cline{2-6}		
		&$(x^5-x-1)$ & $x^2$ & $2$ & $2-x^2(x^2-1)$ & $(-x+2)$ \\		
		\hline
		\hline
		\multirow{2}*{Zero} 
		&$(x^5-x^4-1)$ & \multirow{2}*{$0$} & \multirow{2}*{$\infty$} &\multirow{2}*{$\infty$} &\multirow{2}*{$\infty$}\\
		&$(x^3-x-1)$ & & & &\\
		\hline
		\hline
		\multirow{3}*{Negative} 
		&$(x^4-x^3-1)$ & $(x^2-2)$ & $9$ & $9+(x^2-2)(x^9-1)$ & $(-2x^2-x+11)$ \\
		\cline{2-6}		
		&$(x^3-x^2-1)$ & $(-x^2+x)$ & $7$ & $7+(-x^2+x)(x^7-1)$ & $(-2x+7)$ \\
		\cline{2-6}		
		&$(x^2-x-1)$ & $(x^2-x-1)$ & $6$ & $6+(x^2-x-1)(x^6-1)$ & $(-2x^2+7)$ \\
		\hline		
	\end{tabular}
	\label{tab:genus2}
\end{table*}

\subsection{Calculation of $\psi(3;n)$}
Finally, for $n\geq14$, we add \eqref{eq:100}, \eqref{eq:010}, \eqref{eq:110mortal}, and the contributions from immortal species to get
\begin{align}
	\psi(3;n) 
	&= \frac{(-x^2+x+1)+(-3x^2-10x+23)+(-8x^2-8x+63)+(n-3)+(n-5)}{4}\nonumber\\
	&= \frac{-12x^2-17x+79}{4}+n/2.\nonumber
\end{align}

\section{Proof of Theorem~\ref{thm:wn1}}\label{prf:wn1}
If $j^{n-1}=(1,\dots,n-k,n-k+2,\dots,n)\in {\cal J}_{n,n-1}$, where $1\leq k\leq n$, then 
\begin{align}
	2^{-nr}c(j^{n-1})
	&= (1-2^{-r})\sum_{i=1}^{n}{2^{(i-n)r}} - (1-2^{-r})2^{(n-k+1-n)r}\nonumber\\
	&=(2^r-1)\sum_{i=1}^{n}{2^{-ir}} - (2^r-1)2^{-kr} = (1-2^{-nr}) - (2^r-1)2^{-kr}\nonumber
\end{align}
and
\begin{align}
	V(j^{n-1}) 
	&\simeq \underbrace{(2^r-1)\sum_{i=1}^{k-1}{X_i 2^{-ir}}}_{V_1} + \underbrace{(2^r-1)\sum_{i=k+1}^n{X_i 2^{-ir}}}_{V_2}\nonumber\\
	&= V_1+V_2.\nonumber
\end{align}
Let $f_{V_1}(v)$ denote the pdf of $V_1$ and $f_{V_2}(v)$ denote the pdf of $V_2$. Then $f_{V|j^{n-1}}(v)= f_{V_1}(v)*f_{V_2}(v)$, where $*$ denotes the convolution operation. 

\subsection{$k<\infty$}
According to \eqref{eq:U0infty}, the definition of $U_{0,\infty}$, for $k<\infty$, we have $\lim_{n\to\infty}V_2 \simeq 2^{-kr}U_{0,\infty}$. By the property of pdf, there is $\lim_{n\to\infty}f_{V_2}(v) = 2^{kr}f(v2^{kr})$ for $0\leq v<2^{-kr}$, where $f(u)$ is the asymptotic CCS defined by \eqref{eq:asympccs}. The pdf of $(2^r-1){X_i2^{-ir}}$ is $(\delta(v)+\delta(v-(2^r-1)2^{-ir}))/2$. Since $X^n$ is an i.i.d. sequence,
\begin{align}
	f_{V_1}(v) = 2^{-(k-1)}\left(\bigotimes_{i=1}^{k-1}\left(\delta(v)+\delta(v-(2^r-1)2^{-ir})\right)\right)
	= 2^{-(k-1)}\sum_{x^{k-1}\in\mathbb{B}^{k-1}}\delta(v-l(x^{k-1})),\nonumber
\end{align}
where $\otimes$ denotes the convolution operation and $l(X^i)$ is defined by \eqref{eq:lXi}. In turn,
\begin{align}
	\lim_{n\to\infty}f_{V|j^{n-1}}(v)
	&= f_{V_1}(v)*\lim_{n\to\infty}f_{V_2}(v)\nonumber\\
	&= 2^{kr}f(v2^{kr})\ast2^{-(k-1)}\sum_{x^{k-1}\in\mathbb{B}^{k-1}}\delta(v-l(x^{k-1}))\nonumber\\
	&= 2^{kr-(k-1)}\sum_{x^{k-1}\in\mathbb{B}^{k-1}}f((v-l(x^{k-1}))2^{kr}),\nonumber
\end{align}
where $*$ denotes the convolution operation. It is easy to know
\begin{align}
	\lim_{n\to\infty}{2^{-nr}c(j^{n-1})} = \lim_{n\to\infty}{(1-2^{-nr}) - (2^r-1)2^{-kr}} = 1-(2^r-1)2^{-kr}.\nonumber
\end{align}
From the relation between $V(j^d)$ and $W(j^d)$ given by \eqref{eq:wvpdf}, we have
\begin{align}
	\lim_{n\to\infty}f_{W|j^{n-1}}(w)
	= 2^{-k(1-r)}\sum_{x^{k-1}\in\mathbb{B}^{k-1}}f((\tfrac{1-(2^r-1)2^{-kr}-w}{2}-l(x^{k-1}))2^{kr}).\nonumber
\end{align}

\subsection{$k=\infty$}
It is easy to see $\lim_{k\to\infty}V_1 \simeq U_{0,\infty}$ and $\lim_{k\to\infty}V_2=0$. Hence 
\begin{align}
	\lim_{k\to\infty}{f_{V|j^{n-1}}(v)} = \lim_{k\to\infty}{f_{V_1}(v)} = f(v),\nonumber
\end{align}
where $f(u)$ is the asymptotic CCS defined by \eqref{eq:asympccs}. In turn 
\begin{align}
	\lim_{n\to\infty}f_{W|j^{n-1}}(w)=f(\tfrac{1-w}{2})/2.\nonumber
\end{align}


\end{document}